\documentclass[11pt, a4paper]{article}

\usepackage{setspace}
\usepackage{bm}
\usepackage{upgreek}
\usepackage{rotating}
\usepackage{graphicx}
\usepackage{verbatim}
\usepackage{wrapfig}
\usepackage[margin=2.5cm]{geometry}
\usepackage{amsmath}
\usepackage{placeins}
\usepackage{url}

\usepackage{cite}

\usepackage{fancyhdr}
\usepackage{color}
\usepackage{amsthm}
\usepackage{amsmath, amssymb}
\usepackage{mathabx}
\usepackage{longtable}
\usepackage{lscape}
\usepackage{multicol}
\usepackage{pdfpages}
\usepackage{siunitx}
\usepackage{graphicx}
\usepackage{enumerate}

\usepackage[affil-it]{authblk}

\usepackage{caption}
\usepackage{subcaption}
\usepackage{bbm}
\usepackage{nameref}

\usepackage{tikz}
\usetikzlibrary{positioning}
\tikzset{>=stealth}
\usetikzlibrary{decorations.markings}
\usetikzlibrary{patterns,angles,quotes}
\usetikzlibrary{math}
\usetikzlibrary{calc}
\usetikzlibrary{shapes,backgrounds}
\usetikzlibrary{arrows.meta,intersections}

\usepackage{hyperref}

\usepackage{chngcntr}
\counterwithout{table}{section}

\newtheorem{all}{Theorem}[section]

\theoremstyle{plain}
\newtheorem{lemma}[all]{Lemma}
\newtheorem{thm}[all]{Theorem}

\theoremstyle{definition}

\newtheorem{rem}[all]{Remark}
\newtheorem{asptn}[all]{Assumption}

\newcounter{counter}

\newcommand{\BI}{{\mathbb{I}}}

\newcommand{\BN}{{\mathbb{N}}}

\newcommand{\BR}{{\mathbb{R}}}
\newcommand{\BS}{{\mathbb{S}}}

\newcommand{\BZ}{{\mathbb{Z}}}

\newcommand{\FT}{{\mathcal{F}}}

\newcommand{\dd}{{\mathrm{d}}}
\newcommand{\p}{\partial}
\newcommand{\sgn}{\mathrm{sgn}}

\newcommand{\eps}{\epsilon}

\newcommand{\com}[1]{}

\newcommand{\tch}{T_c^{\Omega_1}}
\newcommand{\tcf}{T_c^{\Omega_0}}
\newcommand{\tcti}{T_c^0}
\newcommand{\ati}{a_T^0}
\newcommand{\af}{a_T^{\Omega_0}}
\newcommand{\ah}{a_T^1}

\DeclareMathOperator\artanh{artanh}

\DeclareMathOperator\arcoth{arcoth}
\DeclareMathOperator\sech{sech}
\DeclareMathOperator\Si{Si}
\DeclareMathOperator\Cin{Cin}

\numberwithin{equation}{section}
\newcommand{\nr}{\addtocounter{equation}{1}\tag{\theequation}}

\begin{document}
\title{BCS Critical Temperature on Half-Spaces}
\author[1]{Barbara Roos\thanks{barbara.roos@ist.ac.at}}
\author[1]{Robert Seiringer\thanks{robert.seiringer@ist.ac.at}}
\affil[1]{Institute of Science and Technology Austria, Am Campus 1, 3400 Klosterneuburg, Austria}

\date{\today}

\maketitle

\begin{abstract}
We study the BCS critical temperature on half-spaces in dimensions $d=1,2,3$ with Dirichlet or Neumann boundary conditions.
We prove that the critical temperature on a half-space is strictly higher than on $\mathbb{R}^d$, at least at weak coupling in $d=1,2$ and weak coupling and small chemical potential in $d=3$.
Furthermore, we show that the relative shift in critical temperature vanishes in the weak coupling limit.
\end{abstract}

\section*{Statements and Declarations}
Financial support by the Austrian Science Fund (FWF) through project number I 6427-N (as part of the SFB/TRR 352) 
is gratefully acknowledged. 
The authors have no competing interests to declare that are relevant to the content of this article.
Data sharing not applicable to this article as no datasets were generated or analysed during the current study.

\tableofcontents

\section{Introduction and Results}
We study the effect of a boundary on the critical temperature of a superconductor in Bardeen-Cooper-Schrieffer theory.
It was recently observed \cite{barkman_elevated_2022,benfenati_boundary_2021,samoilenka_boundary_2020,samoilenka_microscopic_2021,talkachov_microscopic_2022} that the presence of a boundary may increase the critical temperature.
For a one-dimensional system with $\delta$-interaction, a rigorous mathematical justification was given in \cite{hainzl_boundary_nodate}.
Here, we generalize this result to generic interactions and higher dimensions.
While in dimensions $d=2,3$ the existing numerical works only consider lattice models, our analytic approach allows us to study continuum models.
We compare the half infinite superconductor with shape $\Omega_1=(0,\infty) \times \BR^{d-1}$ to the superconductor on $\Omega_0=\BR^d$ in dimensions $d=1,2,3$.
We impose either Dirichlet or Neumann boundary conditions, and prove that in the presence of a boundary the critical temperature can increase.
The critical temperature can be determined from the spectrum of the two-body operator
\begin{equation}\label{H_original}
H_T^\Omega= \frac{-\Delta_x-\Delta_y-2\mu}{\tanh\left(\frac{-\Delta_x-\mu}{2T}\right)+\tanh\left(\frac{-\Delta_y-\mu}{2T}\right)}-\lambda V(x-y)
\end{equation}
acting in $L_{\rm sym}^2(\Omega\times \Omega)=\{\psi \in L^2(\Omega\times \Omega) \vert \psi(x,y)=\psi(y,x)\ \mathrm{for\ all}\ x,y\in \Omega\}$ with appropriate boundary conditions \cite{frank_bcs_2019}.
Here, $\Delta$ denotes the Dirichlet or Neumann Laplacian on $\Omega$ and the subscript indicates on which variable it acts.
Furthermore, $T$ denotes the temperature, $\mu$ is the chemical potential, $V$ is the interaction and $\lambda$ is the coupling constant.
The first term in $H_T^\Omega$ is defined through functional calculus.

Let us explain how $H_T^\Omega$ relates to the BCS critical temperature of a superconductor.
A mathematical introduction to BCS theory can be found in \cite{hainzl_bardeencooperschrieffer_2016}.
BCS theory describes the state of the system as the minimizer of the BCS functional $\FT$.
The normal state $\Gamma_0$ is the minimizer of $\FT$ among states which do not exhibit any superconductivity.
If  perturbations of $\Gamma_0$ that introduce pairing between electrons decrease the value of $\FT$, the system is superconducting.
It turns out that the normal state is always a critical point of $\FT$ and therefore the behavior of $\FT$ in the vicinity of $\Gamma_0$ is determined by the Hessian, which is exactly $2 H_T^\Omega$, as explained in \cite{frank_bcs_2019}.
Importantly, the normal state is unstable and the system is superconducting if $\inf \sigma (H_T^\Omega) <0$.
For translation invariant systems, i.e.~$\Omega=\BR^d$, with suitable interactions $V$ superconductivity is equivalent to $\inf \sigma (H_T^\Omega) <0$.
This was shown in \cite{hainzl_bcs_2008,hainzl_bardeencooperschrieffer_2016} in the case without symmetry restriction on the Cooper pair wave function and can be adapted to the case with symmetry restriction, as explained in \cite{roos_linear_2024}.
In this case, there is a unique critical temperature $T_c$ determined by $\inf \sigma(H_{T_c}^\Omega) = 0$ which separates the superconducting and the normal phase.
The critical temperatures $\tcf$ and $\tch$ 
are defined as
\begin{equation}
T_c^{\Omega_j}(\lambda):= \inf\{T\in (0,\infty) \vert \inf \sigma(H_T^{\Omega_j})\geq 0\}.
\end{equation}
In \cite{samoilenka_boundary_2020} an equivalent definition of the critical temperature was used based on the Birman-Schwinger version of $H_T^\Omega$ and the Mittag-Leffler series for $\tanh$.
In Lemma~\ref{H1_esspec} we prove the inequality $\inf \sigma(H_T^{\Omega_1}) \leq \inf \sigma(H_T^{\Omega_0})$.
Therefore, $\tch(\lambda)\geq \tcf(\lambda)$.
Our main concern is to show that the inequality is strict, which means that there is a temperature range for which the system with boundary is superconducting while the system on $\BR^d$ is not.

Our strategy involves proving $\inf \sigma(H_{T}^{\Omega_1})<0$ for $T=\tcf(\lambda)$ using the variational principle.
The idea is to construct a trial state involving the ground state of $H_{T}^{\Omega_0}$ at temperature $T=\tcf(\lambda)$.
However, $H_T^{\Omega_0}$ is translation invariant in the center of mass coordinate and thus has purely essential spectrum.
To obtain a ground state eigenfunction, we remove the translation invariant directions, and instead consider the reduced operator
\begin{equation}\label{def_h0}
H_T^0=\frac{-\Delta-\mu}{\tanh\left(\frac{-\Delta-\mu}{2T}\right)}-\lambda V(r)
\end{equation}
acting in $L^2(\BR^d)$, which corresponds to zero total momentum in $H_T^{\Omega_0}$. 
At weak enough coupling, the infimum of $\sigma(H_{T}^{\Omega_0})$ for $T=\tcf(\lambda)$ is attained at zero total momentum (c.f. Lemma~\ref{ti_reduction} and Remark~\ref{ti_reduction_fix}).
Our trial state involves the ground state of $H_{T}^0$ at temperature $T=\tcf(\lambda)$.
In the weak coupling limit, $\lambda\to 0$, we can compute the asymptotic form of this ground state provided that $\mu>0$ and the operator $\mathcal{V}_\mu:L^2(\BS^{d-1})\to L^2(\BS^{d-1})$ with integral kernel
\begin{equation}
\mathcal{V}_\mu(p,q)=\frac{1}{(2\pi)^{d/2}}\widehat{V}(\sqrt{\mu}(p-q))
\end{equation}
has a non-degenerate eigenvalue $e_\mu = \sup \sigma(\mathcal{V_\mu})>0$ at the top of its spectrum and the corresponding eigenfunction is even \cite{hainzl_bardeencooperschrieffer_2016,henheik_universality_2023}.
Here, $\widehat{V}(p)=\frac{1}{(2\pi)^{d/2}}\int_{\BR^d} V(r) e^{-i p \cdot r} \dd r$ denotes the Fourier transform of $V$.
For $d=1$, $L^2(\BS^0)$ is a two-dimensional vector space, and $\mathcal{V}_\mu$ has the eigenvalues $\frac{\widehat{V}(0)\pm\widehat{V}(2\sqrt{\mu})}{(2\pi)^{1/2}}$, where the plus and minus sign correspond to an even and odd eigenfunction, respectively.

We make the following assumptions on the interaction potential.
\begin{asptn}\label{aspt_V_halfspace}
Let $d\in\{1,2,3\}$ and $\mu>0$.
Assume that
\begin{enumerate}[(i)]
\item $V\in L^1(\BR^d) \cap L^{p_d}(\BR^d)$, where $p_d=1$ for $d=1$, and $p_d>d/2$ for $d\in \{2,3\}$, \label{aspt.1}
\item $V$ is radial, $V\not \equiv 0$,\label{aspt.2}
\item $\vert \cdot \vert V \in L^1(\BR^d)$,\label{aspt.3}
\item $\widehat{V}(0)>0$,\label{aspt.4}
\item $e_\mu=\sup \sigma(\mathcal{V}_\mu)$ is a non-degenerate eigenvalue and the corresponding eigenfunction is even. \label{aspt.5}
\end{enumerate}
\end{asptn}

\begin{rem}\label{rem:asptn.1}
The assumption $V\in L^1(\BR^d)$ implies that $\widehat V$ is continuous and bounded.
The operator $\mathcal{V_\mu}$ is thus Hilbert-Schmidt and in particular compact.
Due to Assumption~\eqref{aspt.5} we have $e_\mu>0$.
This in turn implies that the critical temperature $\tcf(\lambda)$ for the system on $\BR^d$ is positive for all $\lambda>0$ (cf. Remark~\ref{ti_reduction_fix}).
Furthermore, for $d\geq 2$ radiality of $V$ and \eqref{aspt.5} imply that the eigenfunction corresponding to $e_\mu$ must be rotation invariant, i.e.~the constant function.
Assumption~\eqref{aspt.5} is satisfied  for $d=2,3$ if $\widehat{V}\geq 0$ \cite{hainzl_bardeencooperschrieffer_2016} and for $d=1$ if $\widehat{V}(0),\widehat{V}(2\sqrt{\mu})>0$.
\end{rem}

These assumptions suffice to observe boundary superconductivity in $d=1,2$.
For $d=3$, we need one additional condition.
Let
\begin{equation}\label{jd}
j_d(r;\mu):=\frac{1}{(2\pi)^{d/2}}\int_{\BS^{d-1}} e^{i \omega \cdot r  \sqrt{\mu}} \dd \omega.
\end{equation}
Define
\begin{equation}\label{mtilde}
\widetilde m_3^{D/N}(r;\mu):=\int_{\BR} \left(j_3(z_1,r_2,r_3;\mu)^2 - \vert j_3(z_1,r_2,r_3;\mu) \mp j_3(r;\mu) \vert^2 \chi_{|z_1|<|r_1|}\right)\dd z_1 \mp  \frac{\pi}{\mu^{1/2}} j_3(r;\mu)^2,
\end{equation}
where the indices $D$ and $N$ as well as the upper/lower signs correspond to Dirichlet/Neumann boundary conditions, respectively.
Our main result is as follows:
\begin{thm}\label{thm1}
Let $d\in\{1,2,3\}$, $\mu>0$ and let $V$ satisfy Assumption~\ref{aspt_V_halfspace}.
Assume either Dirichlet or Neumann boundary conditions.
For $d=3$ additionally assume that
\begin{equation}\label{3d_cond}
\int_{\BR^3} V(r) \widetilde m_3^{D/N}(r;\mu) \dd r>0.
\end{equation}
Then there is a $\lambda_1>0$, such that for all $0<\lambda<\lambda_1$, $\tch(\lambda)>\tcf(\lambda)$.
\end{thm}
For $d=3$ we prove that \eqref{3d_cond} is satisfied for small enough chemical potential.
\begin{thm}\label{thm2}
Let $d=3$ and let $V$ satisfy~\ref{aspt_V_halfspace}\eqref{aspt.1}-\eqref{aspt.4}.
For Dirichlet boundary conditions, additionally assume that $\vert \cdot \vert^2 V \in L^1(\BR^3)$ and $\int_{\BR^3} V(r) |r|^2 \dd r >0$.
Then there is a $\mu_0>0$ such that for all $0<\mu<\mu_0$, $\int_{\BR^3} V(r) \widetilde m_3^{D/N}(r;\mu) \dd r>0$.
In particular, if $V$ additionally satisfies~\ref{aspt_V_halfspace}\eqref{aspt.5} for small $\mu$ (e.g.~if $\widehat V\geq 0$), then for small $\mu$ there is a $\lambda_1(\mu)>0$ such that  $\tch(\lambda)>\tcf(\lambda)$ for $0<\lambda<\lambda_1(\mu)$.
\end{thm}

\begin{rem}
Numerical evaluation of $\widetilde m_3^D$ suggests that $\widetilde m_3^D\geq 0$ (see Section~\ref{sec:3d_condition}, in particular Figure~\ref{fig:m3}).
Hence, for Dirichlet boundary conditions \eqref{3d_cond} appears to hold under the additional assumption that $V\geq 0$.
We therefore expect that for Dirichlet boundary conditions also in three dimensions boundary superconductivity occurs for all values of $\mu$.
There is no proof so far, however.
\end{rem}

\begin{rem}
One may wonder why in $d=1,2$ no condition like \eqref{3d_cond} is needed.
Actually, in $d=1,2$ the analogous condition is always satisfied if $\widehat{V}(0)>0$.
The reason is that if one defines $\widetilde m_d^{D/N}(r;\mu)$ by replacing $j_3$ by $j_d$ in \eqref{mtilde}, the first term diverges and $\widetilde m_d^{D/N}(r;\mu)=+\infty$.
\end{rem}

Our second main result is that the relative shift in critical temperature vanishes as $\lambda\to0$.
This generalizes the corresponding result for $d=1$ with contact interaction in \cite{hainzl_boundary_nodate}.
\begin{thm}\label{thm3}
Let $d\in\{1,2,3\}$, $\mu>0$ and let $V$ satisfy Assumption~\ref{aspt_V_halfspace} and $V\geq 0$.
Then
\begin{equation}
\lim_{\lambda\to0} \frac{\tch(\lambda)-\tcf(\lambda)}{\tcf(\lambda)}=0.
\end{equation}
\end{thm}
We expect that the additional assumption $V\geq 0$ in Theorem~\ref{thm3} is not  necessary;  it is required in our proof, however.
\begin{rem}
The temperature $\tch(\lambda)$ is the smallest temperature $T$ satisfying $\inf \sigma(H_T^{\Omega_1})=0$.
In principle, there could be other solutions to this equation, defining larger critical temperatures. 
An inspection of our proof shows that it applies equally well to these larger temperatures, i.e.~Theorem~\ref{thm3} also holds if $\tch(\lambda)$ is replaced by any other solution $T$ of  the  equation $\inf \sigma(H_T^{\Omega_1})=0$.
\end{rem}

The rest of the paper is organized as follows.
In Section~\ref{sec:pre} we prove the Lemmas mentioned in the introduction.
In Section~\ref{sec:gs} we use the Birman-Schwinger principle to study the ground state of $H_{T}^0$.
Section~\ref{sec:thm1} contains the proof of Theorem~\ref{thm1}.
Section~\ref{sec:3d_condition} discusses the conditions under which \eqref{3d_cond} holds and in particular contains the proof Theorem~\ref{thm2}.
In Section~\ref{sec:rel_T} we study the relative temperature shift and prove Theorem~\ref{thm3}.
Section~\ref{sec:pf_aux} contains the proof of auxiliary Lemmas from Section~\ref{sec:rel_T}.

\section{Preliminaries}\label{sec:pre}
The following functions will occur frequently
\begin{equation}
K_{T,\mu}(p,q):=\frac{p^2 + q^2 - 2\mu}{\tanh\left(\frac{p^2-\mu}{2T}\right)+\tanh\left(\frac{q^2-\mu}{2T}\right)}
\end{equation}
and
\begin{equation}\label{B_def}
B_{T,\mu}(p,q):=\frac{1}{K_{T,\mu}(p+q,p-q)}.
\end{equation}
We will suppress the subscript $\mu$ and write $K_{T},B_T$ when the $\mu$-dependence is not relevant.
The following estimate \cite[Lemma 2.1]{hainzl_boundary_nodate} will prove useful.
\begin{lemma}\label{KT-Laplace}
For every $T_0>0$ there is a constant $C_1(T_0,\mu)>0$ such that for $T>T_0$, $C_1(T+p^2+q^2)\leq K_{T}(p,q)$.
For every $T>0$ there is a constant $C_2(T,\mu)>0$ such that $K_{T}(p,q)\leq C_2(p^2+q^2+1)$.
\end{lemma}
The minimal value of $K_T$ is $2T$.
Since $|\tanh(x)|<1$, we have for all $p,q\in \BR^d$ and $T\geq 0$
\begin{equation}\label{BT_bound}
B_T(p,q)\leq \frac{1}{\max\{|p^2+q^2-\mu|,2T\}} \quad {\rm and} \quad B_T(p,q)\chi_{p^2+q^2>2\mu>0} \leq \frac{C(\mu)}{1+p^2+q^2},
\end{equation}
where $C(\mu)$ depends only on $\mu$.

\begin{rem}
Assumption~\ref{aspt_V_halfspace}\eqref{aspt.1} guarantees that $V$ is infinitesimally form bounded with respect to $-\Delta_x-\Delta_y$ \cite{lieb_analysis_2001,roos_two-particle_2022}.
By Lemma~\ref{KT-Laplace}, $H_T^\Omega$ defines a self-adjoint operator via the KLMN theorem.
Furthermore, $H_T^\Omega$ becomes positive for $T$ large enough and hence the critical temperatures are finite.
\end{rem}

Let $K_{T}^\Omega$ be the kinetic term in $H_T^\Omega$.
The corresponding quadratic form acts as $\langle \psi, K_{T}^\Omega \psi \rangle = \int_{\Omega^4} \overline{\psi(x,y)}K_{T}^\Omega(x,y;x',y') \psi(x',y') \dd x \dd y\dd x'\dd y'$
where $K_{T}^\Omega(x,y;x',y') $ is the distribution
\begin{equation}\label{K_kernel}
K_{T}^\Omega(x,y;x',y')=\int_{\BR^{2d}}\overline{F_\Omega(x,p) F_\Omega(y,q)}K_{T}(p,q) F_\Omega(x',p) F_\Omega(y',q)\dd p\dd q,
\end{equation}
with
\begin{equation}
F_{\BR^d}(x,p)=\frac{e^{-i p \cdot x}}{(2\pi)^{d/2}} \quad {\rm and} \quad F_{\Omega_1}(x,p)=\frac{(e^{-i p_1 x_1}\mp e^{i p_1 x_1})  e^{-i \tilde p \cdot \tilde x}}{2^{1/2}(2\pi)^{d/2}},
\end{equation}
where the $-/+$ sign corresponds to Dirichlet and Neumann boundary conditions, respectively.
Here, $\tilde x$ denotes the vector containing all but the first component of $x$.
(In the case $d=1$, $\tilde x$ is empty and can be omitted.)

\begin{lemma}\label{H1_esspec}
Let $T,\lambda>0$, $d\in\{1,2,3\}$, and let $V$ satisfy \ref{aspt_V_halfspace}\eqref{aspt.1}.
Then $\inf \sigma(H_T^{\Omega_1})\leq \inf \sigma(H_T^{\Omega_0})$.
\end{lemma}

With the following Lemma we may use $H_T^0$ instead of $H_T^{\Omega_0}$ to compute $\tcf(\lambda)$ at weak enough coupling.
\begin{lemma}\label{ti_reduction}
Let $T,\lambda>0$, $d\in\{1,2,3\}$, and let $V$ satisfy \ref{aspt_V_halfspace}\eqref{aspt.1}.
Let $ \sigma_{\mathrm{s}}(H_T^{0}) $ denote the spectrum of $H_T^0$ restricted to even functions.
Then $ \inf \sigma(H_T^0)\leq \inf \sigma(H_T^{\Omega_0}) \leq \inf \sigma_{\mathrm{s}}(H_T^0)$.
\end{lemma}

\begin{rem}\label{ti_reduction_fix}
Under Assumption~\ref{aspt_V_halfspace}, for all couplings $\lambda>0$ there is a unique $\tcti(\lambda)>0$ satisfying $\inf \sigma(H_{\tcti(\lambda)}^0)=0$ (see \cite[Theorem 3.2]{hainzl_bardeencooperschrieffer_2016} for $d=3$, and \cite[Theorem 2.5]{henheik_universality_2023} for $d=1,2$).
In Section~\ref{sec:gs}, in particular Remark~\ref{gs_symm_lsmall}, we shall show that there is a $\lambda_0>0$ such that the ground state of $H_{\tcti(\lambda)}^0$ is even for couplings $\lambda\leq \lambda_0$.
By Lemma~\ref{ti_reduction}, $\inf \sigma(H_{\tcti(\lambda)}^0)=\inf \sigma(H^{\Omega_0}_{\tcti(\lambda)})=0$.
Furthermore, for $T<\tcti(\lambda)$, due to strict monotonicity of $H_T^0$ in $T$,
\[
\inf \sigma(H_T^{\Omega_0})\leq \inf \sigma_s(H_{T,\lambda}^0)<\inf \sigma(H_{\tcti(\lambda)}^0)=0.
\]
Hence, $\tcf(\lambda)=\tcti(\lambda)$ for $\lambda\leq\lambda_0$.
In particular, the minimum of $\sigma( H_{T}^{\Omega_0})$ for $T=\tcf(\lambda)$ is attained at zero total momentum.
\end{rem}

\begin{rem}
The essential spectrum of $H_T^0$ satisfies $\inf \sigma_{\rm ess}(H_T^0)=2T$ (see e.g.~\cite[Proof of Thm 3.7]{lauritsen_masters_2020}).
Hence, zero is an eigenvalue of $H_{\tcti(\lambda)}^0$.
\end{rem}

\subsection{Proof of Lemma~\ref{H1_esspec}}
\begin{proof}[Proof of Lemma~\ref{H1_esspec}]
Let $S_l$ be the shift to the right by $l$ in the first component, i.e. $S_l \psi (x,y)= \psi(((x_1-l),\tilde x), (y_1-l,\tilde y))$.
Let $\psi$ be a compactly supported function in $H_{\rm sym}^1(\BR^{2d})$, the Sobolev space restricted to functions satisfying $\psi(x,y)=\psi(y,x)$.
For $l$ large enough, $S_l \psi$ is supported on the half-space and satisfies both Dirichlet and Neumann boundary conditions.
The goal is to prove that $\lim_{l\to \infty} \langle S_l \psi, H^{\Omega_1}_T S_l \psi \rangle = \langle \psi, H^{\Omega_0}_T \psi \rangle$.
Then, since compactly supported functions are dense in $ H_{\rm sym}^1(\BR^{2d})$, the claim follows.

Note that $\langle S_l \psi,V S_l \psi \rangle = \langle \psi,V \psi \rangle$.
Furthermore,
using symmetry of $K_T$ in $p_1$ and $q_1$ one obtains
\begin{multline}\label{2.6}
\langle S_l \psi,K_T^{\Omega_1} S_l \psi \rangle
=\int_{\BR^{2d}}\overline{\widehat \psi(p,q)}K_T(p,q)\Big[\widehat\psi(p,q)\mp \widehat\psi((-p_1,\tilde p),q) e^{2 i l p_1}\mp \widehat\psi(p,(-q_1,\tilde q))e^{2 i l q_1}\\
+\widehat\psi((-p_1,\tilde p),(-q_1,\tilde q))e^{2 i l (p_1+q_1)}\Big] \dd p\dd q
\end{multline}
for $l$ large enough such that $\psi$ is supported on the half-space.
The first term is exactly $\langle \psi,K_T^{\Omega_0} \psi \rangle $.
Note that by the Schwarz inequality and Lemma~\ref{KT-Laplace}, the function
\begin{equation}
(p,q)\mapsto \overline{\widehat \psi(p,q)}K_T(p,q)\widehat\psi((-p_1,\tilde p),q)
\end{equation}
is in $L^1(\BR^{2d})$ since $\psi \in H^1(\BR^{2d})$.
By the Riemann-Lebesgue Lemma, the second term in \eqref{2.6} vanishes for $l\to \infty$.
By the same argument, also the remaining terms vanish in the limit.
\end{proof}

\subsection{Proof of Lemma~\ref{ti_reduction}}
First, we prove the following inequality.
\begin{lemma}\label{ineq_p0_inf}
For all $x,y\in \BR$ we have
\begin{equation}\label{eq_p0_inf}
\frac{x+y}{\tanh(x)+\tanh(y)}\geq \frac{1}{2}\left(\frac{x}{\tanh(x)}+\frac{y}{\tanh(y)}\right)
\end{equation}
\end{lemma}
\begin{proof}[Proof of Lemma~\ref{ineq_p0_inf}]
Suppose $\vert x \vert \neq \vert y \vert$.
Without loss of generality we may assume $x>\vert y \vert$.
Since $\frac{x}{\tanh x}\geq \frac{y}{\tanh y}$,
\begin{equation}
\frac{x}{2 \tanh x} \frac{\tanh x-\tanh y}{\tanh x + \tanh y}\geq \frac{y}{2 \tanh y}\frac{\tanh x-\tanh y}{\tanh x + \tanh y}
\end{equation}
This inequality is equivalent to \eqref{eq_p0_inf}, as can be seen using $\frac{\tanh x-\tanh y}{\tanh x + \tanh y}=\frac{2\tanh x}{\tanh x + \tanh y}-1=1-\frac{2\tanh y}{\tanh x + \tanh y}$ on the left and right side, respectively.
By continuity,  \eqref{eq_p0_inf} also holds in the case $\vert x \vert = \vert y \vert$.
\end{proof}

\begin{proof}[Proof of Lemma~\ref{ti_reduction}]
Let $U$ denote the unitary transform $U\psi(r,z)=\frac{1}{2^{d/2}}\psi((r+z)/2,(z-r)/2)$ for $\psi \in L^2(\BR^{2d})$.
By Lemma~\ref{ineq_p0_inf} we have
\begin{multline}
UH_T^{\Omega_0}U^\dagger= \frac{-\left( \nabla_r+\nabla_z\right)^2-\left(\nabla_r-\nabla_z\right)^2-2\mu}{\tanh\left(\frac{-\left( \nabla_r+\nabla_z\right)^2-\mu}{2T}\right)+\tanh\left(\frac{-\left( \nabla_r-\nabla_z\right)^2-\mu}{2T}\right)}+V(r)\\
\geq \frac{1}{2}\left(\frac{-\left( \nabla_r+\nabla_z\right)^2-\mu}{\tanh\left(\frac{-\left( \nabla_r+\nabla_z\right)^2-\mu}{2T}\right)}+V(r)\right)+ \frac{1}{2}\left(\frac{-\left( \nabla_r-\nabla_z\right)^2-\mu}{\tanh\left(\frac{-\left( \nabla_r-\nabla_z\right)^2-\mu}{2T}\right)}+V(r)\right)
\end{multline}
Both summands are unitarily equivalent to $\frac{1}{2}H_T^0\otimes \BI$, where $\BI$ acts on $L^2(\BR^d)$.
Therefore, $\inf \sigma(H_T^{\Omega_0})\geq \inf\sigma(H_T^0)$.

For the second inequality let $f\in H^1(\BR^d)$ with $f(r)=f(-r)$ and $\psi_{\epsilon}(r,z)=e^{-\epsilon \sum_{j=1}^d\vert z_j \vert} f(r)$.
Note that $\lVert \psi_\eps\rVert^2_2=\frac{1}{\eps^d}\lVert f \rVert_2^2$.
Since the Fourier transform of $e^{-\epsilon \vert r_1 \vert}$ in $L^2(\BR)$ is  $(2/\pi)^{1/2}  {\epsilon}/({\epsilon^2+p_1^2})$, 
we have $\widehat{\psi_\eps}(p,q)= \widehat{f}(q) (2/\pi)^{d/2}  \prod_{j=1}^d {\epsilon}/{(\epsilon^2+p_j^2)}$.
Therefore,
\begin{multline}
\frac{\langle \psi_\epsilon\vert UH_T^{\Omega_0}U^\dagger \psi_\epsilon \rangle}{\lVert \psi_\epsilon\rVert^2}\\=\frac{2^d}{\pi^d \lVert f\rVert^2}\int_{\BR^{2d}} K_{T}(p+q,p-q)\prod_{j=1}^d \frac{\epsilon^3}{(\epsilon^2+p_j^2)^2} |\widehat{f}(q)|^2  \dd p \dd q
+\frac{1}{\lVert f \rVert^2} \int_{\BR^d}V(r) |f(r)|^2\dd r\\
=\frac{2^d}{\pi^d \lVert f\rVert^2}\int_{\BR^{2d}} K_{T}(\eps p+q,\eps p-q)\left(\prod_{j=1}^d \frac{1}{(1+p_j^2)^2}\right) |\widehat{f}(q)|^2  \dd p \dd q+\frac{1}{\lVert f \rVert^2} \int_{\BR^d}V(r) |f(r)|^2\dd r,
\end{multline}
where we substituted $p\to \eps p$ in the second step.
By Lemma~\ref{KT-Laplace},
\begin{equation}
K_{T}(\eps p+q,\eps p-q)\left(\prod_{j=1}^d \frac{1}{(1+p_j^2)^2}\right) |\widehat{f}(q)|^2 \leq C (1+d \eps^2+q^2)\left(\prod_{j=1}^d \frac{1}{1+p_j^2}\right) |\widehat{f}(q)|^2 ,
\end{equation}
which is integrable.
With $\int_{\BR}\frac{1}{(1+p_j^2)^2}  \dd p_j=\pi/2$ it follows by dominated convergence that
\begin{equation}
\lim_{\epsilon \to 0} \frac{\langle \psi_\epsilon\vert UH_T^{\Omega_0}U^\dagger \psi_\epsilon \rangle}{\lVert \psi_\epsilon\rVert^2}=\frac{\langle f \vert H_T^0 f \rangle}{\lVert f\rVert^2}.
\end{equation}
Therefore, $\inf \sigma(H_T^{\Omega_0})\leq \inf\sigma_s(H_T^0)$.
\end{proof}

\section{Ground State of $H_{\tcti(\lambda)}^0$}\label{sec:gs}
Let $\tcti(\lambda)$ be the unique temperature satisfying $\inf \sigma(H_{\tcti(\lambda)}^0)=0$ as in Remark~\ref{ti_reduction_fix}, where $H_T^0$ was defined in \eqref{def_h0}.
To study the ground state of $H_{\tcti(\lambda)}^0$, it is convenient to apply the Birman-Schwinger principle.
For $q\in \BR^d$ let $B_{T}(\cdot,q)$ denote the operator on $L^2(\BR^d)$ which acts as multiplication by $B_{T}(p,q)$ (defined in \eqref{B_def}) in momentum space.
The Birman-Schwinger operator corresponding to $H_{T}^0$ acts on $L^2(\BR^d)$ and is given by
\begin{equation}\label{BS-ti}
A_{T}^0=V^{1/2} B_{T}(\cdot,0) \vert V\vert^{1/2},
\end{equation}
where we use the notation $V^{1/2}(x)=\sgn (V(x))\vert V\vert^{1/2}(x)$.
This operator is compact \cite{hainzl_bardeencooperschrieffer_2016,henheik_universality_2023}.
It follows from the Birman-Schwinger principle that $\sup \sigma(A_{T}^0)=1/\lambda$ exactly for $T=\tcti(\lambda)$ and that the eigenvalue 0 of $H_{\tcti(\lambda)}^0$ has the same multiplicity as the largest eigenvalue of $ A_{\tcti(\lambda)}^0$.

Let $\FT:L^1(\BR^d)\to L^2(\BS^{d-1})$ act as $\FT \psi(\omega)=\widehat \psi (\sqrt{\mu} \omega)$ and define $O_{\mu}=V^{1/2} \FT^\dagger \FT \vert V\vert^{1/2}$ on $L^2(\BR^d)$.
Furthermore, let
\begin{equation}\label{mmu}
m_\mu(T)=\int_0^{\sqrt{2\mu}} B_{T}(t,0)t^{d-1} \dd t.
\end{equation}
Note that $m_\mu(T)= \mu^{d/2-1}\left( \ln\left(\mu/T\right)+c_d\right)+o(1)$ for $T\to0$, where $c_d$ is a number depending only on $d$ \cite[Prop 3.1]{henheik_universality_2023}.

The operator $O_\mu$ captures the singularity of $A_{T}^0$ as $T\to 0$.
The following has been proved in \cite[Lemma 2]{frank_critical_2007} for $d=3$ and in  \cite[Lemma 3.4]{henheik_universality_2023} for $d=1,2$.
\begin{lemma}\label{lea:A0_asy_d}
Let $d\in\{1,2,3\}$ and $\mu>0$ and let $V$ satisfy Assumption~\ref{aspt_V_halfspace}.
Then,
\begin{equation}\label{lea4.3}
\sup_{T\in(0,\infty)} \left \lVert A_{T}^0 -m_\mu(T) O_{\mu}\right \rVert_{\rm HS} <\infty,
\end{equation}
where $\lVert \cdot \rVert_{\rm HS}$ denotes the Hilbert-Schmidt norm.
\end{lemma}

Thus, the asymptotic behavior of $\sup \sigma(A_{T}^0)$ depends on the largest eigenvalue of $O_\mu$.
Note that $O_\mu$ is isospectral to $\mathcal{V}_\mu=\FT V \FT^\dagger$, since both operators are compact.
The eigenfunction of $O_\mu$ corresponding to the eigenvalue $e_\mu$ is
\begin{equation}
\Psi(r):= V^{1/2}(r) j_d(r;\mu),
\end{equation}
where $j_d$ was defined in \eqref{jd}.
Note that
\begin{equation}
j_1(r;\mu)=\sqrt{\frac{2}{\pi}}\cos(\sqrt{\mu} r), \quad
j_2(r;\mu)= J_0(\sqrt{\mu} \vert r \vert), \quad
j_3(r;\mu)=\frac{2}{(2\pi)^{1/2}} \frac{\sin \sqrt{\mu} \vert r \vert}{\sqrt{\mu} \vert r \vert},
\end{equation}
where $J_0$ is the Bessel function of order 0.
Furthermore
\begin{equation}\label{emu}
e_\mu=\frac{1}{(2\pi)^{d/2}} \int_{\BS^{d-1}} \widehat{V}(\sqrt{\mu}((1,0,...,0)-p))\dd p = \frac{1}{\vert \BS^{d-1}\vert}\int_{\BR^d} V(r)j_d(r;\mu)^2 \dd r
\end{equation}

The following asymptotics of $\tcti(\lambda)$ for $\lambda\to 0$ was computed in \cite[Theorem 3.3]{hainzl_bardeencooperschrieffer_2016} and \cite[Theorem 2.5]{henheik_universality_2023}.
\begin{lemma}\label{lea:Tc0_asy}
Let $\mu>0$, $d\in\{1,2,3\}$ and let $V$ satisfy Assumption~\ref{aspt_V_halfspace}. Then
\begin{equation}
\lim_{\lambda \to 0}  \Bigg\vert e_\mu m_\mu(\tcti(\lambda))-\frac{1}{\lambda}\Bigg\vert=\lim_{\lambda \to 0}  \Bigg\vert e_\mu \mu^{d/2-1}\ln \left(\frac{\mu}{\tcti(\lambda)}\right)-\frac{1}{\lambda}\Bigg\vert<\infty.
\end{equation}
\end{lemma}

Lemma~\ref{lea:A0_asy_d} does not only contain information about eigenvalues, but also about the corresponding eigenfunctions.
In the following we prove that the eigenstate corresponding to the maximal eigenvalue of $A_{T}^0$ converges to $\Psi$.
\begin{lemma}\label{psiT_asymptotic_d}
Let $\mu>0$, $d\in\{1,2,3\}$ and let $V$ satisfy Assumption~\ref{aspt_V_halfspace}.
\begin{enumerate}[(i)]
\item There is a $\lambda_0>0$ such that for $\lambda\leq \lambda_0$, the largest eigenvalue of $A_{\tcti(\lambda)}^0$ is non-degenerate.
\item Let $\lambda\leq \lambda_0$ and let $\Psi_{\tcti(\lambda)}$ be the eigenvector of $A_{\tcti(\lambda)}^0$ corresponding to the largest eigenvalue, normalized such that $\lVert \Psi_{\tcti(\lambda)}\rVert_2=\lVert \Psi \rVert_2$.
Pick the phase of $\Psi_{\tcti(\lambda)}$ such that $\langle \Psi_{\tcti(\lambda)},\Psi \rangle \geq 0$.
Then
\begin{equation}
\lim_{\lambda\to 0} \frac{1}{\lambda} \lVert \Psi-\Psi_{\tcti(\lambda)}\rVert_2^2<\infty
\end{equation}
\end{enumerate}
\end{lemma}

\begin{rem}\label{gs_symm_lsmall}
Let $\lambda_0$ be as in Lemma~\ref{psiT_asymptotic_d}. 
By the Birman-Schwinger principle, the multiplicity of the largest eigenvalue of $A_{\tcti(\lambda)}^0$ equals the multiplicity of the ground state of $H^0_{\tcti(\lambda)}$.
Hence, $H^0_{\tcti(\lambda)}$ has a unique ground state for $\lambda\leq \lambda_0$.
For $d\geq 2$, since $H^0_{\tcti(\lambda)}$ is rotation invariant, uniqueness of the ground state implies that the ground state is radial.
For $d=1$, since $\Psi$ is even, the second part of Lemma~\ref{psiT_asymptotic_d} implies that $\Psi_{\tcti(\lambda)}$ is even for small enough $\lambda$.
Hence, also the ground state of $H^0_{\tcti(\lambda)}$ is even for small $\lambda$.
\end{rem}
It follows that for $\lambda\leq \lambda_0$ we have $\tcf(\lambda)=\tcti(\lambda)$ as discussed in Remark~\ref{ti_reduction_fix}.

For values of $\lambda$ such that the operator $H_{\tcti(\lambda)}^0$ has a non-degenerate eigenvalue at the bottom of its spectrum let $\Phi_{\lambda}$ be the corresponding eigenfunction, with normalization and phase chosen such that $\Psi_{\tcti(\lambda)}=V^{1/2}\Phi_{\lambda}$.
The following Lemma with regularity and convergence properties of $\Phi_{\lambda}$ will be useful.
\begin{lemma}\label{phi_infinity_bd}
Let $d\in \{1,2,3\}$, $\mu>0$ and let $V$ satisfy Assumption~\ref{aspt_V_halfspace}.
For all $0<\lambda<\infty$ such that $H_{\tcti(\lambda)}^0$ has a non-degenerate ground state $ \Phi_\lambda $, we have
\begin{enumerate}[(i)]
\item $ \vert \widehat{\Phi_\lambda}(p)\vert\leq \frac{C(\lambda)}{1+p^2} \vert \widehat{V \Phi}_\lambda (p)\vert\leq \frac{C(\lambda) \lVert V\rVert_1^{1/2} \lVert\Psi\rVert_2 }{1+p^2} $ for some number $C(\lambda)$ depending on $\lambda$,\label{phi_infty.1}
\item $p \mapsto \widehat\Phi_\lambda(p)$ is continuous,\label{phi_infty.2}
\item $\lVert \widehat \Phi_\lambda \rVert_1<\infty$ and $\lVert \Phi_\lambda \rVert_\infty<\infty$.\label{phi_infty.3}
\setcounter{counter}{\value{enumi}}
\end{enumerate}
Furthermore, in the limit $\lambda\to 0$
\begin{enumerate}[(i)]
\setcounter{enumi}{\value{counter}}
\item  $\lVert \widehat \Phi_\lambda \chi_{p^2>2\mu}\rVert_1=O(\lambda)$,\label{phi_infty.4}
\item $\lVert \widehat \Phi_\lambda \rVert_1=O(1)$,\label{phi_infty.5}
\item and in particular $\lVert \Phi_\lambda \rVert_\infty=O(1)$.\label{phi_infty.6}
\end{enumerate}
\end{lemma}
In three dimensions, because of the additional condition \eqref{3d_cond}, we need to compute the limit of $\Phi_\lambda$.
\begin{lemma}\label{phi_infinity_conv}
 Let $d=3$, $\mu>0$ and let $V$ satisfy Assumption~\ref{aspt_V_halfspace}.  Then $\lVert \Phi_\lambda -j_3 \rVert_\infty=O(\lambda^{1/2})$ as $\lambda\to0$.
\end{lemma}

\subsection{Proof of Lemma~\ref{psiT_asymptotic_d}}
\begin{proof}[Proof of Lemma~\ref{psiT_asymptotic_d}]
{\it(i)}
The proof uses ideas from \cite[Proof of Thm 1]{hainzl_critical_2008}.
Let $M_T=B_T(\cdot,0)-m_\mu(T) \FT^\dagger \FT$.
By Lemma~\ref{lea:A0_asy_d}, for $\lambda$ small enough the operator $1-\lambda V^{1/2} M_T \vert V\vert^{1/2}$ is invertible for all $T$.
Then we can write
\begin{equation}
1-\lambda A_{T}^0 = (1-\lambda V^{1/2} M_T \vert V\vert^{1/2})\Bigg(1-\frac{\lambda m_\mu(T) }{1-\lambda V^{1/2} M_T \vert V\vert^{1/2} }V^{1/2}  \FT^\dagger \FT \vert V\vert^{1/2} \Bigg)
\end{equation}
Recall that the largest eigenvalue of $A_{\tcti(\lambda)}^0$ equals $1/\lambda$.
Hence, 1 is an eigenvalue of
\begin{equation}\label{pf_psiT_asy.1}
\frac{\lambda m_\mu(\tcti(\lambda))  }{1-\lambda V^{1/2} M_{\tcti(\lambda)} \vert V\vert^{1/2} }V^{1/2}  \FT^\dagger \FT \vert V\vert^{1/2}
\end{equation}
and it has the same multiplicity as the eigenvalue $1/\lambda$ of $A_{\tcti(\lambda)}^0$.
This operator is isospectral to the self-adjoint operator
\begin{equation}\label{pf_psiT_asy.2}
\FT \vert V\vert^{1/2}\frac{\lambda m_\mu({\tcti(\lambda)})  }{1-\lambda V^{1/2} M_{\tcti(\lambda)} \vert V\vert^{1/2} }V^{1/2}  \FT^\dagger .
\end{equation}
Note that the operator difference
\begin{equation}\label{pf_psiT_asy.3}
\FT \vert V\vert^{1/2}\frac{1 }{1-\lambda V^{1/2} M_{\tcti(\lambda)} \vert V\vert^{1/2} }V^{1/2}  \FT^\dagger - \mathcal{V}_\mu=\lambda  \FT \vert V\vert^{1/2}\frac{ V^{1/2} M_{\tcti(\lambda)} \vert V\vert^{1/2}}{1-\lambda V^{1/2} M_{\tcti(\lambda)} \vert V\vert^{1/2} }V^{1/2}  \FT^\dagger
\end{equation}
has operator norm of order $O(\lambda)$ according to Lemma~\ref{lea:A0_asy_d}. 
By assumption, the largest eigenvalue of $\mathcal{V}_\mu$ has multiplicity one, and $\lambda m_\mu({\tcti(\lambda)}) e_\mu =1+O(\lambda)$ by Lemma~\ref{lea:Tc0_asy}.
Let $\alpha<1$ be the ratio between the second largest and the largest eigenvalue of $\mathcal{V}_\mu$.
The second largest eigenvalue of $\lambda m_\mu({\tcti(\lambda)}) \mathcal{V}_\mu$ is of order $\alpha+O(\lambda)$.
Therefore, the largest eigenvalue of \eqref{pf_psiT_asy.2} must have multiplicity 1 for small enough $\lambda$, and it is of order $1+O(\lambda)$, whereas the rest of the spectrum lies below $\alpha+O(\lambda)$.
Hence, 1 is the maximal eigenvalue of \eqref{pf_psiT_asy.2} and it has multiplicity 1 for small enough $\lambda$.

{\it(ii)}
Note that $\Psi_{\tcti(\lambda)}$ is an eigenvector of \eqref{pf_psiT_asy.1} with eigenvalue 1.
Furthermore, let $\psi_{\lambda}$ be a normalized eigenvector of \eqref{pf_psiT_asy.2} with eigenvalue 1.
Then
\begin{equation}
\tilde \Psi_{\tcti(\lambda)}=\frac{\lVert \Psi \rVert_2}{\lVert \frac{1 }{(1-\lambda V^{1/2} M_{\tcti(\lambda)} \vert V\vert^{1/2}) }V^{1/2}  \FT^\dagger \psi_{\lambda}\rVert_2}\frac{1 }{1-\lambda V^{1/2} M_{\tcti(\lambda)} \vert V\vert^{1/2} }V^{1/2}  \FT^\dagger \psi_{\lambda}
\end{equation}
agrees with $\Psi_{\tcti(\lambda)}$ up to a constant phase.
Since $\lVert \Psi_{\tcti(\lambda)}-\Psi\rVert^2 \leq \lVert \tilde \Psi_{\tcti(\lambda)}-\Psi\rVert^2$, it suffices to prove that the latter is of order $O(\lambda)$ for a suitable choice of phase for $\psi_{\lambda}$.

Let $\psi(p)=\frac{1}{\vert \BS^{d-1}\vert^{1/2}}$.
This is the eigenfunction of $\mathcal{V}_\mu$ corresponding to the maximal eigenvalue, and $\Psi=V^{1/2}  \FT^\dagger \psi$.
In particular, for all $\phi \in L^2(\BS^{d-1})$,
\begin{equation}\label{pf_psiT_asy.5}
\langle \phi, \mathcal{V}_\mu \phi \rangle \leq e_\mu |\langle \phi,  \psi \rangle|^2+\alpha e_\mu (\lVert \phi \rVert_2^2-|\langle \phi,  \psi \rangle|^2)
\end{equation}
We choose the phase of $\psi_{\lambda}$ such that $\langle \psi_{\lambda}, \psi \rangle\geq 0$.
We shall prove that $\lVert \psi_{\lambda}-\psi\rVert_2^2=O(\lambda)$.
We have by \eqref{pf_psiT_asy.3} and \eqref{pf_psiT_asy.5}
\begin{multline}
O(\lambda)=\langle \psi_{\lambda},(1-\lambda m_\mu({\tcti(\lambda)}) \mathcal{V}_\mu)\psi_{\lambda}\rangle \\
\geq 1-\lambda m_\mu({\tcti(\lambda)}) e_\mu \vert \langle \psi_\lambda,\psi\rangle\vert^2-\lambda m_\mu({\tcti(\lambda)})\alpha  e_\mu (1-\vert \langle \psi_\lambda,\psi\rangle\vert^2 )\\
=O(\lambda)+(1-\alpha) (1-\vert \langle \psi_\lambda,\psi\rangle\vert^2 )
\end{multline}
where we used Lemma~\ref{lea:Tc0_asy} for the last equality.
In particular, $1-\vert \langle \psi_\lambda,\psi\rangle\vert^2=O(\lambda)$.
Hence,
\begin{equation}\label{pf_psiT_asy.4}
\lVert \psi-\psi_{\lambda} \rVert_2^2 =2 (1-\langle \psi_{\lambda}, \psi \rangle)= 2 \frac{1-\langle \psi_{\lambda}, \psi \rangle^2}{1+\langle \psi_{\lambda}, \psi \rangle} =O(\lambda).
\end{equation}
Using Lemma~\ref{lea:A0_asy_d} and that $V^{1/2}  \FT^\dagger: L^2(\BS^{d-1})\to L^2(\BR^d)$ is a bounded operator, and subsequently \eqref{pf_psiT_asy.4} we obtain
\begin{equation}
\frac{1 }{1-\lambda V^{1/2} M_{\tcti(\lambda)} \vert V\vert^{1/2} }V^{1/2}  \FT^\dagger \psi_\lambda=V^{1/2}  \FT^\dagger \psi_{\lambda}+O(\lambda)=V^{1/2}  \FT^\dagger \psi+O(\lambda^{1/2}),
\end{equation}
where $O(\lambda)$ here denotes a vector with $L^2$-norm of order $O(\lambda)$.
Furthermore,
\begin{multline}
\left \vert \lVert (1-\lambda V^{1/2} M_{\tcti(\lambda)} \vert V\vert^{1/2})^{-1} V^{1/2}  \FT^\dagger \psi_{\lambda}\rVert_2-\lVert V^{1/2}  \FT^\dagger \psi\rVert_2 \right \vert \\
\leq \lVert (1-\lambda V^{1/2} M_{\tcti(\lambda)} \vert V\vert^{1/2})^{-1}V^{1/2}  \FT^\dagger\psi_{\lambda}-V^{1/2}  \FT^\dagger \psi \rVert_2 =O(\lambda^{1/2}).
\end{multline}
In total, we have
\begin{multline}
\tilde \Psi_{\tcti(\lambda)}= \frac{\lVert \Psi \rVert_2}{\lVert V^{1/2}  \FT^\dagger \psi\rVert_2+O(\lambda^{1/2})}(V^{1/2}  \FT^\dagger \psi+O(\lambda^{1/2}))=\frac{\lVert \Psi \rVert_2}{\lVert V^{1/2}  \FT^\dagger \psi\rVert_2}V^{1/2}  \FT^\dagger \psi+O(\lambda^{1/2})\\
=\Psi+O(\lambda^{1/2})
\end{multline}
\end{proof}

\subsection{Regularity and convergence of $\Phi_\lambda$}
In this section, we prove Lemma~\ref{phi_infinity_bd} and Lemma~\ref{phi_infinity_conv}.
The following standard results (see e.g. \cite[Sections 11.3, 5.1]{lieb_analysis_2001}) will be helpful.
\begin{lemma}\label{lp_prop}
\begin{enumerate}[(i)]
\item Let $V \in L^{p}(\BR^d)$, where $p=1$ for $d=1$, $p>1$ for $d=2$ and $p=3/2$ for $d=3$. Let $\psi \in H^1(\BR^d)$. Then $V^{1/2} \psi \in L^2(\BR^d)$. \label{lp_prop.1}
\item If $V\in L^1(\BR^d)$ and $\psi \in L^2(\BR^d)$, then $V^{1/2} \psi \in L^1(\BR^d)$ and hence $\widehat{V^{1/2} \psi}$ is continuous and bounded. \label{lp_prop.2} 
\item For $1\leq t$, $\lVert \widehat{V^{1/2} \psi} \rVert_s \leq C \lVert V \rVert_t^{1/2} \lVert \psi \rVert_2$, where $s=2t/(t-1)$ and $C$ is some constant independent of $\psi$ and $V$. \label{lp_prop.3}
\item Let $f$ be a radial, measurable function on $\BR^3$ and $p\geq 1$. Then there is a constant $C$ independent of $f$ such that $\sup_{p_1 \in \BR} \lVert f(p_1, \cdot) \lVert_{L^p( \BR^{2})} =  \lVert f(0, \cdot) \lVert_{L^p( \BR^{2})}\leq C(\lVert f \rVert_{L^p(\BR^3)}^p+\lVert f \rVert_{L^\infty(\BR^3)}^p)^{1/p}$.\label{lp_prop.4}
\end{enumerate}
\end{lemma}
\begin{proof}
For \eqref{lp_prop.1} and \eqref{lp_prop.2} see e.g.~\cite[Sections 11.3, 5.1]{lieb_analysis_2001}.
For \eqref{lp_prop.3} let $s\geq 2$.
Applying the Hausdorff-Young and H\"older inequality gives
\begin{equation}
\lVert \widehat{V^{1/2} \psi} \rVert_s \leq C \lVert V^{1/2} \psi \rVert_p \leq C \lVert V \rVert_t^{1/2} \lVert \psi \rVert_2,
\end{equation}
where $1=1/p+1/s$ and $1=p/2t+p/2$.
Hence, $s=2t/(t-1)$.

For \eqref{lp_prop.4} we write
\begin{multline}
\lVert f(p_1, \cdot) \lVert_{L^p( \BR^{2})}^p = 2\pi \int_0^\infty |f(\sqrt{p_1^2+t^2})|^p t \dd t =2\pi \int_{|p_1|}^\infty |f(s)|^p s\dd s  \leq \lVert f(0, \cdot) \lVert_{L^p( \BR^{2})}^p \\
\leq 2\pi \int_{0}^1 |f(s)|^p \dd s  +2\pi \int_{0}^\infty |f(s)|^p s^2 \dd s   \leq 2\pi \lVert f \rVert_\infty^p+\frac{1}{2}\lVert f \rVert^p_p,
\end{multline}
where in the second step we substituted $s=\sqrt{p_1^2+t^2}$ and in the third step we used $s\leq \max\{1,s^2\}$.
\end{proof}

\begin{proof}[Proof of Lemma~\ref{phi_infinity_bd}]
The eigenvalue equation $H_{\tcti(\lambda)}^0 \Phi_\lambda =0$ implies that
\begin{equation}\label{eval_eq}
\widehat \Phi_{\lambda}(p)=\lambda B_{\tcti(\lambda)}(p,0)\widehat{V\Phi_\lambda}(p).
\end{equation}
Part \eqref{phi_infty.1} follows with Lemma~\ref{KT-Laplace} and \ref{lp_prop}\eqref{lp_prop.3} and the normalization $\lVert V^{1/2}\Phi_\lambda\rVert_2=\lVert \Psi\rVert_2$.
For part \eqref{phi_infty.2}, note that $p \mapsto B_{T}(p,0)$ is continuous for $T>0$.
Since $\Phi_\lambda \in H^1(\BR^d)$, continuity of $\widehat{V\Phi_\lambda}$ follows by Lemma~\ref{lp_prop}\eqref{lp_prop.1} and \eqref{lp_prop.2}.

Note that $\lVert \Phi_\lambda \rVert_\infty \leq (2\pi)^{-d/2} \lVert \widehat{\Phi_\lambda} \rVert_1 =(2\pi)^{-d/2} (\lVert \widehat{\Phi_\lambda} \chi_{p^2<2\mu}\rVert_1 + \lVert \widehat{\Phi_\lambda} \chi_{p^2>2\mu}\rVert_1 ) $.
In particular, the second part of \eqref{phi_infty.3} and \eqref{phi_infty.6} follow from the first part of \eqref{phi_infty.3} and \eqref{phi_infty.5}, respectively.
Using \eqref{eval_eq} and $ \lVert \Psi_{\tcti(\lambda)}\rVert_2= \lVert \Psi\rVert_2$ we obtain
\begin{equation}\label{pf_wl_4}
\lVert \widehat{\Phi_\lambda} \chi_{p^2<2\mu}\rVert_1 \leq \lambda m_\mu(\tcti(\lambda))|\BS^{d-1}| \lVert \widehat{V^{1/2} \Psi_{\tcti(\lambda)}}\rVert_\infty\leq \lambda m_\mu(\tcti(\lambda))|\BS^{d-1}| \lVert V\rVert_1^{1/2} \lVert \Psi\rVert_2,
\end{equation}
where $m_\mu$ was defined in \eqref{mmu}.
In particular, for fixed $\lambda$, $\lVert \widehat{\Phi_\lambda} \chi_{p^2<2\mu}\rVert_1 <\infty$ and from Lemma~\ref{lea:Tc0_asy} it follows that $\lVert \widehat{\Phi_\lambda} \chi_{p^2<2\mu}\rVert_1 $ is bounded for $\lambda\to 0$.

It only remains to prove that $\lVert \widehat{\Phi_\lambda} \chi_{p^2>2\mu}\rVert_1$ is bounded for fixed $\lambda$ and is $O(\lambda)$ for $\lambda\to 0$.
By \eqref{BT_bound} $B_{T}(p,0)\chi_{p^2>2\mu}\leq C/(1+p^2)$ for some $C$ independent of $T$.
Using \eqref{eval_eq} and applying H\"older's inequality and Lemma~\ref{lp_prop}\eqref{lp_prop.3},
\begin{equation}
\lVert \widehat{\Phi_\lambda} \chi_{p^2>2\mu}\rVert_s \leq C \lambda \left \lVert \frac{1}{1+\vert \cdot \vert^2}\right \rVert_p  \lVert \widehat{V^{1/2} \Psi_{\tcti(\lambda)}} \rVert_q
\leq C \lambda \left \lVert \frac{1}{1+\vert \cdot \vert^2}\right \rVert_p  \lVert V\rVert_t^{1/2} \rVert \Psi \rVert_2
\end{equation}
where $1/s=1/p+1/q$ and $q=2t/(t-1)$.
For $d=1$ the claim follows with the choice $t=p=1$.
For $d=2$, $V\in L^{1+\eps}$ for some $0<\eps\leq 1$. With the choice $t=1+\eps, p=2t/(t+1)>1$ the claim follows.

For $d=3$, we may choose $1 \leq t\leq 3/2$ and $3/2<p\leq \infty$ which gives
\begin{equation}\label{pf_wl_10}
\lVert \widehat \Phi_\lambda \chi_{p^2>2\mu}\rVert_s=O(\lambda)
\end{equation} for all $6/5<s\leq \infty$.
We use a bootstrap argument to decrease $s$ to one.
Let us use the short notation $B$ for multiplication with $B_{T}(p,0)$ in momentum space and $F:L^2(\BR^d)\to L^2(\BR^d)$ the Fourier transform.
Using \eqref{eval_eq} one can find by induction that
\begin{equation}\label{pf_wl_9}
\widehat{\Phi_\lambda} \chi_{p^2>2\mu}= \lambda^{n} (\chi_{p^2>2\mu} B F V F^\dagger)^n \widehat{\Phi_\lambda} \chi_{p^2>2\mu}+\sum_{j=1}^n \lambda^{j}(\chi_{p^2>2\mu} B F VF^\dagger )^j \widehat{\Phi_\lambda}\chi_{p^2<2\mu}
\end{equation}
for any $n\in \BN$.
The strategy is to prove that applying $\chi_{p^2>2\mu} B F VF^\dagger$ to an $L^r$ function will give a function in $L^s\cap L^\infty$, where $s/r< c <1$ for some fixed constant $c$.
For $n$ large enough, the first term will be in $L^1$, while the second term is in $L^1$ for all $n$ since $\widehat{\Phi_\lambda}\chi_{p^2<2\mu}$ is $L^1$.

\begin{lemma}\label{BVpsi}
Let $V\in L^1\cap L^{3/2+\eps}(\BR^3)$ for some $0<\eps\leq 1/2$ and let $1\leq r \leq 3/2$ and $f \in L^r(\BR^3)$.
Let $2\geq q\geq r$ and $3/2<t\leq \infty$.
\begin{enumerate}[(i)]
\item Then,
\begin{equation}
\lVert \chi_{p^2>2\mu} B F V F^\dagger f \rVert_s \leq C(r,q)\left \lVert \frac{1}{1+|\cdot|^2}\right \rVert_t \lVert V \rVert_{q}\lVert f \rVert_{r}
\end{equation}
where $1/s=1/t+1/r-1/q$ and $C(r,q) < \infty$.
(For $s<1$, $\lVert \cdot \rVert_s$ has to be interpreted as $\lVert f \rVert_s=\left(\int_{\BR^3} |f(p)|^s \dd p\right)^{1/s}$.) \label{BVpsi.1}
\item Let $c=\frac{\eps}{(3+\eps)(3+2\eps)}>0$ and let $r/(1+c)\leq s \leq \infty$.
Then $\lVert \chi_{p^2>2\mu} B F V F^\dagger f \rVert_s \leq C(r,s) \lVert f \rVert_{r}$ for $C(r,s)<\infty$. \label{BVpsi.2}
\end{enumerate}
\end{lemma}
\begin{proof}[Proof of Lemma~\ref{BVpsi}]
\eqref{BVpsi.1}:
Using \eqref{BT_bound} we have $|\chi_{p^2>2\mu} B F V F^\dagger f  (p)|\leq \frac{C}{1+p^2} | \widehat V * f (p)|$.
By the Young and Hausdorff-Young inequalities, the convolution satisfies
\begin{equation}
\lVert \widehat V * f \rVert_p \leq C(q,r) \lVert V \rVert_q \rVert f \rVert_r
\end{equation}
for some finite constant $ C(q,r)$, where $1/p=1/r-1/q$.
The claim follows from H\"older's inequality.

\eqref{BVpsi.2}:
For fixed $r$ and choosing $q,t$ in the range $r \leq q \leq 3/2+\epsilon$ and $3/2+\eps/2 \leq t \leq \infty$, $s=(1/t+1/r-1/q)^{-1}$ can take all values in $[r/(1+c), \infty]$.
The claim follows immediately from \eqref{BVpsi.1}.
\end{proof}

Let $n$ be the smallest integer such that $\frac{7}{5} \frac{1}{(1+c)^n}\leq 1$.
To bound the first term in \eqref{pf_wl_9}, recall from \eqref{pf_wl_10} that $\lVert \widehat \Phi_\lambda \chi_{p^2>2\mu}\rVert_{s}=O(\lambda)$ for $s=7/5$.
We apply the second part of Lemma~\ref{BVpsi} $n$ times.
After the $j$th step, we have $\lVert (\chi_{p^2>2\mu} B F V F^\dagger)^j \widehat{\Phi_\lambda} \chi_{p^2>2\mu}\rVert_s=O(\lambda)$ for $s=\frac{7}{5} \frac{1}{(1+c)^j}$.
In the $n$th step we pick $s=1$ and obtain $\lVert (\chi_{p^2>2\mu} B F V F^\dagger)^n \widehat{\Phi_\lambda} \chi_{p^2>2\mu}\rVert_1=O(\lambda)$.
To bound the second term in \eqref{pf_wl_9} recall that $\lVert \widehat{\Phi_\lambda}\chi_{p^2<2\mu}\rVert_1=O(1)$.
Applying the first part of Lemma~\ref{BVpsi} with $r=1, t=q=3/2+\eps$ implies that
\begin{multline}
\left \lVert \sum_{j=1}^n \lambda^{j}(\chi_{p^2>2\mu} B F V F^\dagger)^j \widehat{\Phi_\lambda}\chi_{p^2<2\mu}\right\rVert_1 \\
\leq \sum_{j=1}^n \lambda^{j}\left(C(1,3/2+\eps)\left \lVert \frac{1}{1+|\cdot|^2}\right \rVert_{3/2+\eps} \lVert V \rVert_{3/2+\eps}\right)^j\lVert \widehat{\Phi_\lambda}\chi_{p^2<2\mu} \rVert_{1} =O(\lambda).
\end{multline}
It follows that $\lVert \widehat \Phi_\lambda \chi_{p^2>2\mu}\rVert_1$ is finite and $O(\lambda)$ for $d=3$.
\end{proof}

\begin{proof}[Proof of Lemma~\ref{phi_infinity_conv}]
Using the eigenvalue equation \eqref{eval_eq}, we write
\begin{multline}\label{phi_infinity_conv_1}
\Phi_\lambda(r)= \int_{\vert p \vert>\sqrt{2\mu}}\frac{e^{i p \cdot  r}}{(2\pi)^{3/2}} \widehat{\Phi}_\lambda(p) \dd p \\
+\lambda \int_{\vert p \vert<\sqrt{2\mu}}\frac{e^{i p \cdot (r-r')}-e^{i \sqrt{\mu}\frac{p}{|p|} \cdot (r-r')}}{(2\pi)^3} B_{\tcti(\lambda)}(p,0)|V|^{1/2} (r') \Psi_{\tcti(\lambda)}(r')  \dd p \dd r'\\
+\lambda \int_{\vert p \vert<\sqrt{2\mu}}\frac{e^{i \sqrt{\mu}\frac{p}{|p|} \cdot (r-r')}}{(2\pi)^3} B_{\tcti(\lambda)}(p,0)|V|^{1/2} (r') (\Psi_{\tcti(\lambda)}(r')-V^{1/2} (r') j_3(r'))  \dd p \dd r'\\
+\lambda \int_{\vert p \vert<\sqrt{2\mu}}\frac{e^{i \sqrt{\mu}\frac{p}{|p|} \cdot (r-r')}}{(2\pi)^3} B_{\tcti(\lambda)}(p,0)V(r')j_3(r')  \dd p \dd r'
\end{multline}
We prove that the first three terms have $L^\infty$-norm of order $O(\lambda^{1/2})$.
For the first term this follows from Lemma~\ref{phi_infinity_bd}.
For the second term in \eqref{phi_infinity_conv_1}, we proceed as in the proof of \cite[Lemma 3.1]{hainzl_bardeencooperschrieffer_2016}.
First, integrate over the angular variables
\begin{multline} \label{phi_infinity_conv_2}
\int_{\vert p \vert<\sqrt{2\mu}}\left[e^{i p \cdot (r-r')}-e^{i \sqrt{\mu}\frac{p}{|p|} \cdot (r-r')}\right] B_{\tcti(\lambda)}(p,0)\dd p\\
=\int_{\vert p \vert<\sqrt{2\mu}}\left[\frac{\sin \vert p\vert  \vert r-r'\vert}{ \vert p\vert  \vert r-r'\vert}-\frac{\sin \sqrt{\mu}  \vert r-r'\vert}{ \sqrt{\mu} \vert r-r'\vert}\right] B_{\tcti(\lambda)}(\vert p\vert,0)  |p|^2 \dd \vert p\vert,
\end{multline}
where we slightly abuse notation writing $B_T(|p|,0)$ for the radial function $B_T(p,0)$.
Bounding the absolute value of this using $\vert \sin x/x-\sin y/y|<C |x-y|/|x+y|$ and $B_{T}(p,0)\leq 1/|p^2-\mu|$ gives
\begin{equation}
\eqref{phi_infinity_conv_2} \leq C\int_{\vert p \vert<\sqrt{2\mu}}\frac{|p|^2}{( \vert p\vert  +\sqrt{\mu})^2} \dd \vert p\vert=:\tilde C<\infty.
\end{equation}
In particular, the second term in \eqref{phi_infinity_conv_1} is bounded uniformly in $r$ by
\begin{equation}
\lambda \frac{\tilde C}{(2\pi)^3} \lVert V\rVert_1^{1/2} \lVert \Psi_{\tcti(\lambda)} \rVert_2
\end{equation}
which is of order $O(\lambda)$.

To bound the absolute value of the third term in \eqref{phi_infinity_conv_1}, we pull the absolute value into the integral, carry out the integration over $p$ and use the Schwarz inequality in $r'$. This results in the bound
\begin{equation}
\lambda \frac{\vert \BS^2 \vert}{(2\pi)^3} m_\mu(\tcti(\lambda)) \lVert V\rVert_1^{1/2} \lVert \Psi_{\tcti(\lambda)} -\Psi \rVert_2.
\end{equation}
By Lemma~\ref{lea:Tc0_asy}, $\lambda m_\mu(\tcti(\lambda))$ is bounded and by Lemma~\ref{psiT_asymptotic_d}, $\lVert\Psi_{\tcti(\lambda)}- \Psi \rVert_2$ decays like $\lambda^{1/2}$ for small $\lambda$.

The fourth term in \eqref{phi_infinity_conv_1} equals $\lambda m_\mu(\tcti(\lambda)) \FT^\dagger \FT V j_3$, where we carried out the radial part of the $p$ integration.
Recall that $j_3=\FT^\dagger 1_{\BS^2}$ and $ \mathcal{V}_\mu 1_{\BS^2} =e_\mu 1_{\BS^2}  $, where $1_{\BS^2} $ is the constant function with value 1 on $\BS^2$.
Hence, $ \FT^\dagger \FT V j_3=\FT^\dagger \mathcal{V}_\mu 1_{\BS^2} = e_\mu j_3$ and the fourth term in \eqref{phi_infinity_conv_1} equals $\lambda m_\mu(\tcti(\lambda)) e_\mu  j_3$.
By Lemma~\ref{lea:Tc0_asy}, $\lambda m_\mu(\tcti(\lambda)) e_\mu= 1+O(\lambda)$ as $\lambda \to 0$.
Thus, $ \lVert \Phi_\lambda -j_3 \rVert_\infty= \vert \lambda m_\mu(\tcti(\lambda)) e_\mu-1 \vert \lVert j_3 \rVert_\infty+O(\lambda)=O(\lambda) $.
\end{proof}

\section{Proof of Theorem~\ref{thm1}}\label{sec:thm1}
Instead of directly looking at $H_T^{\Omega_1}$, we extend the domain to $L^2(\BR^{2d})$ by extending the wavefunctions (anti)symmetrically across the boundary.
Recall that $\tilde x$ denotes the vector containing all but the first component of $x$.
The half-space operator $H_T^{\Omega_1}$ with Dirichlet/Neumann boundary conditions is unitarily equivalent to
\begin{equation}
H^{\rm ext}_{T}=K_{T}^{\BR^d}-\lambda V(x-y) \chi_{\vert x_1-y_1 \vert <\vert x_1+y_1 \vert}-\lambda V(x_1+y_1,\tilde x-\tilde y) \chi_{\vert x_1+y_1 \vert<\vert  x_1-y_1 \vert}
\end{equation}
on $L^2(\BR^{d}\times \BR^{d})$ restricted to functions antisymmetric/symmetric under swapping $x_1\leftrightarrow -x_1$ and symmetric under exchange of $x \leftrightarrow y$.
Next, we express $H^{\rm ext}_{T}$ in relative and center of mass coordinates $r=x-y$ and $z=x+y$.
Let $U$ be the unitary on $L^2(\BR^{2d})$ given by $U\psi(r,z)=2^{-d/2} \psi((r+z)/2,(z-r)/2)$.
Then
\begin{equation}
H^1_{T}:= U H^{\rm ext}_{T}U^\dagger=UK_{T}^{\BR^d}U^\dagger-\lambda V(r) \chi_{\vert r_1 \vert <\vert z_1 \vert}-\lambda V(z_1,\tilde r) \chi_{\vert z_1 \vert<\vert  r_1 \vert}
\end{equation}
on $L^2(\BR^{2d})$ restricted to functions antisymmetric/symmetric under swapping $r_1\leftrightarrow z_1$ and symmetric in $r$.
The spectra of $H^1_{T}$ and $H_T^{\Omega_1}$ agree.

For an upper bound on $\inf \sigma(H_{T}^1)$, we restrict $H_{T}^1$ to zero momentum in the translation invariant center of mass directions and call the resulting operator $\tilde H_{T}^1$.
The operator $\tilde H_{T}^1$ acts on $\{\psi \in L^2(\BR^d\times \BR)| \psi(r,z_1)=\psi(-r,z_1)=\mp \psi((z_1,\tilde r),r_1)\}$.
The kinetic part of $\tilde H_{T}^1$ reads
\begin{equation}
\tilde K_{T}(r,z_1;r',z_1')=\int_{\BR^{d+1}} \frac{e^{i p(r-r')+iq_1(z_1-z_1')}}{(2\pi)^{d+1}} B_{T}^{-1}(p,(q_1,\tilde 0)) \dd p \dd q_1.
\end{equation}

An important property is the continuity of $\inf \sigma(H_T^1)$, proved in Section~\ref{sec:pf_H1-Tcont}.
\begin{lemma}\label{H1-Tcont}
Let $d\in\{1,2,3\}$ and let $V$ satisfy Assumption~\ref{aspt_V_halfspace}.
Then $\inf \sigma(H_T^0)$, $\inf \sigma(H_T^{\Omega_0})$ and $\inf \sigma(H_T^1)$ depend continuously on $T$ for $T>0$.
\end{lemma}

To prove Theorem~\ref{thm1} we show that there is a $\lambda_1>0$ such that for $\lambda\leq \lambda_1$, $\inf \sigma(H_{\tcf(\lambda)}^1)\leq \inf \sigma(\tilde H_{\tcf(\lambda)}^1)<0$.
For all $T<\tcf(\lambda)$ we have by Lemma~\ref{H1_esspec} that $\inf \sigma(H_{T}^1)\leq \inf \sigma(H_{T}^{\Omega_0})<0$.
By continuity (Lemma~\ref{H1-Tcont}) there is an $\epsilon>0$ such that $\inf \sigma(H_{T}^1)<0$ for all $T<\tcf(\lambda)+\epsilon$.
Therefore, $\tch(\lambda)> \tcf(\lambda)$.

To prove that $\inf \sigma(\tilde H_{\tcf(\lambda)}^1)<0$ for small enough $\lambda$, we pick a suitable family of trial states $\psi_\eps(r,z_1)$.
Let $\lambda$ be such that $\tcf(\lambda)=\tcti(\lambda)$ and $H_{\tcti(\lambda)}^0$ has a unique and radial ground state $\Phi_\lambda$.
According to Remark~\ref{gs_symm_lsmall}, this is the case for $0<\lambda\leq \lambda_0$.
We choose the trial states
\begin{equation}
\psi_\epsilon(r,z_1)=\Phi_{\lambda}(r)e^{-\epsilon \vert z_1 \vert}\mp \Phi_{\lambda}(z_1, \tilde r)e^{-\epsilon\vert r_1 \vert},
\end{equation}
with the $-$ sign for Dirichlet and $+$ for Neumann boundary conditions.
Since $\Phi_{\lambda}(r)=\Phi_{\lambda}(-r)=\Phi_{\lambda}(-r_1,\tilde r)$, these trial states satisfy the symmetry constraints and lie in the form domain of $\tilde H_T^1$.
The norm of $\psi_\eps$ diverges as $\eps\to0$.
\begin{rem}
The trial state is the (anti-)symmetrization of $\Phi_{\lambda}(r)e^{-\epsilon \vert z_1 \vert}$, i.e.~the projection of $\Phi_{\lambda}(r)e^{-\epsilon \vert z_1 \vert}$ onto the domain of $\tilde H_T^{1}$.
The intuition behind our choice is that, as we will see in Section~\ref{sec:rel_T}, at weak coupling the Birman-Schwinger operator corresponding to $H_T^{\Omega_1}$  approximately looks like $A^0_T$ (defined in \eqref{BS-ti}) on a restricted domain.
This is why we want our trial state to look like the ground state $\Phi_{\lambda}$ of $H_T^{0}$ projected onto the domain of  $\tilde H_T^{1}$.
\end{rem}
We shall prove that $\lim_{\eps\to0} \langle \psi_\eps, \tilde H^1_{\tcf(\lambda)} \psi_\epsilon \rangle$ is negative for weak enough coupling.
This is the content of the next two Lemmas, which are proved in Sections~\ref{pf:plugin_triaL_state} and \ref{pf:claim_weak_lambda_asy}, respectively.
\begin{lemma}\label{plugin_triaL_state}
Let $d\in\{1,2,3\}$, $\mu>0$ and let $V$ satisfy Assumption~\ref{aspt_V_halfspace}.
Let $\lambda$ be such that $H_{\tcti(\lambda)}^0$ has a unique ground state $\Phi_\lambda$.
Then,
\begin{multline}\label{aspt_terms}
\lim_{\epsilon \to 0} \langle \psi_\epsilon , \tilde H^1_{\tcti(\lambda)} \psi_\epsilon \rangle = -2\lambda \Bigg(\int_{\BR^{d+1}} V(r) \Bigg[ -\vert \Phi_\lambda(r)\mp \Phi_\lambda(z_1, \tilde r) \vert^2 \chi_{\vert z_1 \vert< \vert r_1 \vert}
 +\vert \Phi_\lambda(z_1, \tilde r)\vert ^2 \Bigg] \dd r \dd z_1  \\
\mp 2\pi \int_{\BR^{d-1}} \overline{ \widehat \Phi_{\lambda}(0,\tilde p)} \widehat{V\Phi_{\lambda}}(0,\tilde p) \dd \tilde p\Bigg),
\end{multline}
where the upper signs correspond to Dirichlet and the lower signs to Neumann boundary conditions.
For $d=1$, the last term in \eqref{aspt_terms} is to be understood as $ \mp 2\pi \overline{ \widehat \Phi_{\lambda}(0)} \widehat{V\Phi_{\lambda}}(0)$.
\end{lemma}
For small $\lambda$ we shall prove that the expression in the round bracket in \eqref{aspt_terms} is positive.
\begin{lemma}\label{claim_weak_lambda_asy}
Let $d\in\{1,2,3\}$, $\mu>0$ and  let $V$ satisfy Assumption~\ref{aspt_V_halfspace}.
Let $\lambda_0$ be as in Remark~\ref{gs_symm_lsmall}.
Assume Dirichlet or Neumann boundary conditions.
For $d=3$ assume that $\int_{\BR^3} V(r) \widetilde m_3^{D/N}(r) \dd r>0$, where $\widetilde m_3^{D/N}$ was defined in \eqref{mtilde}.
Then there is a $\lambda_0\geq \lambda_1>0$ such that for $\lambda\leq \lambda_1$ the right hand side in \eqref{aspt_terms} is negative.
\end{lemma}
Therefore, for small enough $\eps$, $\langle \psi_\eps, \tilde H^1_{\tcti(\lambda)} \psi_\epsilon \rangle<0$.
Since $\tcf$ and $\tcti$ coincide at weak coupling, this proves that $\inf \sigma(\tilde H_{\tcf(\lambda)}^1)<0$ at weak coupling.
This concludes the proof of Theorem~\ref{thm1}.

\begin{rem}
The additional condition $\int_{\BR^3} V(r) \widetilde m_3^{D/N}(r) \dd r>0$ for $d=3$ is exactly the limit of the terms in the round brackets in \eqref{aspt_terms} for $\lambda\to0$.
Taking the limit amounts to replacing $\Phi_\lambda$ by $j_3$ (cf. Lemma~\ref{phi_infinity_conv}).
\end{rem}

\subsection{Proof of Lemma~\ref{H1-Tcont}}\label{sec:pf_H1-Tcont}
\begin{proof}[Proof of Lemma~\ref{H1-Tcont}]
Let $0<T_0<T_1<\infty$.
We claim that there exists a constant $C_{T_0,T_1}$ such that $|K_T(p,q)-K_{T'}(p,q)|\leq C_{T_0,T_1} |T-T'|(1+p^2+q^2)$ for all $T_0\leq T,T'\leq T_1$.
To see this, compute
\begin{equation}
\frac{\p}{\p T} K_T(p,q)=\frac{K_T(p,q)}{2T^2} \frac{\sech\left(\frac{p^2-\mu}{2T}\right)^2(p^2-\mu)+\sech\left(\frac{q^2-\mu}{2T}\right)^2(q^2-\mu)}{\tanh\left(\frac{p^2-\mu}{2T}\right)+\tanh\left(\frac{q^2-\mu}{2T}\right)}.
\end{equation}
$K_T$ can be estimated using Lemma~\ref{KT-Laplace} and the remaining term is bounded.

The kinetic part $K_T^0$ of $H_T^0$ acts as multiplication by $K_T(p,0)$ in momentum space.
For $T_0<T,T'<T_1$ and $\psi$ in the Sobolev space $H^1(\BR^{d})$, therefore
\begin{equation}
\langle \psi, (K_T^{0}-K_{T'}^{0} )\psi \rangle \leq C_{T_0,T_1} |T-T'| \lVert \psi\rVert_{H^1(\BR^{d})}.
\end{equation}
Similarly, for $T_0<T,T'<T_1$ and $\psi \in H^1(\BR^{2d})$,
\begin{equation}
\langle \psi, (K_T^{\BR^d}-K_{T'}^{\BR^d} )\psi \rangle \leq C_{T_0,T_1} |T-T'|  \lVert \psi\rVert_{H^1(\BR^{2d})}.
\end{equation}

Set $D_0:=H^1(\BR^{d})$, $D_{\Omega_0}:=\{\psi \in H^1(\BR^{2d}) \vert \psi(x,y)=\psi(y,x)\}$ and $D_1:=\{\psi \in H^1(\BR^{2d}) \vert \psi(x,y)=\psi(y,x)=\mp\psi((-x_1,\tilde x) ,y)\}$, where $-/+$ corresponds to Dirichlet/Neumann boundary conditions, respectively.
Let $j\in\{0,1, \Omega_0\}$ and $\epsilon>0$.
There is a family $\{\psi_T \}$ of functions in $D_j$ such that $\lVert \psi_T \rVert_2=1$ and $\langle \psi_T,H_T^j \psi_T \rangle \leq \inf \sigma(H_T^j) +\epsilon$.

We first argue that there is a constant $C>0$ such that for all $T\in[T_0,T_1]: \lVert \psi_T \rVert_{H^1}<C$.
Recall that $2T$ lies in the essential spectrum of $H_T^0$ restricted to even functions.
Together with Lemmas~\ref{H1_esspec} and \ref{ti_reduction}, $\langle \psi_T,H_T^j \psi_T \rangle \leq 2T_1 +\epsilon$.
Furthermore, by Lemma~\ref{KT-Laplace}, the kinetic part of $H_T^j$ is bounded below by some constant $C_1(T_0)(1-\Delta)$, where $\Delta$ denotes the Laplacian in all variables.
Since the interaction is infinitesimally form bounded with respect to the Laplacian, there is a finite constant $C_2(T_0)$, such that for all $\psi \in D_j$ with $\lVert \psi \rVert_2=1$,
$\langle \psi,H_T^j \psi \rangle \geq \frac{C_1(T_0)}{2}\langle \psi, (1-\Delta)\psi \rangle -C_2(T_0)=\frac{C_1(T_0)}{2} \lVert \psi\rVert_{H^1}-C_2(T_0)$.
In particular, $\lVert \psi_T\rVert_{H^1} \leq \frac{2}{C_1(T_0)}( 2T_1 +\epsilon+C_2(T_0))=:C$.

Let $T,T'\in [T_0,T_1]$.
Then
\begin{multline}
 \inf \sigma(H_T^j) +\epsilon \geq \langle \psi_T,H_T^j \psi_T \rangle=\langle \psi_T,H_{T'}^j \psi_T \rangle+\langle \psi_T, (K_T-K_{T'}) \psi_T \rangle\\
\geq  \inf \sigma(H_{T'}^j) -|T-T'| C_{T_0,T_1} C  .
\end{multline}
Swapping the roles of $T,T'$, we obtain
\begin{equation}
 \inf \sigma(H_T^j) -\epsilon - |T-T'| C_{T_0,T_1} C\leq  \inf \sigma(H_{T'}^j) \leq  \inf \sigma(H_T^j) +\epsilon +|T-T'| C_{T_0,T_1} C
\end{equation}
and thus
\begin{equation}
 \inf \sigma(H_T^j) -\epsilon  \leq  \lim_{T'\to T} \inf \sigma(H_{T'}^j) \leq  \inf \sigma(H_T^j) +\epsilon .
\end{equation}
Since $\epsilon$ was arbitrary, equality follows.
Hence $ \inf \sigma(H_T^j) $ is continuous in $T$ for $T>0$.
\end{proof}

\subsection{Proof of Lemma~\ref{plugin_triaL_state}}\label{pf:plugin_triaL_state}
The following technical lemma will be helpful for $d=3$.
\begin{lemma}\label{lea:pf_plugin_1}
Let $V,W \in L^1\cap L^{3/2}(\BR^3)$, let $W$ be radial and let $\psi\in L^2(\BR^3)$.
Then
\begin{multline}
\int_{\BR^5} \vert \widehat{V^{1/2} \psi}(p) \vert \frac{1}{1+p^2} \widehat{W}(0,\tilde p-\tilde q) \frac{1}{1+p_1^2 +\tilde q^2} \vert \widehat{V^{1/2} \psi}  (p_1,\tilde q)\vert \dd p \dd \tilde q\\
\leq C \lVert \widehat{W}(0,\cdot) \rVert_{L^3(\BR^2)} \lVert V\rVert_{3/2} \lVert \psi \rVert_2^2<\infty
\end{multline}
for some constant $C$ independent of $V$, $W$ and $\psi$.
\end{lemma}
\begin{proof}[Proof of Lemma~\ref{lea:pf_plugin_1}]
By Lemma~\ref{lp_prop}\eqref{lp_prop.4}, $\widehat{W}(0, \cdot) \in L^3(\BR^2)\cap L^\infty(\BR^2)$.
By Young's inequality, the integral is bounded by
\begin{equation}
C \lVert \widehat{W}(0,\cdot) \rVert_{L^3(\BR^2)} \int_\BR \left \vert \int_{\BR^2} \left\vert\frac{1}{1+p^2}\widehat{V^{1/2} \psi}(p) \right\vert^{6/5} \dd \tilde p \right \vert^{5/3} \dd p_1
\end{equation}
By Lemma~\ref{lp_prop}\eqref{lp_prop.3}, $\lVert \widehat{V^{1/2} \psi}\rVert_6 \leq C \lVert V \rVert_{3/2}^{1/2} \lVert \psi \rVert_2$.
Applying H\"older's inequality in the $\tilde p$ variables, we obtain the bound
\begin{equation}
C \lVert \widehat{W}(0,\cdot) \rVert_{L^3(\BR^2)} \int_\BR \left \vert \int_{\BR^2} \frac{1}{(1+p^2)^{3/2}}\dd \tilde p \right \vert^{4/3} \left \vert \int_{\BR^2} |\widehat{V^{1/2} \psi}(p)|^6\dd \tilde p \right \vert^{1/3}\dd p_1
\end{equation}
Applying H\"older's inequality in $p_1$, we further obtain
\begin{equation}
C \lVert \widehat{W}(0,\cdot) \rVert_{L^3(\BR^2)} \left(\int_\BR \left \vert \int_{\BR^2} \frac{1}{(1+p^2)^{3/2}}\dd \tilde p \right \vert^{2} \dd p_1\right)^{2/3} \lVert \widehat{V^{1/2}\psi}\rVert_6^2
\end{equation}
The remaining integral is finite.
\end{proof}

\begin{proof}[Proof of Lemma~\ref{plugin_triaL_state}]
Plugging in the trial state and regrouping terms we obtain
\begin{multline}\label{H1_exp}
\langle \psi_\epsilon , \tilde H^1_{\tcti(\lambda)} \psi_\epsilon \rangle =
2\int_{\BR^{2d+2}}  \Bigg [ \overline{\Phi_{\lambda}(r)}e^{-\eps \vert z_1 \vert}(K_{T}(r,z_1;r',z_1')-\lambda V(r)\delta(r-r' )\delta(z_1-z_1'))\Phi_{\lambda}(r')e^{-\eps \vert z_1' \vert}\\
\mp\overline{\Phi_{\lambda}(r)}e^{-\eps \vert z_1 \vert}(K_{T}(r,z_1;r',z_1')-\lambda V(r) \delta(r-r' )\delta(z_1-z_1'))e^{-\eps \vert r_1' \vert}\Phi_{\lambda}(z_1',\tilde r')\Bigg]  \dd r\dd z_1 \dd r' \dd z_1' \\
+2\int_{\BR^{d+1}}  \Bigg [  \lambda V(r)\chi_{\vert z_1 \vert<\vert r_1 \vert}|\Phi_{\lambda}(r)|^2 e^{-2\eps \vert z_1 \vert}
\mp\overline{\Phi_{\lambda}(r)}e^{-\eps \vert z_1 \vert} \lambda V(r)\chi_{\vert z_1 \vert<\vert r_1 \vert}e^{-\eps \vert r_1 \vert}\Phi_{\lambda}(z_1,\tilde r)\\
+ \lambda V(z_1, \tilde r)\chi_{\vert r_1 \vert<\vert z_1 \vert}|\Phi_{\lambda}(r)|^2 e^{-2\eps \vert z_1 \vert}
\mp\overline{\Phi_{\lambda}(r)}e^{-\eps \vert z_1 \vert} \lambda  V(z_1, \tilde r)\chi_{\vert z_1 \vert>\vert r_1 \vert}e^{-\eps \vert r_1 \vert}\Phi_{\lambda}(z_1,\tilde r) \\
- \lambda V(z_1, \tilde r)|\Phi_{\lambda}(r)|^2 e^{-2\eps \vert z_1 \vert}\pm\overline{\Phi_{\lambda}(r)}e^{-\eps \vert z_1 \vert} \lambda  V(z_1, \tilde r)e^{-\eps \vert r_1 \vert}\Phi_{\lambda}(z_1,\tilde r)
\Bigg]  \dd r\dd z_1
\end{multline}
We will prove that the first integral vanishes due to the eigenvalue equation $H^0_{\tcti(\lambda)}\Phi_\lambda =0$ as $\eps\to 0$.
For the second integral in \eqref{H1_exp}, we will show that it is bounded as $\epsilon \to 0$ and argue that it is possible to interchange limit and integration.
The limit of the second integral is exactly the right hand side of $\eqref{aspt_terms}$.

The first two terms in the integrand of the second integral in \eqref{H1_exp} can be bounded by $\lambda \lVert\Phi_{\lambda}\rVert_\infty^2  |V(r)|\chi_{|z_1| <| r_1 |}$.
This is an $L^1$ function, since $\vert \cdot \vert V\in L^1$ and $\lVert\Phi_{\lambda}\rVert_\infty< \infty$ by Lemma~\ref{phi_infinity_bd}.
The same argument applies to the next two terms as well.

For the fifth term in the second integral, we can interchange limit and integration by dominated convergence if $\int_{\BR^{d+1}} \vert V(r)\vert \vert \Phi_{\lambda}(z_1, \tilde r) \vert^2 \dd r \dd z_1<\infty$.
Observe that
\begin{equation}\label{pf_e0_1}
\int_{\BR^{d+1}} \vert V(r)\vert \vert \Phi_{\lambda}(z_1, \tilde r) \vert^2 \dd r \dd z_1 = (2\pi)^{1-d/2} \int_{\BR^{2d-1}} \overline{\widehat\Phi_{\lambda}(p)}\widehat {\vert V\vert }(0,\tilde p - \tilde q) \widehat \Phi_{\lambda}(p_1,\tilde q) \dd p \dd \tilde q
\end{equation}
According to Lemma~\ref{phi_infinity_bd}\eqref{phi_infty.1} the latter is bounded by
\begin{equation}
C  \int_{\BR^{2d-1}}\vert \widehat{V^{1/2} \Psi_{\tcti(\lambda)}}(p) \vert \frac{1}{1+p^2} \widehat{\vert V\vert}(0,\tilde p-\tilde q) \frac{1}{1+p_1^2 +\tilde q^2} \vert \widehat{V^{1/2} \Psi_{\tcti(\lambda)}}  (p_1,\tilde q)\vert \dd p \dd \tilde q
\end{equation}
For $d=1,2$ we bound this by
\begin{equation}\label{pf_e0_2}
C \lVert V \rVert_1^2 \lVert \Psi\rVert_2^2 \int_{\BR^{2d-1}} \frac{1}{(1+p_1^2+\tilde p^2)(1+p_1^2+\tilde q^2)} \dd p \dd \tilde q,
\end{equation}
which is finite.
For $d=3$, \eqref{pf_e0_2} is finite by Lemma~\ref{lea:pf_plugin_1} since $W=|V|$ is radial and in $L^1\cap L^{3/2}$. 
Hence, limit and integration can be interchanged for the fifth term in the second integral in \eqref{H1_exp}.

For the last term in \eqref{H1_exp} we have
\begin{multline}\label{4.105}
\int_{\BR^{d+1}} \overline{\Phi_{\lambda}(r)}e^{-\eps \vert z_1 \vert}  V(z_1,\tilde r)e^{-\eps \vert r_1 \vert}\Phi_{\lambda}(z_1,\tilde r) \dd z_1 \dd r\\
=\frac{2}{\pi}\int_{\BR^{d+1}} \overline{ \widehat \Phi_{\lambda}(p)} \frac{\eps}{\eps^2+q_1^2} \frac{\eps}{\eps^2+p_1^2} \widehat{V\Phi_{\lambda}}(q_1,\tilde p) \dd p  \dd q_1\\
=\frac{2}{\pi}\int_{\BR^{d+1}} \overline{ \widehat \Phi_{\lambda}(\eps p_1,\tilde p)} \frac{1}{1+q_1^2} \frac{1}{1+p_1^2} \widehat{V\Phi_{\lambda}}(\eps q_1,\tilde p) \dd p  \dd q_1.
\end{multline}
According to Lemma~\ref{phi_infinity_bd}\eqref{phi_infty.1} and Lemma~\ref{lp_prop}\eqref{lp_prop.3}, the integrand is bounded by $ \frac{C(\lambda) }{1+\tilde p^2} \frac{ \lVert V\rVert_1 \lVert \Psi \rVert_2^2}{(1+q_1^2)(1+p_1^2)}$.
For $d=1,2$ this is integrable, so by dominated convergence and since $\int_\BR \frac{1}{1+x^2} \dd x =\pi$, this term converges to the last term in \eqref{aspt_terms}.
For $d=3$, the following result will be useful.
\begin{lemma}\label{pf_e0_3d_fn}
Let $\lambda,T,\mu>0$ and $d=3$ and let $V$ satisfy~\ref{aspt_V_halfspace}.
The functions
\begin{equation}
f(p_1,q_1)=\int_{\BR^{2}}\overline{ \widehat \Phi_{\lambda}(p)} \widehat{V\Phi_{\lambda}}(q_1,\tilde p) \dd \tilde p
\end{equation}
and
\begin{equation}
g(p_1,q_1)= \int_{\BR^{2}} B^{-1}_{T}(( p_1, \tilde p),( q_1,\tilde 0)) \overline{\widehat\Phi_{\lambda}( p_1, \tilde p)}  \widehat\Phi_{\lambda}( q_1,\tilde p) \dd \tilde p
\end{equation}
are bounded and continuous.
\end{lemma}
Its proof can be found after the end of the current proof.

We write the term in \eqref{4.105} as
\begin{equation}\label{pf_e0_3}
\frac{2}{\pi}\int_{\BR^{2}} \frac{f(\eps p_1,\eps q_1)}{(1+q_1^2)(1+p_1^2)}\dd p_1  \dd q_1.
\end{equation}
By Lemma~\ref{pf_e0_3d_fn} we can exchange limit and integration by dominated convergence and \eqref{pf_e0_3} converges to the last term in \eqref{aspt_terms}.

For the second summand in the first integral in \eqref{H1_exp} we also want to argue using dominated convergence.
The interaction term agrees with \eqref{4.105}.
The kinetic term can be written as
\begin{multline}
\frac{4}{\pi} \int_{\BR^{d+1}} \frac{1}{(1+q_1^2)(1+p_1^2)}B^{-1}_{T}((\epsilon p_1, \tilde p),(\eps q_1,\tilde 0)) \overline{\widehat\Phi_{\lambda}(\epsilon p_1, \tilde p) } \widehat\Phi_{\lambda}(\epsilon q_1,\tilde p) \dd p \dd q_1 \\
=\frac{4}{\pi} \int_{\BR^{2}} \frac{1}{(1+q_1^2)(1+p_1^2)} g(\epsilon p_1,\eps q_1) \dd p_1 \dd q_1
\end{multline}
For $d=3$, we can apply dominated convergence according to Lemma~\ref{pf_e0_3d_fn}.
For $d=1,2$ note that by Lemma~\ref{phi_infinity_bd} and Lemma~\ref{KT-Laplace},
\begin{multline}\label{pf_e0_5}
B^{-1}_{T}(p,(q_1,\tilde 0)) \vert \widehat\Phi_{\lambda}(p) \vert\vert \widehat\Phi_{\lambda}(q_1, \tilde p)\vert\leq C_{T,\mu,\lambda}\frac{1+p^2+q_1^2}{(1+p^2)(1+q_1^2+\tilde p^2)}\lVert V\rVert_1 \lVert \Psi \rVert_2^2 \\
\leq 2 C_{T,\mu,\lambda}\frac{\lVert V\rVert_1 \lVert \Psi \rVert_2^2 }{1+\tilde p^2}
\end{multline}
Therefore, the integrand is bounded by $\frac{C \lVert V\rVert_1 \lVert \Psi \rVert_2^2 }{(1+q_1^2)(1+p_1^2)(1+\tilde p^2)}$.
For $d=1,2$ this is integrable and we can apply dominated convergence.
We conclude
that the limit of the second summand in the first integral in \eqref{H1_exp} as $\epsilon \to 0$ equals
\begin{equation}
4 \pi \int_{\BR^{d-1}}\left(  \frac{\vert \widehat\Phi_{\lambda}(0,\tilde p) \vert^2}{B_{T}((0,\tilde p),0)}-\lambda \overline{ \widehat\Phi_{\lambda}(0,\tilde p)} \widehat{ V \Phi_{\lambda}}(0,\tilde p) \right) \dd \tilde p= 0
\end{equation}
where we used that $\int_\BR \frac{1}{1+x^2} \dd x =\pi$ and \eqref{eval_eq}.

To see that the first summand in the first integral in \eqref{H1_exp} vanishes as $\eps\to0$, we use \eqref{eval_eq} to obtain
\begin{equation}
\frac{2}{\epsilon}\lambda  \int_{\BR^d} V(r) \vert\Phi_{\lambda}(r)\vert^2\dd r = \frac{2}{\eps} \int_{\BR^d} B_{T}^{-1}(p,0)\vert\widehat\Phi_{\lambda}(p)\vert^2 \dd p
= \frac{4}{\pi} \int_{\BR^{d+1}} \frac{\eps^2}{(\eps^2+q_1^2)^2} B_{T}^{-1}(p,0)\vert\widehat\Phi_{\lambda}(p)\vert^2 \dd p  \dd q_1.
\end{equation}
Hence, we need to prove that
\begin{equation}
\lim_{\epsilon \to 0} \int_{\BR^{d+1}} \frac{\epsilon^2}{(\epsilon^2+q_1^2)^2}  (B^{-1}_{T}(p,(q_1, \tilde 0))- B^{-1}_{T}(p,0))\vert \widehat\Phi_{\lambda}(p) \vert^2 \dd p \dd q_1=0
\end{equation}
We split the integration into two regions with $\vert q_1 \vert>C_1$ and $\vert q_1 \vert<C_1$, respectively.
By Lemma~\ref{KT-Laplace}, we have $B^{-1}_{T}(p,q)\leq C_2 (1+p^2+q^2)$. 
Together with $\Phi_{\lambda} \in H^1(\BR^d)$ therefore
\begin{multline}
\int_{\BR^{d+1}, \vert q_1 \vert>C_1}\frac{\epsilon^2}{(\epsilon^2+q_1^2)^2}  \vert B^{-1}_{T}(p,(q_1,\tilde 0))- B^{-1}_{T}(p,0)\vert \vert \widehat\Phi_{\lambda}(p) \vert^2 \dd p \dd q_1\\
\leq 2C_2  \int_{\BR^2, \vert q _1\vert>C_1} \frac{\eps^2(1+p^2+q_1^2)\vert \widehat\Phi_{\lambda}(p) \vert^2 }{q_1^4}  \dd p \dd q_1 <  C_3  \eps^2 \lVert \Phi_\lambda \rVert_{H^1}^2,
\end{multline}
which vanishes in the limit $\epsilon \to 0$.
For the case $\vert q_1 \vert<C$, the following Lemma is useful.
Its proof can be found at the end of this Section.
\begin{lemma}\label{Bq-B0-bound}
Let $T,\mu>0$, $d\in\{1,2,3\}$.
The function
\begin{equation}
k(p,q):=\frac{1}{\vert q\vert }(B_{T}(p,q)- B_{T}(p,0))
\end{equation}
is continuous at $q=0$ and satisfies $k(p,0)=0$ for all $p\in \BR^d$.
Furthermore, there is a constant $C$ depending only on $T,\mu,d$ such that $\vert k(p,q) \vert < \frac{C}{1+p^2}$ for all $p,q\in \BR^d$.
\end{lemma}
Since $ B^{-1}_{T}(p,q)- B^{-1}_{T}(p,0)= -\frac{ \vert q \vert k(p,q)}{B_{T}(p,q) B_{T}(p,0)}$, we have
\begin{multline}
\int_{\BR^{d+1}, \vert q_1 \vert<C_1}\frac{\epsilon^2}{(\epsilon^2+q_1^2)^2}  ( B^{-1}_{T}(p,(q_1,\tilde 0))- B^{-1}_{T}(p,0)) \vert \widehat\Phi_{\lambda}(p) \vert^2 \dd p \dd q_1\\
=-\int_{\BR^{d+1}}\frac{ \vert q_1 \vert \chi_{\vert q_1 \vert<C_1/\epsilon}}{(1+q_1^2)^2} \frac{ k(p,(\epsilon q_1,\tilde 0))}{B_{T}(p,(\epsilon q_1,\tilde 0)) B_{T}(p,0)} \vert \widehat\Phi_{\lambda}(p) \vert^2 \dd p \dd q_1
\end{multline}
By Lemma~\ref{KT-Laplace} and Lemma~\ref{Bq-B0-bound}, we can bound the absolute value of the integrand by
\begin{equation}
C \frac{ \vert q_1 \vert \chi_{\vert q_1 \vert<C_1/\epsilon}}{(1+q_1^2)^2}  (1+p^2+\eps^2 q_1^2) \vert \widehat\Phi_{\lambda}(p) \vert^2
\leq C \frac{ \vert q_1 \vert }{(1+q_1^2)^2}  (1+p^2+C_1^2)\vert \widehat\Phi_{\lambda}(p) \vert^2
\end{equation}
The latter is integrable since $\Phi_\lambda \in H^1(\BR^d)$.
Thus, by dominated convergence and since $k(p,0)=0$, the integral vanishes in the limit $\epsilon \to 0$.
\end{proof}

\begin{proof}[Proof of Lemma~\ref{pf_e0_3d_fn}]
For convenience, we introduce the notation $D_f(p,q_1)=\lambda B_{T}(p,0)$ and
\begin{equation}
D_g(p,q_1)=\lambda^2 B_{T}(p,0) B_{T}(p,(q_1,\tilde0))^{-1}B_{T}((q_1,\tilde p),0).
\end{equation}
For $h\in\{f,g\}$, $D_f(p,q_1),D_g(p,q_1)\leq \frac{C}{1+\tilde p^2}$ by Lemma~\ref{KT-Laplace} and \eqref{BT_bound}.
Furthermore,
\begin{equation}
h(p_1,q_1)=\int_{\BR^{d-1}}\overline{ \widehat{V\Phi_{\lambda}}(p_1,\tilde p) }D_h(p,q_1) \widehat{V\Phi_{\lambda}}(q_1,\tilde p) \dd \tilde p
\end{equation}
using \eqref{eval_eq}.

\begin{lemma}\label{pf_e0_6}
For $h\in\{f,g\}$,
\begin{equation}
\sup_{p_1,q_1,w_1 \in \BR} \lVert D_h((p_1,\cdot),q_1)\widehat{V\Phi_{\lambda}}(w_1,\cdot)\rVert_{L^1(\BR^2)}\leq  \sup_{w_1 \in \BR} \left\lVert \frac{C}{1+\vert \cdot \vert^2}\widehat{V\Phi_{\lambda}}(w_1,\cdot)\right\rVert_{L^1(\BR^2)}<\infty.
\end{equation}
\end{lemma}
\begin{proof}
Using H\"older's inequality,
\begin{multline}\label{pf_e0_7}
\lVert D_h((p_1,\cdot),q_1)\widehat{V\Phi_{\lambda}}(w_1,\cdot)\rVert_{L^1(\BR^2)} \leq \left\lVert \frac{C}{1+\vert \cdot \vert^2}\widehat{V\Phi_{\lambda}}(w_1,\cdot)\right\rVert_{L^1(\BR^2)}\\
\leq \frac{1}{(2\pi)^{3/2}}\int_{\BR^{2}}\int_{\BR^3}\frac{C}{1+\tilde p^2} \vert \widehat{V}((w_1,\tilde p)-k)\vert \vert \widehat{\Phi_{\lambda}}(k)\vert \dd k  \dd \tilde p \\
\leq C  \left \Vert \frac{1}{1+\vert \cdot \vert^2} \right \Vert_{L^r(\BR^{2})} \int_\BR \left(\int_{\BR^2}\vert \widehat{V}(w_1-k_1,\tilde p)\vert^s \dd \tilde p\right)^{1/s} \left(\int_{\BR^2}\vert \widehat{\Phi_{\lambda}}(k)\vert \dd \tilde k\right)\dd k_1 \\
\leq C \left \Vert \frac{1}{1+\vert \cdot \vert^2} \right \Vert_{L^r(\BR^{2})} \sup_{k_1} \lVert \widehat{V}(k_1,\cdot)\rVert_s  \lVert \widehat \Phi_{\lambda}\rVert_1 ,
\end{multline}
where $1=1/r+1/s$.
For this to be finite we need $r>1$, i.e. $s<\infty$.
By Lemma~\ref{lp_prop}\eqref{lp_prop.4}, $\sup_{q_1} \lVert \widehat{V}(q_1,\cdot)\rVert_3 <\infty$.
Furthermore $\lVert \widehat \Phi_{\lambda} \rVert_1 $ is bounded by Lemma~\ref{phi_infinity_bd}.
\end{proof}

The functions $f$ and $g$ are bounded, as can be seen using that $\lVert \widehat{V \Phi_{\lambda}}\rVert_\infty\leq C \lVert V\rVert_1^{1/2} \lVert \Psi_{\tcti(\lambda)} \rVert_2$ by Lemma~\ref{lp_prop}\eqref{lp_prop.3} and  $\lVert \Psi_{\tcti(\lambda)} \rVert_2=1$, hence we get that for $h\in \{f,g\}$
\begin{equation}
|h(p_1,q_1)|\leq C  \lVert V\rVert_1^{1/2} \sup_{p_1,q_1} \lVert D_h((p_1,\cdot),q_1)\widehat{V\Phi_{\lambda}}(q_1,\cdot)\rVert_{L^1(\BR^2)},
\end{equation}
which is finite by Lemma~\ref{pf_e0_6}.
To see continuity, we write for $h\in \{f,g\}$
\begin{multline}\label{pf_e0_4}
|h(p_1+\epsilon_1,q_1+\epsilon_2)-h(p_1,q_1)|\leq \\
\Bigg[ \Bigg| \int_{\BR^2}\overline{(\widehat{V\Phi_{\lambda}}(p_1+\eps_1,\tilde p)- \widehat{V \phi_{T,\lambda}}(p))} D_h((p_1+\eps_1,\tilde p),q_1+\eps_2) \widehat{V\Phi_{\lambda}}(q_1+\eps_2,\tilde p)\dd \tilde p\Bigg|\\
+ \Bigg|\int_{\BR^2}\overline{ \widehat{V\Phi_{\lambda}}(p)} D_h((p_1+\eps_1,\tilde p),q_1+\eps_2) (\widehat{V\Phi_{\lambda}}(q_1+\eps_2,\tilde p)-\widehat{V\Phi_{\lambda}}(q_1,\tilde p)) \dd \tilde p \dd k \Bigg|\\
+ \Bigg|\int_{\BR^2}\overline{ \widehat{V \Phi_{\lambda}}(p)} (D_h((p_1+\eps_1,\tilde p),q_1+\eps_2)-D_h(p,q_1))\widehat{V\Phi_{\lambda}}(q_1,\tilde p) \dd \tilde p\Bigg| \Bigg]
\end{multline}
Observe that
\begin{equation}
|\widehat{V \Phi_{\lambda}}(p_1+\eps_1,\tilde p)- \widehat{V \Phi_{\lambda}}(p)|
 \leq \frac{1}{(2\pi)^{d/2}}\int_{\BR^d} | e^{i \eps_1 r_1}-1| |V(r)||\Phi_{\lambda}(r)| \dd r\\
 \leq \frac{\eps_1 \lVert\Phi_{\lambda}\rVert_\infty \lVert |\cdot | V \rVert_1}{(2\pi)^{d/2}}
\end{equation}
With Lemma~\ref{pf_e0_6} and Lemma~\ref{phi_infinity_bd}, we bound the first two terms in \eqref{pf_e0_4} by $C \eps_1$ and $C\eps_2$, respectively.
Hence they vanish as $\eps_1,\eps_2\to0$.
The absolute value of the integrand in the last term in \eqref{pf_e0_4} is bounded by $\lVert \widehat{V \Phi_{\lambda}}\rVert_\infty \frac{2C}{1+\tilde p^2} \widehat{V\Phi_{\lambda}}(q_1,\tilde p) $.
By Lemma~\ref{pf_e0_6}, this is an $L^1$ function.
Hence, when taking the limit $\eps_1,\eps_2\to0$, we are allowed to pull the limit into the integral by dominated convergence, showing that also the last term vanishes.
Therefore, the functions $f$ and $g$ are continuous.
\end{proof}

\begin{proof}[Proof of Lemma~\ref{Bq-B0-bound}]
This Lemma is a generalization of \cite[Lemma 3.2]{hainzl_boundary_nodate} and its proof follows the same ideas.
For $\vert q \vert>1$, Lemma~\ref{KT-Laplace} implies the bound $\vert k(p,q) \vert < \frac{C}{1+p^2}$.
For $\vert q \vert<1$, we use the partial fraction expansion (see \cite[(2.2)]{hainzl_boundary_nodate})
\begin{equation}
k(p,q)=2T\sum_{n\in \BZ} \frac{\vert q \vert (2\mu  -q^2-2p^2+4  (p \cdot \frac{q}{\vert q \vert})^2)- 4 i w_n p \cdot \frac{q}{\vert q \vert}}{\left(\left(p+q\right)^2-\mu-iw_n\right)\left(\left(p-q\right)^2-\mu+iw_n\right)\left(p^2-\mu-iw_n\right)\left(p^2-\mu+iw_n\right)}
\end{equation}
where $w_n=(2n+1)\pi T$.
Continuity of $k$ follows e.g.~using the Weierstrass M-test.
Noting that $w_n=-w_{-n-1}$, it is easy to see that $k(p,0)=0$.

With the estimates
\begin{multline}
\sup_{(p,q)\in \BR^{2d}, \vert q \vert<1} \left \vert\frac{\vert q \vert (2\mu  -q^2-2p^2+4  (p \cdot \frac{q}{\vert q \vert})^2)}{ \left(\left(p+q\right)^2-\mu-iw_n\right)\left(\left(p-q\right)^2-\mu+iw_n\right)}\right \vert\\
\leq \sup_{(p,q)\in \BR^{2d}, \vert q \vert<1} \frac{ \vert q \vert (2\mu  +q^2+6p^2) }{\sqrt{\left[\left(p+q\right)^2-\mu\right]^2+w_0^2}\sqrt{\left[\left(p-q\right)^2-\mu\right]^2+w_0^2}}=:c_1 <\infty
\end{multline}
and
\begin{multline}
\sup_{(p,q)\in \BR^{2d}, \vert q \vert<1}\left \vert \frac{4 i w_n p}{\left(\left(p+q\right)^2-\mu-iw_n\right)\left(\left(p-q\right)^2-\mu+iw_n\right)}\right \vert\\
\leq \sup_{(p,q)\in \BR^{2d}, \vert q \vert<1} \frac{ 4\vert p\vert }{\sqrt{\left[\left( p  + q\right)^2-\mu\right]^2+w_0^2}}=:c_2 <\infty
\end{multline}
one obtains
\begin{equation}
\vert k(p,q) \vert \leq  2T(c_1+c_2)\sum_{n\in \BZ} \frac{1}{\left(p^2-\mu \right)^2+w_n^2}
\end{equation}
Using that the summands are decreasing in $n$, we can estimate the sum by an integral
\begin{multline}
\vert k(p,q) \vert \leq4T(c_1+c_2)\left[ \frac{1}{\left(p^2-\mu \right)^2+w_0^2} +\int_{1/2}^\infty  \frac{1}{\left(p^2- \mu\right)^2+4\pi^2T^2 x^2} \dd x \right]\\
=4T(c_1+c_2) \left[ \frac{1}{\left(p^2-\mu \right)^2+w_0^2} +\frac{\arctan \left(\frac{\vert p^2- \mu\vert}{\pi T}\right)}{2\pi T \vert p^2- \mu\vert}\right] <C \frac{1}{1+p^2}
\end{multline}
for some constant $C$ independent of $p$ and $q$.
\end{proof}

\subsection{Proof of Lemma~\ref{claim_weak_lambda_asy}}\label{pf:claim_weak_lambda_asy}
\begin{proof}[Proof of Lemma~\ref{claim_weak_lambda_asy}]
Recall that $\Psi_{\tcti(\lambda)} =V^{1/2}\Phi_\lambda$ with normalization $\lVert \Psi_{\tcti(\lambda)}  \rVert_2^2=\lVert \Psi  \rVert_2^2=\int_{\BR^d} V(r)j_d(r)^2 \dd r $, where $j_d$ was defined in \eqref{jd}.
Recall from \eqref{aspt_terms} that
\begin{multline}\label{pf_wl_1}
-\frac{1}{2\lambda}\lim_{\epsilon \to 0} \langle \psi_\epsilon , H^1_{\tcti(\lambda),\lambda} \psi_\epsilon \rangle = \int_{\BR^{d+1}} V(r) \vert \Phi_\lambda(z_1, \tilde r)\vert^2 \dd r \dd z_1\\
- \int_{\BR^{d+1}} V(r)\vert \Phi_\lambda(z_1, \tilde r)\mp\Phi_\lambda(r) \vert^2 \chi_{\vert z_1 \vert< \vert r_1 \vert}\dd r \dd z_1
\mp 2\pi \int_{\BR^{d-1}} \overline{ \widehat \Phi_{\lambda}(0,\tilde p)} \widehat{V\Phi_{\lambda}}(0,\tilde p) \dd \tilde p
\end{multline}
The claim follows, if we prove that the right hand side is positive in the limit $\lambda\to 0$.
For $d\in \{1,2\}$ we prove that the terms on the second line are bounded and the first term diverges as $\lambda\to0$.
For $d=3$ the first term is bounded too, so we need to compute the limit of all terms.
The idea is that in the limit, one would like to replace $\Phi_\lambda$ by $j_3$ using Lemmas~\ref{psiT_asymptotic_d} and \ref{phi_infinity_conv}.
We consider each of the three summands in \eqref{pf_wl_1} separately.

{\bf Second term:}
The second term is bounded by $4\lVert |\cdot | V \rVert_1 \lVert \Phi_\lambda \rVert_\infty^2$, which is bounded for small $\lambda$ by Lemma~\ref{phi_infinity_bd}.
For $d=3$ we want to compute the limit.
By Lemma~\ref{phi_infinity_conv} the integrand is bounded by $8 |V(r)| \lVert j_3 \rVert_\infty^2 \chi_{|z_1|<|r_1|}$ for $\lambda$ small enough, which is integrable.
By dominated convergence, the term thus converges to
\begin{equation}
-\int_{\BR^4} V(r) \vert j_3(z_1,\tilde r) \mp j_3(r) \vert^2 \chi_{|z_1|<|r_1|} \dd r \dd z_1.
\end{equation}

{\bf Third term:}
Using \eqref{eval_eq} the third term in \eqref{pf_wl_1} equals
\begin{equation}\label{pf_wl_8}
 \mp  2\pi \lambda \int_{\BR^{d-1}} \vert \widehat{V^{1/2} \Psi_{\tcti(\lambda)} }(0,\tilde p)\vert^2 B_{\tcti(\lambda)} ((0,\tilde p ),0) \dd \tilde p
\end{equation}
For $d=1$, this is bounded by $2\pi \lambda B_{\tcti(\lambda) } (0,0)  \lVert \widehat{V^{1/2} \Psi_{\tcti(\lambda) } }\rVert_\infty^2$.
By Lemma~\ref{lp_prop}\eqref{lp_prop.3} and since $\sup_T B_{T}(0,0)=\frac{1}{\mu}$, this is $O(\lambda)$ as $\lambda\to0$.
For $d=2$ we use \eqref{BT_bound} to bound \eqref{pf_wl_8} by
\begin{equation}\label{pf_wl_5}
2\pi \lambda \int_{\vert \tilde p \vert^2 <2\mu} B_{\tcti(\lambda)}((0,\tilde p ),0) \dd \tilde p  \lVert \widehat{V^{1/2} \Psi_{\tcti(\lambda)} }\rVert_\infty^2
+C\lambda \int_{\vert \tilde p \vert^2 >2\mu}\frac{1}{1+\tilde p^2}\dd \tilde p \lVert \widehat{V^{1/2} \Psi_{\tcti(\lambda) }}\rVert_\infty^2,
\end{equation}
where $C$ is independent of $\lambda$.
By Lemma~\ref{lp_prop}\eqref{lp_prop.3} $ \lVert \widehat{V^{1/2} \Psi_{\tcti(\lambda)} }\rVert_\infty$ is bounded as $\lambda\to0$.
The second term in \eqref{pf_wl_5} thus vanishes as $\lambda\to 0$.
For the first term, recall from \eqref{mmu} that $\int_{\vert \tilde p \vert^2 <2\mu} B_{\tcti,\mu}((0,\tilde p ),0) \dd \tilde p = 2 \pi m_\mu^{d=2}(\tcti(\lambda))$.
By Lemma~\ref{lea:Tc0_asy} the first term is bounded for small $\lambda$.
For $d=3$, we rewrite \eqref{pf_wl_8} as
\begin{multline}\label{pf_wl_7}
\mp  2\pi \lambda \int_{\tilde p^2>2\mu} \vert \widehat{V^{1/2} \Psi_{\tcti(\lambda)}}(0,\tilde p)\vert^2 B_{\tcti} ((0,\tilde p ),0) \dd \tilde p\\
\mp  \lambda \int_{\tilde p^2<2\mu}\int_{\BR^6} \overline{V^{1/2}\Psi_{\tcti(\lambda)}(x) } \frac{e^{i\tilde p\cdot(\tilde x-\tilde y)}-e^{i\sqrt{\mu} \frac{\tilde p}{|\tilde p|}\cdot(\tilde x-\tilde y)}}{(2\pi)^2}B_{\tcti}((0,\tilde p ),0)V^{1/2} \Psi_{\tcti(\lambda)}(y)\dd x \dd y\dd \tilde p\\
\mp  \lambda \int_{\tilde p^2<2\mu}\int_{\BR^6} \overline{\left(V^{1/2}\Psi_{\tcti(\lambda)}(x)-Vj_3(x) \right)} \frac{e^{i\sqrt{\mu} \frac{\tilde p}{|\tilde p|}\cdot(\tilde x-\tilde y)}}{(2\pi)^2}B_{\tcti}((0,\tilde p ),0) V^{1/2} \Psi_{\tcti(\lambda)} (y)\dd x \dd y\dd \tilde p\\
\mp  \lambda \int_{\tilde p^2<2\mu}\int_{\BR^6} V(x)j_3(x) \frac{e^{i\sqrt{\mu} \frac{\tilde p}{|\tilde p|}\cdot(\tilde x-\tilde y)}}{(2\pi)^2}B_{\tcti}((0,\tilde p ),0) \left(V^{1/2} \Psi_{\tcti(\lambda)}(y)-Vj_3(y) \right) \dd x \dd y\dd \tilde p\\
\mp  \lambda \int_{\tilde p^2<2\mu}\int_{\BR^6} V(x)j_3(x)  \frac{e^{i\sqrt{\mu} \frac{\tilde p}{|\tilde p|}\cdot(\tilde x-\tilde y)}}{(2\pi)^2}B_{\tcti} ((0,\tilde p ),0) V(y)j_3(y)\dd x \dd y\dd \tilde p.
\end{multline}
We prove that the first four integrals vanish as $\lambda \to 0$ and compute the limit of the expression in the last line.

Using \eqref{BT_bound}, Lemma~\ref{lp_prop}\eqref{lp_prop.3} and $ \Psi_{\tcti(\lambda)}=V^{1/2}\Phi_{\lambda}$ the first term in \eqref{pf_wl_7} is bounded by
\begin{equation}
 C \lambda \lVert V\rVert_1^{1/2} \lVert \Psi_{\tcti(\lambda)}\rVert_2   \left\lVert \frac{1}{1+\vert \cdot \vert^2}\widehat{V\Phi_{\lambda}}(0,\cdot)\right\rVert_{L^1(\BR^2)}
\end{equation}
where $C$ is independent of $\lambda$.
By \eqref{pf_e0_7},
\begin{equation}
  \left\lVert \frac{1}{1+\vert \cdot \vert^2}\widehat{V\Phi_{\lambda}}(0,\cdot)\right\rVert_{L^1(\BR^2)}\leq  \left \Vert \frac{1}{1+\vert \cdot \vert^2} \right \Vert_{L^{3/2}(\BR^{2})} \sup_{k_1} \lVert \widehat{V}(k_1,\cdot)\rVert_3  \lVert \widehat \Phi_{\lambda}\rVert_1
\end{equation}
By Lemma~\ref{lp_prop}\eqref{lp_prop.4}, $\sup_{k_1} \lVert \widehat{V}(k_1,\cdot)\rVert_3 <\infty$.
Furthermore $\lVert \widehat \Phi_\lambda \rVert_1 $ is bounded uniformly in $\lambda$ by Lemma~\ref{phi_infinity_bd}.
In total, the first term in \eqref{pf_wl_7} is $O(\lambda)$ as $\lambda\to0$.

For the second line of \eqref{pf_wl_7} we use that
\begin{equation}
\sup_{\lambda>0} \sup_{\tilde x,\tilde y \in \BR^2} \left \vert \int_{\BR^{2},\tilde p^2<2\mu}\frac{e^{i\tilde p\cdot(\tilde x-\tilde y)}-e^{i\sqrt{\mu} \frac{\tilde p}{|\tilde p|}\cdot(\tilde x-\tilde y)}}{(2\pi)^3}B_{\tcti(\lambda)}((0,\tilde p ),0)\dd \tilde p \right \vert <\infty,
\end{equation}
as was shown in the proof of \cite[Lemma 3.4]{henheik_universality_2023}.
Applying the Schwarz inequality, the second line is bounded by $C \lambda \lVert V \rVert_1\lVert \Psi_{\tcti(\lambda)} \rVert_2^2$ for some constant $C$ and vanishes for $\lambda\to0$.

We bound the third line of \eqref{pf_wl_7} by
\begin{multline}
 \frac{\lambda}{(2\pi)^2} \int_{\BR^{2},\tilde p^2<2\mu}\int_{\BR^6} \vert \overline{\left(V^{1/2} \Psi_\lambda(x)-Vj_3(x) \right)}\vert B_{\tcti(\lambda)}((0,\tilde p ),0) \vert V^{1/2} \Psi_{\tcti(\lambda)} (y) \vert \dd x \dd y\dd \tilde p\\
\leq \lambda\frac{\vert \BS^1 \vert }{(2\pi)^2} m_\mu^{d=2}(\tcti(\lambda))\lVert V \rVert_1 \lVert \Psi_{\tcti(\lambda)} \rVert_2\lVert \Psi_{\tcti(\lambda)}-\Psi\rVert_2,
\end{multline}
where in the second step we carried out the $\tilde p $ integration and used the Schwarz inequality in $x$ and $y$.
By Lemma~\ref{lea:Tc0_asy}, $\lambda m_\mu^{d=2}(\tcti(\lambda))$ is bounded and by Lemma~\ref{psiT_asymptotic_d}, $\lVert \Psi_{\tcti(\lambda)} -\Psi \rVert_2$ decays like $\lambda^{1/2}$.
Hence, this vanishes for $\lambda\to0$.
Similarly, the fourth integral in \eqref{pf_wl_7} is bounded by
\begin{equation}
\lambda\frac{\vert \BS^1 \vert }{(2\pi)^2} m_\mu^{d=2}(\tcti(\lambda))\lVert V \rVert_1 \lVert V^{1/2}j_3  \rVert_2\lVert \Psi_{\tcti(\lambda)} -\Psi \rVert_2,
\end{equation}
which vanishes for $\lambda\to0$.

For the last line of \eqref{pf_wl_7} we first carry out the integration over $x,y$ and the radial part of $\tilde p$, and then use that $\widehat{Vj_3}$ is a radial function. This way we obtain
\begin{equation}
\mp  \lambda m_\mu^{d=2}(\tcti(\lambda)) 2\pi \int_{\BS^{1}}\vert \widehat{Vj_3}(0,\sqrt{\mu}w)\vert^2\dd w=\mp  \lambda m_\mu^{d=2}(\tcti(\lambda)) \pi \int_{\BS^{2}}\vert \widehat{Vj_3}(\sqrt{\mu}w)\vert^2\dd w
\end{equation}
The latter integral equals $\langle \vert V\vert^{1/2}j_3,O_\mu V^{1/2} j_3\rangle = e_\mu \int_{\BR^3} V(x)j_3(x)^2 \dd x$.
By Lemma~\ref{lea:Tc0_asy},
\begin{equation}
\lim_{\lambda\to 0}\lambda m_\mu^{d=2}(\tcti(\lambda))e_\mu =\lim_{\lambda\to 0} \lambda \ln(\mu/\tcti(\lambda))e_\mu =\frac{1}{\mu^{1/2}}.
\end{equation}
Therefore, the limit of the last line of \eqref{pf_wl_7} for $\lambda\to 0$ equals
\begin{equation}
\mp  \frac{\pi}{\mu^{1/2}} \int_{\BR^3} V(x)j_3(x)^2 \dd x.
\end{equation}

{\bf First term:}
It remains to consider the first term in \eqref{pf_wl_1}.
If $V\geq 0$, one could argue directly using the convergence of $\Phi_\lambda$ in Lemma~\ref{phi_infinity_conv} for $d=3$.
However, the analogue of Lemma~\ref{phi_infinity_conv} does not hold for $d=1$.
Instead, the strategy is to use the $L^2$-convergence of the ground state in the Birman-Schwinger picture, Lemma~\ref{psiT_asymptotic_d}.
This approach also allows us to treat $V$ that take negative values.

Switching to momentum space and using the eigenvalue equation \eqref{eval_eq}, we rewrite the first term in \eqref{pf_wl_1} as
\begin{equation}\label{pf_wl_2}
(2\pi)^{1-\frac{d}{2}}\int_{\BR^{2d-1}} \overline{\widehat \Phi_\lambda (p)} \widehat V(0,\tilde p-\tilde q)\widehat\Phi_\lambda (p_1,\tilde q) \dd p \dd \tilde q
=(2\pi)^{1-\frac{d}{2}} \lambda^2\langle \Psi_{\tcti(\lambda)}, D_{\tcti(\lambda)} \Psi_{\tcti(\lambda)} \rangle ,
\end{equation}
where $D_T$ is the operator given by
\begin{equation}
\langle \psi, D_T \psi \rangle=\int_{\BR^{2d-1}} \overline {\widehat{|V|^{1/2} \psi} (p)}  B_{T} (p,0) \widehat V(0,\tilde p-\tilde q) B_{T}((p_1,\tilde q),0)\widehat{|V|^{1/2} \psi}(p_1,\tilde q) \dd p \dd \tilde q
\end{equation}
for $\psi \in L^2(\BR^d)$.
We decompose \eqref{pf_wl_2} as
\begin{multline}\label{pf_wl_6}
(2\pi)^{1-\frac{d}{2}} \lambda^2\langle \Psi_{\tcti(\lambda)}, D_{\tcti(\lambda)} \Psi_{\tcti(\lambda)} \rangle = (2\pi)^{1-\frac{d}{2}} \lambda^2 \Big(\langle \Psi_{\tcti(\lambda)}-\Psi, D_{\tcti(\lambda)} \Psi_{\tcti(\lambda)} \rangle  \\
+\langle \Psi, D_{\tcti(\lambda)} ( \Psi_{\tcti(\lambda)}- \Psi) \rangle+\langle \Psi, D_{\tcti(\lambda)} \Psi \rangle  \Big).
\end{multline}
Recall that by Lemma~\ref{psiT_asymptotic_d}, $\lVert  \Psi_{\tcti} - \Psi\rVert_2=O(\lambda^{1/2})$.
The strategy is to prove that $\lVert D_T \rVert$ and $\langle \Psi, D_{T} \Psi \rangle$ are of the same order for $T\to 0$.
Then, the positive term $\langle \Psi, D_{\tcti(\lambda)} \Psi \rangle $ will be the leading order term in \eqref{pf_wl_6} as $\lambda\to0$.
The asymptotic behavior of  $\lVert D_T \rVert$ and $\langle \Psi, D_{T} \Psi \rangle$ is the content of the following two Lemmas.
These asymptotics strongly depend on the dimension and this is where the different treatment of $d=3$ versus $d\in\{1,2\}$ in Theorem~\ref{thm1} originates.

It will be convenient to introduce the operator $D_T^{<}$ as
\begin{equation}\label{DT<}
\langle \psi, D_T^{<} \psi \rangle=\int_{\vert p|^2<2\mu, |(p_1,\tilde q)|^2<2\mu, p_1^2<\mu} \overline {\widehat{|V|^{1/2} \psi} (p)}  B_{T} (p,0) \widehat V(0,\tilde p-\tilde q) B_{T}((p_1,\tilde q),0)\widehat{|V|^{1/2} \psi}(p_1,\tilde q) \dd p \dd \tilde q
\end{equation}
for $\psi \in L^2(\BR^d)$.
Furthermore, for $d=2$ we define for $0<\delta<\mu$ the operator $D_T^\delta$ as
\begin{equation}\label{DTd}
\langle \psi, D_T^{\delta} \psi \rangle=\int_{ \mu-\delta<p_1^2<\mu,p_2^2<2\delta, q_2^2<2\delta} \overline {\widehat{|V|^{1/2} \psi} (p)}  B_{T} (p,0) \widehat V(0,p_2- q_2) B_{T}((p_1,\tilde q),0)\widehat{|V|^{1/2} \psi}(p_1, q_2) \dd p \dd q_2
\end{equation}
for $\psi \in L^2(\BR^2)$.
\begin{lemma}\label{DTmax}
Let $\mu>\delta>0$ and let $V$ satisfy~\ref{aspt_V_halfspace}.
There are constants $C,T_0>0$ such that for all $0<T<T_0$ for $d=1$ $\lVert D_T \rVert \leq C/T$, for $d=2$ $\lVert D_T \rVert \leq C(\ln \mu/T)^3$ and $\lVert D_T-D_T^\delta \rVert \leq C (\ln \mu/T)^2$, and for $d=3$ $\lVert D_T \rVert \leq C(\ln \mu/T)^2$ and $\lVert D_T-D_T^< \rVert \leq C \ln \mu/T$.
\end{lemma}
\begin{lemma}\label{DTj}
Let $\mu>0$ and let $V$ satisfy~\ref{aspt_V_halfspace}.
Recall that $\Psi=V^{1/2}j_d$.
There are constants $C,T_0>0$ such that for all $0<T<T_0$, $\langle \Psi, D_T  \Psi \rangle \geq C/T$ for $d=1$ and $\geq C(\ln \mu/T)^3$ for $d=2$.
For $d=3$, $\lim_{\lambda \to 0} (2\pi)^{-1/2} \lambda^2\langle \Psi, D_{\tcti(\lambda)} \Psi \rangle = \int_{\BR^4} V(r) j_3(z_1,\tilde r;\mu)^2 \dd r \dd z_1$.
\end{lemma}
For $\lambda \to 0$, by Lemma~\ref{lea:Tc0_asy}, $\ln(\mu/\tcti(\lambda))$ is of order $1/\lambda$, hence the last term in \eqref{pf_wl_6} diverges for $d=1,2$.
For $d=3$ we get the desired constant by Lemma~\ref{DTj}.
\end{proof}

\begin{proof}[Proof of Lemma~\ref{DTmax}]
Assume that $T/\mu<1/2$.
We treat the different dimensions $d$ separately.

{\bf Dimension one:} Note that
\begin{equation}
|\langle \psi, D_T \psi \rangle |= |\widehat{V}(0)| \int_\BR B_{T} (p,0)^2 \vert \widehat{|V|^{1/2} \psi }(p) \vert^2 \dd p \leq \lVert V \rVert_1^2 \int_\BR B_{T} (p,0)^2 \dd p \lVert \psi \rVert_2^2,
\end{equation}
 where we used Lemma~\ref{lp_prop}.
Recall from \eqref{BT_bound} that $B_{T} (p,0) \leq \min\left\{\frac{1}{\vert p^2-\mu \vert}, \frac{1}{2T}\right\}$.
We estimate the integral
\begin{multline}
\int_\BR B_{T} (p,0)^2 \dd p \leq  \int_{\sqrt{\mu}-\frac{T}{\sqrt{\mu}}< \vert p \vert<\sqrt{\mu}+\frac{T}{\sqrt{\mu}} }\frac{1}{4T^2} \dd p+\int_{\BR} \frac{\chi_{\vert p \vert< \sqrt{\mu}-\frac{T}{\sqrt{\mu}}}+\chi_{  \sqrt{\mu}+\frac{T}{\sqrt{\mu}}<p<2\sqrt{\mu}}}{\mu  (\vert p \vert-\sqrt{\mu} )^2} \dd p \\
+ \int_{ p>2\sqrt{\mu}} \frac{1}{( p^2-\mu )^2} \dd p
\end{multline}
The first term equals $(\sqrt{\mu}T)^{-1}$.
The last term is a finite constant independent of $T$.
In the second term we substitute $\vert \vert p \vert -\sqrt{\mu}\vert$ by $x$ and get the bound
\begin{equation}
2 \int_{\frac{T}{\sqrt{\mu}}}^{\sqrt{\mu}} \frac{1}{\mu x^2 } \dd x = \frac{2}{\sqrt{\mu}}(1/T-1/\mu)
\end{equation}

{\bf Dimension two:}
Using the Schwarz inequality we have
\begin{equation}
\langle \psi, D_T \psi \rangle\leq C \lVert V\rVert_1^2 \int_{\BR^{3}}  B_{T,\mu} (p,0) B_{T,\mu}((p_1,\tilde q),0)\dd p \dd \tilde q \lVert \psi\rVert_2^2
\end{equation}
The integral can be rewritten as
\begin{equation}\label{2dBB}
\int_\BR \left(\int_\BR B_{T, \mu-p_1^2}(p_2, 0)  \dd p_2 \right)^2 \dd p_1,
\end{equation}
where $B_{T,\mu}$ here is understood as the function on $\BR \times \BR$ instead of $\BR^2\times \BR^2$.
Similarly,
\begin{multline}\label{2dBBd}
\vert\langle \psi, (D_T-D_T^\delta )\psi \rangle \vert\leq C \lVert V\rVert_1^2 \int_{\BR^{3}}(1-\chi_{\mu-\delta<p_1^2<\mu} \chi_{p_2^2<2\delta}\chi_{p_2'^2<2\delta})  B_{T,\mu} (p,0) B_{T,\mu}((p_1,\tilde q),0)\dd p \dd \tilde q \lVert \psi\rVert_2^2
\end{multline}
We prove that \eqref{2dBB} and \eqref{2dBBd} are of order $O(\ln(\mu/T)^3)$ and $O(\ln(\mu/T)^2)$ for $T\to0$, respectively.
To bound the integrals we consider three regimes, $p_1^2<\mu-T$, $\mu-T<p_1^2<\mu+T$, and $\mu+T<p_1^2$.
Corresponding to these regimes, we need to understand $\int_\BR B_{T, \mu}(p, 0)  \dd p$ for $T/\mu<1$, $-1<\mu/T <1$, and $\mu/T<-1$.

In the first regime, there is a constant $C_1$, such that for all $T/\mu<1$
\begin{equation}\label{intB_est_1}
\left | \sqrt{\mu} \int_{\BR} B_{T,\mu}(p,0)\chi_{p^2<2\mu} \dd p - 2 \ln \frac{\mu}{T} \right |+\left | \sqrt{\mu} \int_{\BR} B_{T,\mu}(p,0)\chi_{p^2>2\mu} \dd p \right |\leq C_1
\end{equation}
This follows from rescaling $ \sqrt{\mu} \int_\BR B_{T,\mu}(p,0) \dd p = \int_\BR B_{T/\mu,1} (p,0) \dd p$ and applying \cite[Lemma 3.5]{hainzl_boundary_nodate}.
For the second regime, we rewrite
\begin{equation}\label{intB_est_4}
\int_\BR B_{T,\mu}(p,0) \dd p = \frac{1}{\sqrt{T}}\int_\BR \frac{\tanh((p^2-\mu/T)/2)}{p^2-\mu/T} \dd p
\end{equation}
Since $\tanh (x)/x \leq \min\{1,1/\vert x\vert\}$ the latter integral is uniformly bounded for $\vert \mu/T \vert<1$,
\begin{equation}\label{intB_est_2}
\int_\BR B_{T,\mu}(p,0) \dd p \leq \frac{C_2}{\sqrt{T}}.
\end{equation}
For the third regime, it follows from \eqref{intB_est_4} that
\begin{equation}
\int_\BR B_{T,\mu}(p,0) \dd p\leq  \frac{1}{\sqrt{T}}\int_\BR \frac{1}{p^2-\mu/T} \dd p= \frac{1}{\sqrt{-\mu}} \int_\BR \frac{1}{p^2+1} \dd p=:\frac{C_3}{\sqrt{-\mu}}.
\end{equation}
Combining the bounds in the three regimes, we bound \eqref{2dBB} from above by
\begin{equation}\label{intB_est_5}
\int_{\vert p_1 \vert<\sqrt{\mu-T}} \frac{\left(2 \ln \left(\frac{\mu-p_1^2}{T}\right)+C_1\right)^2}{\mu-p_1^2} \dd p_1+\int_{\sqrt{\mu-T}<\vert p_1 \vert<\sqrt{\mu+T}} \frac{C_2^2}{T} \dd p_1+\int_{\sqrt{\mu+T}<\vert p_1 \vert} \frac{C_3^2}{p_1^2-\mu} \dd p_1
\end{equation}
The first integral is bounded above by
\begin{equation}
\left(2 \ln \left(\frac{\mu}{T}\right)+C_1\right)^2 \int_{\vert p_1 \vert<\sqrt{\mu-T}} \frac{1}{\mu-p_1^2} \dd p_1.
\end{equation}
Since
\begin{equation}\label{intB_est_6}
\int_{\vert p_1 \vert<\sqrt{\mu-T}} \frac{1}{\mu-p_1^2} \dd p_1=\frac{1}{\sqrt{\mu}} \ln \left(\frac{2\mu-T+\sqrt{\mu(\mu-T)}}{T}\right)=O(\ln (\mu/T)),
\end{equation}
the first integral in \eqref{intB_est_5} is of order $O(\ln (\mu/T)^3)$.
In the second integral, the size of the integration domain is $2T/\sqrt{\mu}+O(T^2)$, so the integral is bounded as $T\to 0$.
The third integral equals
\begin{equation}
\frac{C_3^2}{\sqrt{\mu}} \ln \left(\frac{2\mu+T+\sqrt{\mu(\mu+T)}}{T}\right) =O(\ln \mu/T).
\end{equation}
In total \eqref{2dBB} is of order $O(\ln (\mu/T)^3)$.

For the integral in \eqref{2dBBd} we obtain the upper bound similar to \eqref{intB_est_5}.
The main difference is that in the regime $\sqrt{\mu-\delta}<\vert p_1 \vert<\sqrt{\mu-T}$, at least one of the variables $p_2,p_2'$ is constrained to absolute values larger than $\sqrt{2\delta}\geq \sqrt{2(\mu-p_1^2})$, and thus for the integration over this variable there will be no $\ln \left(\frac{\mu-p_1^2}{T}\right)$ contribution from \eqref{intB_est_1}.
The upper bound for \eqref{2dBBd} is
\begin{multline}\label{intB_est_5d}
\int_{\vert p_1 \vert<\sqrt{\mu-\delta}} \frac{\left(2 \ln \left(\frac{\mu-p_1^2}{T}\right)+C_1\right)^2}{\mu-p_1^2} \dd p_1
+\int_{\sqrt{\mu-\delta}<\vert p_1 \vert<\sqrt{\mu-T}}  \frac{2\left(2 \ln \left(\frac{\mu-p_1^2}{T}\right)+C_1\right) C_1}{\mu-p_1^2} \dd p_1\\
+\int_{\sqrt{\mu-T}<\vert p_1 \vert<\sqrt{\mu+T}} \frac{C_2^2}{T} \dd p_1
+\int_{\sqrt{\mu+T}<\vert p_1 \vert} \frac{C_3^2}{p_1^2-\mu} \dd p_1
\end{multline}
We have already seen above that the last two integrals are of order $O(1)$ and $O(\ln \mu/T)$, respectively.
The first integral in \eqref{intB_est_5d} is bounded above by
$\left(2 \ln \left(\frac{\mu}{T}\right)+C_1\right)^2 \int_{\vert p_1 \vert<\sqrt{\mu-\delta}} \frac{1}{\mu-p_1^2} \dd p_1=O(\ln (\mu/T)^2).$
Similarly, the second integral in \eqref{intB_est_5d} is of order $O(\ln (\mu/T)^2)$ by \eqref{intB_est_6}.

{\bf Dimension three:}
For $d=3$, we first prove that $\lVert D_T^<\rVert=O(\ln(\mu/T)^2)$.
We bound \eqref{DT<} using the Schwarz inequality
\begin{equation}\label{intB_est_3d.1}
\langle \psi, D_T^< \psi \rangle\leq \lVert V \rVert^2_1 \lVert \psi\rVert_2^2 \int_{\BR^{5}}\chi_{|p|^2<2\mu, |(p_1,\tilde q)|<2\mu, p_1^2<\mu} B_{T,\mu} (p,0) B_{T,\mu}((p_1,\tilde q),0)\dd p \dd \tilde q.
\end{equation}
The integral can be rewritten as
\begin{equation}\label{3dBB}
4\pi^2 \int_0^{\sqrt{\mu}} \left(\int_{0}^{\sqrt{2\mu-p_1^2}} B_{T, \mu-p_1^2}(t, 0) t \dd t \right)^2 \dd p_1
\end{equation}
Substituting $s=(t^2+p_1^2-\mu)/T$ gives
\begin{equation}\label{3dBB.1}
\pi^2 \int_0^{\sqrt{\mu}} \left(\int_{-(\mu-p_1^2)/T}^{\mu/T} \frac{\tanh(s)}{s} \dd s \right)^2 \dd p_1 \leq \sqrt{\mu}\pi^2 \left(\int_{-\mu/T}^{\mu/T} \frac{\tanh(s)}{s} \dd s \right)^2
\end{equation}
Since $\tanh (x)/x \leq \min\{1,1/\vert x\vert\}$, this is bounded by
\begin{equation}\label{3dBB.2}
\sqrt{\mu}4\pi^2 \left(1+\ln(\mu/T) \right)^2.
\end{equation}

To bound $\Vert D_T-D_T^<\rVert$, we distinguish the cases were $p^2$ and $ (p_1,\tilde q)^2$ are larger or smaller than $2\mu$.
Using the bound on $B_T$ given in \eqref{BT_bound} we estimate
\begin{multline}\label{intB_est_3d}
\vert\langle \psi, (D_T-D_T^< )\psi \rangle\vert\leq \lVert V \rVert^2_1 \lVert \psi\rVert_2^2 \int_{\BR^{5}}\chi_{|p|^2<2\mu, |(p_1,\tilde q)|<2\mu,p_1^2>\mu} B_{T,\mu} (p,0) B_{T,\mu}((p_1,\tilde q),0)\dd p \dd \tilde q\\
+2\lVert V \rVert_1 \lVert \psi\rVert_2^2 \int_{\BR^{5}}  \frac{C}{\tilde p^2+1}\vert \widehat V(0,\tilde p-\tilde q)\vert B_{T,\mu}((p_1,\tilde q),0)\chi_{|(p_1,\tilde q)|^2<2\mu}\dd p \dd \tilde q\\
+\int_{\BR^{5}}  \overline {\widehat{|V|^{1/2} \psi} (p)}  \frac{C}{p^2+1}\vert \widehat V(0,\tilde p-\tilde q) \vert\frac{C }{p_1^2+\tilde q^2+1}\widehat{|V|^{1/2} \psi}(p_1,\tilde q) \dd p \dd \tilde q,
\end{multline}
where $C$ is a constant independent of $T$.
For the first term, proceeding similarly to \eqref{3dBB}--\eqref{3dBB.2}, the integral equals
\begin{equation}
\pi^2 \int_{\sqrt{\mu}}^{\sqrt{2\mu}} \left(\int_{(p_1^2-\mu)/T}^{\mu/T} \frac{\tanh(s)}{s} \dd s \right)^2 \dd p_1 \leq \pi^2 \int_{\sqrt{\mu}}^{\sqrt{2\mu}} \ln \left(\frac{\mu}{p_1^2-\mu}\right)^2 \dd p_1<\infty
\end{equation}
For the second term in \eqref{intB_est_3d} we apply Young's inequality to bound the integral by
\begin{equation}
C \left\Vert \frac{1}{1+|\cdot|^2}\right \Vert_{L^{3/2}(\BR^2)} \lVert \widehat V(0,\cdot) \rVert_{L^3(\BR^2)} |\BS^2| m_\mu(T)
\end{equation}
which is $O(\ln \mu/T)$.
The third term in \eqref{intB_est_3d} is bounded by $C \lVert \psi \rVert_2^2$ by Lemma~\ref{lea:pf_plugin_1}.
\end{proof}

\begin{proof}[Proof of Lemma~\ref{DTj}]
By assumption, $0<e_\mu= \frac{1}{(2\pi)^{d/2}}\int_{\BS^{d-1}} \widehat V(p- \sqrt{\mu} \omega) \dd \Omega(\omega)=\widehat{V j_d}(\vert p\vert=\sqrt{\mu})$.
By continuity of $\widehat{Vj_d}(p)$ in $p$, there is an $\eps>0$ such that $\widehat{V j_d}(p)>\frac{1}{2}\widehat{V j_d}(\vert p\vert=\sqrt{\mu})>0$ for all $\sqrt{\mu}-\eps < \vert p\vert< \sqrt{\mu}+\eps$.
In the following we treat the different dimensions separately.

{\bf Dimension one:}
Suppose $T<\eps$.
Since $\widehat{V}(0) >0$,
\begin{equation}
\langle V^{1/2} j_1, D_T V^{1/2} j_1\rangle = \widehat{V}(0) \int_\BR B_{T} (p,0)^2 \vert \widehat{V j_1}(p) \vert^2 \dd p \geq  \frac{1}{4}\widehat{V}(0)  \vert \widehat{V j_1}(\sqrt{\mu}) \vert^2 \int_{\sqrt{\mu}+T}^{\sqrt{\mu}+\eps} B_{T} (p,0)^2 \dd p
\end{equation}
For $p\in[\sqrt{\mu}+T,\sqrt{\mu}+\eps]$, $B_{T} (p,0)\geq \frac{\tanh(\sqrt{\mu})}{p^2-\mu} \geq \frac{\tanh(\sqrt{\mu})}{(2\sqrt{\mu}+\eps)(p-\sqrt{\mu})}$.
Since $ \int_{\sqrt{\mu}+T}^{\sqrt{\mu}+\eps} \frac{1}{(p-\sqrt{\mu})^2}\dd p = 1/T-1/\eps$, we obtain the lower bound
\begin{equation}
\langle V^{1/2} j_1, D_T V^{1/2} j_1\rangle \geq  \frac{1}{4}\widehat{V}(0)  \vert \widehat{V j_1}(\sqrt{\mu}) \vert^2 \frac{\tanh(\sqrt{\mu})^2}{(2\sqrt{\mu}+\eps)^2} \left(\frac{1}{T}-\frac{1}{\eps}\right)
\end{equation}
and the claim follows.

{\bf Dimension two:}
Since $\widehat{V}(0)> 0$, by continuity also $\widehat{V}(p)>0$ for small $|p|$.
Therefore, there are constants $0<\delta<\mu$ and $C>0$ such that for all $\sqrt{\mu-\delta}<p_1\leq \sqrt{\mu}$ and $|p_2|,|q_2|<(2\delta)^{1/2}$
\begin{equation}
\overline {\widehat{V j_2} (p_1,p_2)} \widehat V(0,p_2-q_2) \widehat{V j_2}(p_1,q_2) >C.
\end{equation}
By Lemma~\ref{DTmax}, we have  $\langle V^{1/2} j_2, D_T V^{1/2} j_2 \rangle= \langle V^{1/2} j_2, D_T^\delta V^{1/2} j_2 \rangle +O((\ln \mu/T)^2)$.
It hence suffices to show that $\langle V^{1/2} j_2, D_T^\delta V^{1/2} j_2 \rangle$ grows like $(\ln \mu/T)^3$.
Let $A:=\{(p_1,p_2,q_2)\in \BR^3 | \sqrt{\mu-\delta}<p_1<\sqrt{\mu}, 0<p_2,q_2<\delta^{1/2}, p_1^2+p_2^2>\mu+T,p_1^2+q_2^2>\mu+T\}$.
This is a subset of the support in $D_T^\delta$.
Using that all terms in the integrand of $\langle V^{1/2} j_2, D_T^\delta V^{1/2} j_2\rangle$ are positive, we estimate 
\begin{equation}\label{DTj_2d_1}
\langle V^{1/2} j_2, D_T^\delta V^{1/2} j_2 \rangle\geq C \int_{A}  B_{T} (p,0) B_{T}((p_1, q_2),0) \dd p \dd q_2.
\end{equation}
For $(p_1,p_2,q_2)\in A$ we have $ p_1^2+p_2^2-\mu >T$ and thus
\begin{equation}
B_{T} (p,0)\geq \frac{\tanh\left(\frac{1}{2}\right)}{ p_1^2+p_2^2-\mu }
\end{equation}
For $p_1^2>\mu+T-\delta$
\begin{equation}
\int_{\sqrt{\mu+T-p_1^2}}^{\delta^{1/2}} \frac{1}{p_1^2+p_2^2-\mu}\dd p_2 =\frac{1}{\sqrt{\mu-p_1^2}}\left[\artanh\left(\sqrt{1-\frac{T}{\mu+T-p_1^2}}\right)-\artanh\left(\sqrt{\frac{\mu-p_1^2}{\delta}}\right)\right].
\end{equation}
Hence, the integral in \eqref{DTj_2d_1} is bounded below by
\begin{equation}\label{DTj_2d_2}
\tanh\left(\frac{1}{2}\right)^2 \int_{\sqrt{\mu+T-\delta}}^{\sqrt{\mu}} \frac{1}{\mu-p_1^2}\left[\artanh\left(\sqrt{1-\frac{T}{\mu+T-p_1^2}}\right)-\artanh\left(\sqrt{\frac{\mu-p_1^2}{\delta}}\right)\right]^2 \dd p_1
\end{equation}
Assume that $T<\delta/2$.
For a lower bound, we further restrict the $p_1$-integration to the interval $\left(\sqrt{\mu-\delta/2}, \sqrt{\mu-\mu^{1/2}T^{1/2}}\right).$
For these values of $p_1$, we have
\begin{equation}
\artanh\left(\sqrt{\frac{\mu-p_1^2}{\delta}}\right) \leq \artanh\left(\frac{1}{\sqrt{2}}\right) \leq \artanh\left(\sqrt{1-\frac{T^{1/2}}{\mu^{1/2}}}\right)  \leq \artanh\left(\sqrt{1-\frac{T}{\mu+T-p_1^2}}\right).
\end{equation}
Furthermore,
\begin{equation}
\int_{\sqrt{\mu-\delta/2}}^{\sqrt{\mu-\mu^{1/2}T^{1/2}}} \frac{1}{\mu-p_1^2} \dd p_1 = \frac{1}{\sqrt{\mu}}\artanh \left(1-\frac{(\sqrt{\mu}/a+1)(1-b/\sqrt{\mu})}{\sqrt{\mu}/a-b/\sqrt{\mu}}\right),
\end{equation}
where $a=\sqrt{\mu-\delta/2}$ and $b=\sqrt{\mu-\mu^{1/2}T^{1/2}} \leq \sqrt{\mu}$.
This is bounded below by
\begin{equation}
\frac{1}{\sqrt{\mu}}\artanh \left(1-\frac{(\sqrt{\mu}/a+1)(1-b/\sqrt{\mu})}{\sqrt{\mu}/a-1}\right).
\end{equation}
In total, \eqref{DTj_2d_2} is bounded from below by
\begin{multline}\label{DTj_2d_3}
\frac{1}{\sqrt{\mu}}\tanh\left(\frac{1}{2}\right)^2 \left( \artanh\left(\sqrt{1-\frac{T^{1/2}}{\mu^{1/2}}}\right) -\artanh\left(\frac{1}{\sqrt{2}}\right) \right)^2\times \\
\artanh \left(1-\frac{(\sqrt{\mu}/a+1)(1-\sqrt{1-(T/\mu)^{1/2}})}{\sqrt{\mu}/a-1}\right)
\end{multline}
With $\artanh(1-x) = \frac{1}{2} \ln 2/x+o(1)$ as $x\to 0$, we obtain that for $T\to0$
\begin{equation}
\artanh\left(\sqrt{1-\frac{T^{1/2}}{\mu^{1/2}}}\right) =\frac{1}{4}\ln\left(16\frac{\mu}{T}\right) +o(1)
\end{equation}
and
\begin{equation}
\artanh \left(1-\frac{(\sqrt{\mu}/a+1)(1-\sqrt{1-(T/\mu)^{1/2}})}{\sqrt{\mu}/a-1}\right)=\frac{1}{4}\ln\left(16\left(\frac{\sqrt{\mu}/a-1}{\sqrt{\mu}/a+1}\right)^2\frac{\mu}{T}\right) +o(1)
\end{equation}
In particular, we obtain
\begin{equation}
\langle V^{1/2} j_2, D_T V^{1/2} j_2 \rangle \geq \frac{C }{\sqrt{\mu}} \ln\left(\frac{\mu}{T}\right)^3+O\left(\ln\left(\frac{\mu}{T}\right)^2\right)
\end{equation}
for some $C>0$, which implies the claim.

{\bf Dimension three:}
Using that $\lVert D_T-D_T^<\rVert\leq C \ln \mu/T$ according to Lemma~\ref{DTmax} and  that $\ln \mu/\tcti(\lambda)\sim 1/\lambda$ by Lemma~\ref{lea:Tc0_asy},
\begin{equation}
\lim_{\lambda \to 0} \lambda^2\langle V^{1/2}j_3, D_{\tcti(\lambda)} V^{1/2}j_3 \rangle = \lim_{\lambda \to 0}\lambda^2\langle V^{1/2}j_3, D_{\tcti(\lambda)}^< V^{1/2}j_3 \rangle .
\end{equation}
By integrating out the angular variables $\int_{\BR^{3}} V(r) j_3(r;\mu) \frac{e^{i \sqrt{\mu}r \cdot p/\vert p \vert }}{(2\pi)^{3/2}} \dd r = \frac{1}{|\BS^2|}\int_{\BR^3} V(r) j_3(r;\mu)^2=e_\mu$.
Therefore, we can write
\begin{multline}\label{DTj_3d.2}
\langle V^{1/2}j_3, D_{\tcti(\lambda)}^< V^{1/2}j_3  \rangle
=\frac{1}{(2\pi)^{3}}\int_{\BR^{11}; \tilde p^2, \tilde q^2<2\mu-p_1^2, p_1^2<\mu}  \Big( V j_3(r;\mu)(e^{i r \cdot p}-e^{i \sqrt{\mu}r \cdot p/\vert p \vert })\times \\
 B_{\tcti(\lambda)} (p,0) \widehat V(0,\tilde p-\tilde q) B_{\tcti(\lambda)}((p_1,\tilde q),0)e^{-i p \cdot r'} V j_3(r';\mu) \\
+ V j_3(r;\mu)e^{i \sqrt{\mu}r \cdot p/\vert p \vert } B_{\tcti(\lambda)} (p,0) \widehat V(0,\tilde p-\tilde q) B_{\tcti(\lambda)}((p_1,\tilde q),0)(e^{-i p \cdot r'}-e^{-i \sqrt{\mu}r' \cdot p/\vert p \vert}) V j_3(r';\mu)  \Big)\dd p \dd \tilde q \dd r \dd r'\\
+e_\mu^2 \int_{\BR^{8}; \tilde p^2, \tilde q^2<2\mu-p_1^2, p_1^2<\mu} B_{\tcti(\lambda)} (p,0) \frac{e^{i (\tilde p-\tilde q)\tilde r}}{(2\pi)^{3/2}}V(r) B_{\tcti(\lambda)}((p_1,\tilde q),0)\dd p \dd \tilde q \dd r
\end{multline}
By \cite[Proof of Lemma 3.1]{hainzl_bardeencooperschrieffer_2016}
\begin{equation}
\left \vert \int_{\BS^2} e^{i  \vert r \vert w \cdot p}-e^{i \sqrt{\mu} |r| w \cdot p/\vert p \vert }\dd w \right  \vert \leq C \frac{|p|-\sqrt{\mu}}{|p|+\sqrt{\mu}}.
\end{equation}
Furthermore, note that $B_{T} (p,0) \frac{|p|-\sqrt{\mu}}{|p|+\sqrt{\mu}} \leq \frac{1}{\mu}$.
Hence, the first integral in \eqref{DTj_3d.2} is bounded by
\begin{equation}
\frac{C}{\mu}\lVert V j_3 \rVert_1^2  \lVert \widehat V \rVert_\infty \int_{p_1^2+\tilde q^2<2\mu, \tilde p^2<2\mu} B_{\tcti(\lambda)}((p_1,\tilde q),0) \dd p_1 \dd \tilde p \dd \tilde q
\leq C\lVert V j_3 \rVert_1^2  \lVert \widehat V \rVert_\infty m_\mu(\tcti(\lambda)),
\end{equation}
which is of order $1/\lambda$ by Lemma~\ref{lea:Tc0_asy}.

Changing to angular coordinates for the $\tilde p$ and $\tilde q$ integration, the integral on the last line of \eqref{DTj_3d.2} can be rewritten as
\begin{multline}\label{DTj_3d.3}
2 \int_{\BR^3} \dd r\int_0^{\sqrt{\mu}} \dd p_1  \int_0^{\sqrt{2\mu-p_1^2}} \dd t  \int_0^{\sqrt{2\mu-p_1^2}} \dd s \int_{\BS^1} \dd w \int_{\BS^1} \dd w' B_{\tcti(\lambda)} (\sqrt{p_1^2+t^2},0) t \frac{e^{i (t w-s w')\cdot\tilde r}}{(2\pi)^{3/2}}\times \\
V(r)B_{\tcti(\lambda)}(\sqrt{p_1^2+s^2},0) s\\
=2 \int_{\BR^3} \dd r \int_0^{\sqrt{\mu}} \dd p_1  \int_{p_1}^{\sqrt{2\mu}} \dd x  \int_{p_1}^{\sqrt{2\mu}} \dd y \int_{\BS^1} \dd w \int_{\BS^1} \dd w' B_{\tcti(\lambda)} (x,0) x \frac{e^{i (\sqrt{x^2-p_1^2} w-\sqrt{y^2-p_1^2} w')\cdot\tilde r}}{(2\pi)^{3/2}}\times \\
V(r) B_{\tcti(\lambda)}(y,0) y
\end{multline}
where we substituted $x=\sqrt{p_1^2+t^2}, y=\sqrt{p_1^2+s^2}$.
Next, we want to replace the $x^2$ and $y^2$ in the exponent by $\mu$.
We rewrite \eqref{DTj_3d.3} as
\begin{multline}\label{DTj_3d.1}
2\int B_{\tcti(\lambda)} (x,0) x \frac{\left(e^{i \sqrt{x^2-p_1^2} w\cdot\tilde r}-e^{i \sqrt{\mu-p_1^2} w\cdot\tilde r}\right)}{(2\pi)^{3/2}}V(r)e^{-i \sqrt{y^2-p_1^2} w'\cdot\tilde r} B_{\tcti(\lambda)}(y,0) y\dd p_1\dd r \dd x \dd y  \dd w\dd w'\\
+2 \int B_{\tcti(\lambda)} (x,0) x e^{i \sqrt{\mu-p_1^2} w\cdot\tilde r}V(r)\frac{\left(e^{-i \sqrt{y^2-p_1^2} w'\cdot\tilde r}-e^{i \sqrt{\mu-p_1^2} w'\cdot\tilde r}\right)}{(2\pi)^{3/2}} B_{\tcti(\lambda)}(y,0) y\dd p_1\dd r \dd x \dd y  \dd w\dd w'\\
+2 \int B_{\tcti(\lambda)} (x,0) x \frac{e^{i \sqrt{\mu-p_1^2} (w-w')\cdot\tilde r}}{(2\pi)^{3/2}}V(r) B_{\tcti(\lambda)}(y,0) y\dd p_1\dd r \dd x \dd y  \dd w\dd w'
\end{multline}
By \cite[Proof of Lemma 3.4]{henheik_universality_2023}
\begin{equation}
\left \vert \int_{\BS^1} \frac{e^{i  \sqrt{x-p_1^2} w \cdot \tilde r}-e^{i \sqrt{\mu-p_1^2} w \cdot \tilde r}}{(2\pi)^{2}} \dd w \right \vert \leq C \left|\sqrt{x^2-p_1^2}-\sqrt{\mu-p_1^2}\right|^{1/3}\left |(x^2-p_1^2)^{-1/6}+(\mu-p_1^2)^{-1/6}\right|
\end{equation}
We bound this further by $ C \left|x^2-\mu\right|^{1/3}\left ((x^2-p_1^2)^{-1/3}+(\mu-p_1^2)^{-1/3}\right)$. 
Using that $B_{\tcti(\lambda)}(x,0)\leq 1/|x^2-\mu|$ by \eqref{BT_bound} and recalling the definition of $m_\mu$ in \eqref{mmu} we bound the first two lines in \eqref{DTj_3d.1} by
\begin{equation}
C \lVert V \rVert_1 m_{\mu}^{d=2}(\tcti(\lambda)) \int_{0}^{\sqrt{\mu}} \dd p_1 \int_{p_1}^{\sqrt{2\mu}} \dd x \frac{1}{\vert x-\sqrt{\mu} \vert^{2/3} (x+\sqrt{\mu})^{2/3}}\left (\frac{1}{(x^2-p_1^2)^{1/3}}+\frac{1}{(\mu-p_1^2)^{1/3}}\right)
\end{equation}
The integral is bounded by
\begin{equation}
\sqrt{\mu} \int_{0}^{\sqrt{2}} \dd x \int_{0}^{x} \dd p_1 \frac{1}{\vert x-1 \vert^{2/3}}\left (\frac{1}{x^{1/3}(x-p_1)^{1/3}}+\frac{1}{(1-p_1)^{1/3}}\right)<\infty
\end{equation}
Hence, the first two lines in \eqref{DTj_3d.1} are of order $O(1/\lambda)$ by Lemma~\ref{lea:Tc0_asy}.
For the third line we carry out the $r$-integration and obtain
\begin{equation}
2 \int_0^{\sqrt{\mu}}\left( \int_{p_1}^{\sqrt{2\mu}} B_{\tcti(\lambda)} (x,0) x  \dd x \right)^2 \left(\int_{\BS^1} \int_{\BS^1}\widehat V\left(0,\sqrt{\mu-p_1^2} (w-w')\right) \dd w\dd w'\right) \dd p_1.
\end{equation}
Note that $ \int_{p_1}^{\sqrt{2\mu}} B_{\tcti(\lambda)} (x,0) x  \dd x =m_\mu^{d=2}(\tcti(\lambda)) -\int_{0}^{p_1} B_{\tcti(\lambda)} (x,0) x  \dd x$ and
\begin{equation}
\int_{0}^{p_1} B_{\tcti(\lambda)} (x,0) x  \dd x= \frac{1}{2} \int_{(\mu-p_1^2)/\tcti(\lambda)}^{\mu/\tcti(\lambda)} \frac{\tanh s}{s } \dd s \leq \frac{1}{2} \ln \frac{\mu}{\mu-p_1^2}
\end{equation}
where we substituted $s=(\mu-x^2)/\tcti(\lambda)$.
In particular,
\begin{multline}
\Bigg \vert 2 \int_0^{\sqrt{\mu}}\left[\left( \int_{p_1}^{\sqrt{2\mu}} B_{\tcti(\lambda)} (x,0) x  \dd x \right)^2 -m_\mu^{d=2}(\tcti(\lambda)) ^2\right]\times\\
\left(\int_{\BS^1} \int_{\BS^1}\widehat V\left(0,\sqrt{\mu-p_1^2} (w-w')\right) \dd w\dd w'\right) \dd p_1 \Bigg \vert\\
\leq 2|\BS^1|^2 \lVert \widehat V \rVert_\infty  \int_0^{\sqrt{\mu}}\left( \frac{1}{4} \left(\ln \frac{\mu}{\mu-p_1^2}\right)^2+\ln \frac{\mu}{\mu-p_1^2} m_\mu^{d=2}(\tcti(\lambda))\right)\dd p_1 \leq C(1+m_\mu^{d=2}(\tcti(\lambda)))
\end{multline}
which is of order $O(1/\lambda)$ by Lemma~\ref{lea:Tc0_asy}.
In total, we thus obtain
\begin{multline}
\lim_{\lambda \to 0} \lambda^2\langle V^{1/2}j_3, D_{\tcti(\lambda)} V^{1/2}j_3 \rangle\\
=\lim_{\lambda \to 0}2\lambda^2 m_\mu^{d=2}(\tcti) ^2 \sqrt{\mu}e_\mu^2 \int_0^{1}\left(\int_{\BS^1} \int_{\BS^1}\widehat V\left(0,\sqrt{\mu}\sqrt{1-p_1^2} (w-w')\right) \dd w\dd w'\right) \dd p_1
\end{multline}
By writing out the definition of $j_3$ and then switching to spherical coordinates and carrying out the $r$ integration, we have
\begin{multline}
\int_{\BR^4} V(r) j_3(z_1,\tilde r;\mu)^2 \dd r \dd z_1=\int_{\BS^2} \dd u \int_{\BS^2} \dd v \int_{\BR^7} \dd p \dd r \dd z_1\frac{e^{i p \cdot r}\widehat{V}(p)}{(2\pi)^{3/2}} \frac{e^{i\sqrt{\mu} (z_1,\tilde r)\cdot(u-v)}}{(2\pi)^3}
= \frac{1}{(2\pi)^{3/2}}\times\\
\int_\BR \left( \int_0^\pi \sin \theta \dd \theta \int_0^\pi \sin \theta' \dd \theta' \int_{\BS^1} \dd w \int_{\BS^1} \dd w' \widehat{V}(0,\sqrt{\mu} (\sin \theta w-\sin \theta' w') e^{i \sqrt{\mu} z_1(\cos \theta-\cos \theta')}\right)\dd z_1\\
=\frac{1}{\sqrt{\mu}(2\pi)^{1/2}} \int_{-1}^1 \dd t \int_{-1}^1\dd s \int_{\BS^1} \dd w \int_{\BS^1} \dd w' \widehat{V}(0,\sqrt{\mu} (\sqrt{1-t^2} w-\sqrt{1-s^2} w')\delta(s-t),
\end{multline}
where in the last step we substituted $t=\cos \theta, s=\cos \theta'$ and carried out the $z_1$ integration.
Furthermore, according to Lemma~\ref{lea:Tc0_asy}, $\lim_{\lambda \to 0} \lambda m_\mu^{d=2}(\tcti) e_\mu = \frac{1}{\sqrt{\mu}}$.
This gives the desired
\begin{equation}
\lim_{\lambda \to 0} \lambda^2\langle V^{1/2}j_3, D_{\tcti(\lambda)} V^{1/2}j_3 \rangle = (2\pi)^{1/2}\int_{\BR^4} V(r) j_3(z_1,\tilde r;\mu)^2 \dd r \dd z_1
\end{equation}
\end{proof}

\section{Boundary Superconductivity in 3d}\label{sec:3d_condition}
In this section we shall prove Theorem~\ref{thm2}, which provides sufficient conditions for \eqref{3d_cond} to hold.
Due to rotation invariance, we consider the spherical average of $\tilde m^{D/N}_3$ (defined in \eqref{mtilde}).
With
\begin{equation}\label{md}
m^{D/N}_3(\vert r\vert;\mu):=\frac{1}{4\pi }\int_{\BS^2} \tilde m^{D/N}_3(\vert r\vert \omega;\mu) \dd \omega
\end{equation}
we have $\int_{\BR^3} V(r) \tilde m^{D/N}_3(r;\mu)\dd r =\int_{\BR^3} V(r) m^{D/N}_3(\vert r\vert;\mu) \dd r$.
Furthermore, we have the scaling property
\begin{equation}\label{md_scale}
m_3^{D/N}(\vert r\vert ;\mu)=\frac{1}{\sqrt{\mu}} m_3^{D/N}(\sqrt{\mu}\vert r\vert ;1).
\end{equation}
We shall derive the following, more explicit, expression for $m^{D/N}_3$ in Section~\ref{sec:pf_t1_expr}.
\begin{lemma}\label{lea:t1_expr}
For $x\geq0$ we can write $m^{D}_3(x;1)=\sum_{j=1}^4 t_j(x)$ and $m^{N}_3(x;1)=\sum_{j=1}^2 t_j(x)-\sum_{j=3}^4t_j(x)$, where
\begin{align*}
t_1(x)&=\frac{4}{\pi x} \int_1^\infty \frac{\sin^2(x k)}{k} \arcoth(k) \dd k\\
t_2(x)&=- \frac{2}{\pi} \frac{\sin^2(x)}{x}\\
t_3(x)&=- 2\frac{\sin^2( x)}{x^2}\\
t_4(x)&= \frac{4\sin x}{ \pi x^2}\left(\sin  x\Si 2  x-\cos x \Cin 2 x\right) \\
&= \frac{\sin x}{2\pi^{3} x}\int_{\BS^2}\int_{\BS^{2}} \frac{ \sin(x \omega_1 \vert \omega'_1 \vert) e^{-i x \tilde \omega \cdot \tilde \omega ' }}{\omega_1}  \dd \omega \dd \omega'
\end{align*}
with $\Cin(x)=\int_0^x \frac{1-\cos t}{t} \dd t $ and $\Si (x)=\int_0^x \frac{\sin t}{t} \dd t$.
\end{lemma}

To determine for which interactions $\int_{\BR^3} V(r)  m^{D/N}_3(\vert r\vert;\mu)\dd r >0$ holds, we need to understand $m^{D/N}_3(\vert r\vert;\mu)$.
In Figures~\ref{fig:m3} and~\ref{fig:m3N1} we plot $m_3^D$ and $m_3^N$ for $\mu=1$, respectively.
\begin{figure}[h]
\includegraphics[width=\textwidth]{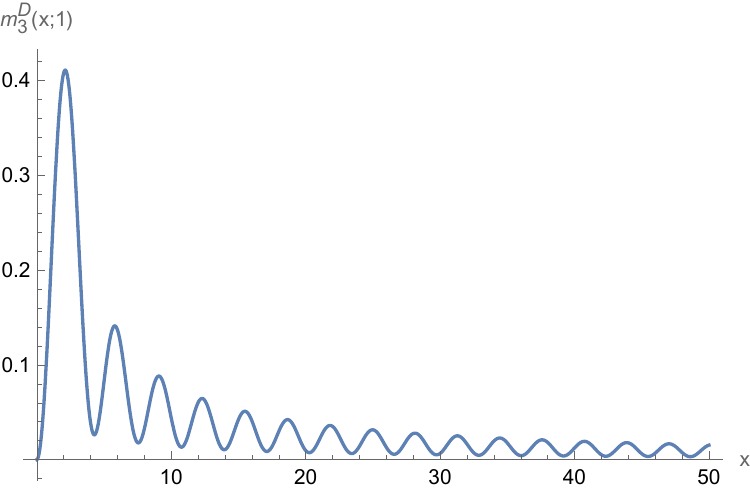}
\caption{Plot of $m_3^D$ for $\mu=1$, created using \cite{inc_mathematica_2022}. }
\label{fig:m3}
\end{figure}
\begin{figure}[h]
\includegraphics[width=\textwidth]{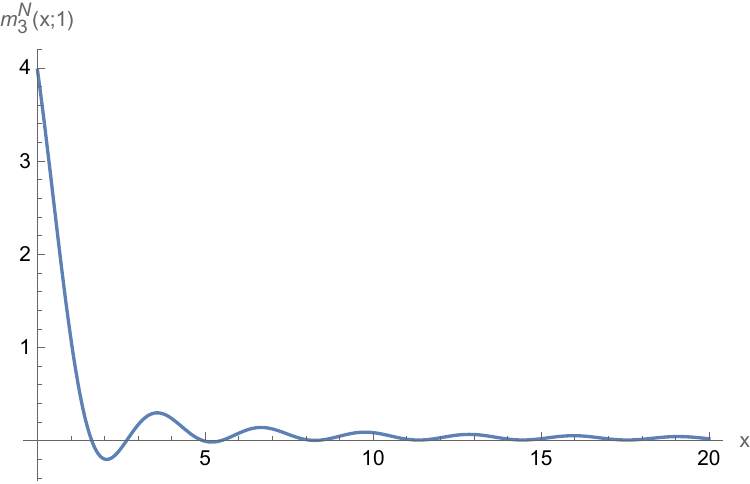}
\caption{Plot of $m_3^N$ for $\mu=1$, created using \cite{inc_mathematica_2022}. }
\label{fig:m3N1}
\end{figure}
The function $m_3^D$ seems to be nonnegative.
If one could prove that $m_3^D\geq0$, then Theorem~\ref{thm1} would apply to all $V\geq 0$ satisfying Assumption~\ref{aspt_V_halfspace}.
Unfortunately, this is beyond our reach.
On the other hand, the function $m_3^N$ changes sign, but is positive in a neighborhood of zero.

\begin{rem}
To create the plots, it is computationally more efficient to use the first expression for $t_4$, whereas for the following analytic computations the second expression is more convenient.
\end{rem}

Intuitively, if we let $\mu\to 0$, due to the scaling \eqref{md_scale} the sign of $\int_{\BR^3} V(r) m^{D/N}_3(\vert r\vert;\mu)\dd r$ is determined by the values of $m_3^{D/N}(\vert r\vert;1)$ for $r$ in the vicinity of zero.
To obtain Theorem~\ref{thm2}, we prove that both functions $m_3^{D/N}(\vert r\vert;1)$ are non-negative in a neighborhood of zero.

The following is proved in Section~\ref{sec:pf_tj_reg}.
\begin{lemma}\label{lea:tj_reg}
The functions $t_j$ for $j=1,2,3,4$ are bounded and twice continuously differentiable.
The values of the functions and their derivatives at zero are listed in Table~\ref{m3_derivatives}.
\begin{table}[h]
\centering
\begin{tabular}{|l|cccc|cc|}
\hline
$f$    & $t_1$ & $t_2$ & $t_3$ & $t_4$ & $m^D_3(\cdot;1)$ & $m^N_3(\cdot;1)$ \\ \hline
$f(0)$ &  $2$     &  $0$     &  $-2$     &   $0$    &     $0$  &     $4$  \\
$f' (0)$   & $ - 2/\pi$    &  $- 2/\pi$     &    $0$   &    $4/\pi$   &   $0$    &       \\
 $f'' (0)$ &   $-8/9$    &    $0$   &    $4/3$   &    $0$   &   $4/9$   &   \\ \hline
\end{tabular}
\caption{Values of the functions $t_j$ and $m^{D/N}_3$ and their derivatives at zero. The missing entries are not needed.}
\label{m3_derivatives}
\end{table}
\end{lemma}

\begin{proof}[Proof of Theorem~\ref{thm2}]
We start with the case of Neumann boundary condition.
By \eqref{md_scale}, it suffices to prove that $\lim_{\mu\to 0} \int_{\BR^3} V(r) m_3^{N}(\sqrt{\mu}\vert r\vert ;1)\dd r>0$.
With $V\in L^1$ and Lemma~\ref{lea:tj_reg} it follows by dominated convergence that  $\lim_{\mu\to 0} \int_{\BR^3} V(r) m_3^{N}(\sqrt{\mu}\vert r\vert ;1)\dd r=m_3^{N}(0 ;1) \int_{\BR^3} V(r) \dd r=4 \int_{\BR^3} V(r) \dd r$.
Since $\widehat V(0)>0$ by assumption, this is positive.

For Dirichlet boundary conditions, according to Lemma~\ref{lea:tj_reg}, $m_3^{D}(0 ;1)$ and its first derivative vanish.
Thus, we consider $I(\sqrt{\mu}):=\frac{1}{\mu}\int_{\BR^3}  m_3^D(\sqrt{\mu} \vert r\vert;1) V(r) \dd r$.
Since $m_3^D(\cdot;1)$ is bounded, $I$ is continuous away from $0$.
It suffices to prove that $\lim_{\mu\to0} I(\sqrt{\mu})>0$.
According to Lemma~\ref{lea:tj_reg} and Taylor's theorem,
we have $m_3^D(x;1) =\frac{1}{2} (m_3^D)''(0;1) x^2+R(x)$, where $R$ is continuous with $\lim_{x\to 0}\frac{\vert R(x)\vert}{x^2} =0$.
Let $\eps>0$ and $c:=\sup_{0\leq x< \eps} \frac{\vert R(x)\vert}{x^2}<\infty$.
One can bound
\begin{multline}
\left \vert \frac{1}{\mu}m_3^D(\sqrt{\mu} \vert r\vert;1)  V(r) \right \vert \leq \chi_{\sqrt{\mu}  \vert r \vert<\eps} \Big(\frac{1}{2} (m_3^D)''(0;1)+c\Big) \vert |r|^2 V(r) \vert + \chi_{\sqrt{\mu} \vert r \vert>\eps} \frac{\lVert m_3^D \rVert_\infty}{\eps^2} \vert |r|^2 V(r) \vert \\
\leq \Big(\frac{1}{2} (m_3^D) ''(0;1)+c+\frac{\lVert m_3^D \rVert_\infty}{\eps^2} \Big) \vert |r|^2 V(r) \vert,
\end{multline}
which is integrable by the assumptions on $V$.
By dominated convergence
\begin{equation}
\lim_{\mu\to0} I(\sqrt{\mu})=\int_{\BR^3} \lim_{\mu\to0} \frac{m_3^D(\sqrt{\mu} \vert r\vert;1) }{\mu \vert r\vert^2} V(r) \vert r \vert^2 \dd r= \frac{1}{2}\int_{\BR^3}  (m_3^D)''(0;1) V(r) \vert r \vert^2 \dd r=\frac{2}{9}\int_{\BR^3} V(r) \vert r \vert^2 \dd r,
\end{equation}
which is positive by assumption.
\end{proof}

\subsection{Proof of Lemma~\ref{lea:t1_expr}}\label{sec:pf_t1_expr}
\begin{proof}[Proof of Lemma~\ref{lea:t1_expr}]
With
\begin{align*}
\tilde t_1(r)&=\int_{\BR} j_3(z_1,r_2,r_3;1)^2 \chi_{|z_1|>|r_1|} \dd z_1 \\
\tilde t_2(r)&=-  j_3(r;1)^2 \int_{\BR} \chi_{|z_1|<|r_1|}\dd z_1 \\
\tilde t_3(r)&= \mp  \pi j_3(r;1)^2\\
\tilde t_4(r)&=\pm 2 j_3(r;1) \int_{\BR} j_3(z_1,r_2,r_3;1)\chi_{|z_1|<|r_1|}\dd z_1
\end{align*}
one can write $\widetilde m^{D}_3(r;1)=\sum_{j=1}^4 \tilde t_j(r)$ and $\widetilde m^{N}_3(r;1)=\sum_{j=1}^2 \tilde t_j(r)-\sum_{j=3}^4 \tilde t_j(r)$.
Let $t_j(\vert r\vert)=\frac{1}{4\pi }\int_{\BS^2} \tilde t_j^{D/N}(\vert r\vert \omega;\mu) \dd \omega.$
The following explicit computations show that the $t_j$ agree with the claimed expressions.

Recall that $j_3(r;1)=\sqrt{\frac{2}{\pi}} \frac{\sin\vert r \vert}{ \vert r \vert}$.
For $t_1$ we write out the integral in spherical coordinates and substitute $z_1=x y$ and $s=\cos \theta$
\begin{multline}
t_1(x)
= \frac{1}{\pi} \frac{2\pi}{4\pi} \int_{0}^\pi \int_{\BR} \frac{\sin^2 \sqrt{z_1^2+(x \sin \theta)^2}}{  z_1^2+(x \sin \theta)^2 } \chi_{|z_1|>x| \cos \theta|}  \sin \theta \dd z_1 \dd \theta\\
= \frac{1}{\pi x} \int_{-1}^1 \int_{\BR} \frac{\sin^2 x \sqrt{y^2+1-s^2}}{y^2+1-s^2} \chi_{|y|>\vert s\vert }  \dd y \dd s
\end{multline}
Next, we use the reflection symmetry of the integrand in $s$ and $y$, substitute $y$ by $k=\sqrt{y^2+1-s^2}$ and then carry out the $s$ integration to obtain
\begin{equation}
t_1(x)= \frac{4}{\pi x} \int_{0}^1 \int_{1}^\infty \frac{\sin^2 x k}{k \sqrt{k^2+s^2-1}} \dd k \dd s =  \frac{4}{\pi x} \int_{1}^\infty \frac{\sin^2 x k}{k } \arcoth(k) \dd k.
\end{equation}
For $t_2$, we have
\begin{equation}
t_2(x)=-  \frac{2}{\pi} \frac{\sin^2 x }{ x^2} \frac{1}{4\pi} \int_{\BS^2} 2 x \vert \omega_1 \vert \dd \omega=- \frac{2}{\pi} \frac{\sin^2 x}{ x}.
\end{equation}
Since $\tilde t_3$ is radial, we have $ t_3=\tilde t_3$.
For $t_4$ we want to derive two expressions.
For the first, we perform the same substitutions as for $t_1$
\begin{multline}
t_4(x)=  \frac{4}{\pi} \frac{\sin x}{ x}\frac{2\pi}{4\pi} \int_{0}^\pi \int_{\BR} \frac{\sin \sqrt{z_1^2+(x \sin \theta)^2}}{ \sqrt{ z_1^2+(x \sin \theta)^2 }} \chi_{|z_1|<x| \cos \theta|}  \sin \theta \dd z_1 \dd \theta\\
=  \frac{2}{\pi} \frac{\sin x}{ x} \int_{-1}^1 \int_{\BR} \frac{\sin x \sqrt{y^2+1-s^2}}{ \sqrt{y^2+1-s^2}} \chi_{|y|<| s|} \dd y \dd s
= \frac{8}{\pi} \frac{\sin x}{ x} \int_{0}^1 \int_{0}^1 \frac{\sin x k}{\sqrt{k^2+s^2-1}} \chi_{k^2+s^2>1}\dd k \dd s\\
= \frac{8}{\pi} \frac{\sin x}{ x} \int_{0}^1\sin x k \artanh k \dd k = \frac{4\sin x}{ \pi x^2}\left(\sin  x \Si 2  x-\cos x \Cin 2 x \right)
\end{multline}
To obtain the second expression for $t_4$, note that $\int_\BR e^{-i  \omega_1 z_1} \chi_{\vert z_1 \vert<\vert r_1 \vert} \dd z_1=\frac{2 \sin  \omega_1 |r_1|}{ \omega_1}$.
Therefore,
\begin{multline}
t_4(x)=2 \sqrt{\frac{2}{\pi}} \frac{\sin x}{ x} \frac{1}{4\pi}\int_{\BS^2} \int_\BR \int_{\BS^2}\frac{e^{-i  \omega \cdot (z_1,x \tilde \omega' )}}{(2\pi)^{3/2}} \chi_{\vert z_1 \vert<x \vert \omega_1' \vert}\dd \omega\dd z_1 \dd \omega'  \\
= \frac{1}{2\pi^{3}} \frac{\sin x}{ x}\int_{\BS^2} \int_{\BS^2}\frac{ \sin x \omega_1 \vert \omega_1'\vert}{ \omega_1} e^{-i x \tilde \omega \cdot \tilde \omega' }\dd \omega\dd \omega'
\end{multline}

\end{proof}

\subsection{Proof of Lemma~\ref{lea:tj_reg}}\label{sec:pf_tj_reg}
\begin{proof}[Proof of Lemma~\ref{lea:tj_reg}]
Since $\sin(x)/x$ is a bounded and smooth function, also $t_2$ and $t_3$ are bounded and smooth.
Elementary computations give the entries in Table~\ref{m3_derivatives}.

For $t_4$ use the second expression in Lemma~\ref{lea:t1_expr}.
Since the integrand is bounded and smooth and the domain of integration is compact, the integral is bounded and we can exchange integration and taking limits and derivatives.
In particular, $t_4$ is bounded and smooth and it is then an elementary computation to verify the entries in Table~\ref{m3_derivatives}.
For instance,
\begin{equation}
t_4'(0)=\frac{1}{2 \pi^3} \int_{\BS^2}\int_{\BS^{2}}  \vert \omega_1' \vert \dd \omega \dd \omega'= \frac{4}{\pi}.
\end{equation}

To study $t_1$ we define auxiliary functions $f(x)=\frac{4}{\pi x} \artanh(x)$ and $g(x)=\frac{\sin(x)^2}{x^2}$.
Note that $f(x)$ diverges logarithmically for $x\to1$ and is continuous otherwise with $f(0)=\frac{4}{\pi}$.
Furthermore, $f(x)$ is increasing on $[0,1)$ and for every $0<\eps<1$, $\sup_{0\leq x<\eps} \frac{f'(x)}{x}=\frac{f'(\eps)}{\eps}<\infty$ since all coefficients in the Taylor series of $\artanh(x)$ are positive.

We can write
\begin{equation}
t_1(x)=\int_1^\infty x g(x k) f(1/k) \dd k =\int_1^c x g(x k) f(1/k) \dd k +\int_{c x}^\infty g(k) f(x/k) \dd k
\end{equation}
for any constant $c>1$.
The first integrand is bounded by $Cx \arcoth(k)$, the second one by $C\frac{1}{k^2}$ (since $f$ is bounded on the integration domain).
By dominated convergence we obtain that $t_1$ is continuous and $t_1(0)=\frac{4}{\pi} \int_0^\infty g(k) \dd k =2$.

For $x>0$ we compute the derivative
\begin{multline}
t_1'(x)=\int_1^c (g(x k)+x k g'(x k)) f(1/k) \dd k -cg(cx)f(1/c)+\int_{c x}^\infty g(k) f'(x/k) \frac{1}{k}\dd k\\
=\int_1^c (g(x k) +x k g'(xk))f(1/k) \dd k  -cg(cx)f(1/c)+\int_{c}^\infty g(k x) f'(1/k) \frac{1}{k}\dd k,
\end{multline}
where we could apply the Leibnitz integral rule since $f'(1/k)$ decays like $1/k$ for $k\to \infty$.
By dominated convergence, $ t_1'$ is continuous for $x>0$.
By continuity of $t_1$ and the mean value theorem, $t_1'(0)=\lim_{x\to0}\frac{t_1(x)-t_1(0)}{x}=\lim_{x\to0}\lim_{y\to 0}\frac{t_1(x)-t_1(y)}{x-y}=\lim_{x\to0}t_1'(x)$.
We evaluate
\begin{multline}
t_1'(0)=\int_1^c f(1/k) \dd k -c f(1/c)+\int_{c}^\infty f'(1/k) \frac{1}{k}\dd k\\
=\int_1^c \left(f(1/k) -f(1/c)\right)\dd k-f(1/c)+\int_{c}^\infty f'(1/k) \frac{1}{k}\dd k
\end{multline}
This is a number independent of $c$.
To compute the number, we let $c\to \infty$, and by monotone convergence
\begin{equation}
t_1'(0)=\int_1^\infty \left(f(1/k) -f(0)\right)\dd k-f(0)=\frac{2}{\pi}-\frac{4}{\pi}=-\frac{2}{\pi}.
\end{equation}

Note that $g'(k)=2(\cos k -\frac{\sin k}{k})\frac{\sin k}{k^2}$ has a zero of order one at $k=0$.
Therefore, $ \left | g'(k x) f'(1/k)\right|<\frac{C}{x^2 k^3}$ and for $x>0$ the second derivative is
\begin{multline}
 t_1''(x)=\int_1^c (2x g'(x k)+x k^2 g''(x k)) f(1/k) \dd k  -c^2 g'(cx)f(1/c)+\int_{c}^\infty g'(k x) f'(1/k)\dd k\\
= \int_1^c (2 x g'(x k)+x k^2 g''(x k)) f(1/k) \dd k -c^2 g'(cx)f(1/c)+\int_{cx}^\infty \frac{g'(y)}{y} \frac{f'(x/y)}{x/y}\dd y
\end{multline}
We can bound  $\frac{g'(y)}{y}\leq \frac{C}{1+y^3}$ and $\sup_y |\frac{f'(x/y)}{x/y}\chi_{y>cx}|=c f'(1/c)<\infty$.
By dominated convergence, the function above is continuous (also at zero).
We have
\begin{equation}
t_1''(0)= \int_{0}^\infty \frac{g'(y)}{y} \dd y \lim_{x\to0} \frac{f'(x)}{x}
\end{equation}
Since $\int_{0}^\infty \frac{g'(y)}{y} \dd y =-\frac{\pi}{3}$ and $\lim_{x\to0} \frac{f'(x)}{x}=\frac{8}{3\pi}$ we obtain
\begin{equation}
t_1''(0)=-\frac{8}{9}.
\end{equation}

\end{proof}

\section{Relative Temperature Shift}\label{sec:rel_T}
In this section we shall prove Theorem~\ref{thm3}, which states that the relative temperature shift vanishes in the weak coupling limit.
We proceed similarly to the $\delta$-interaction case in one dimension analyzed in \cite{hainzl_boundary_nodate}.
For this, we switch to the Birman-Schwinger formulation.
Let $\tilde \Omega_1=\{(r,z)\in \BR^{2d}\vert \vert r_1 \vert < z_1 \}$.
Define the operator $A_T^1$ on $\psi \in L_{\rm s}^2(\tilde \Omega_1)=\{\psi \in L^2(\tilde \Omega_1) \vert \psi(r,z)=\psi(-r,z)\}$ via
\begin{multline}
\langle \psi, A_{T}^1 \psi \rangle
=\int_{\BR^{4d+2(d-1)}} \dd r \dd r' \dd p \dd q \dd \tilde z \dd \tilde z' \int_{\vert r_1 \vert<z_1} \dd z_1   \int_{\vert r_1 '\vert<z_1'} \dd z'_1  \frac{1}{(2\pi)^{2d}}\overline{\psi(r,z)}  V(r)^{1/2} e^{i(p \cdot z+q \cdot r)}\times \\
B_{T}\left(p,q\right) \Bigg(e^{-i(p_1 z'_1+q_1 r'_1)}+e^{i(p_1 z'_1+q_1 r'_1)}\mp e^{-i(q_1 z'_1+p_1 r'_1)}\mp e^{i(q_1 z'_1+p_1 r'_1)}\Bigg) e^{-i(\tilde p \cdot  \tilde z'+\tilde q \cdot \tilde r')}\vert V(r')\vert^{1/2}\psi(r',z'),
\end{multline}
where the upper signs correspond to Dirichlet and the lower signs to Neumann boundary conditions, respectively.
It follows from a computation analogous to \cite[Lemma 2.4]{hainzl_boundary_nodate} that the operator $A_T^1$ is the Birman-Schwinger operator corresponding to $H_T^{\Omega_1}$ in relative and center of mass variables.
The Birman-Schwinger principle implies that $\sgn \inf \sigma(H_T^{\Omega_1}) = \sgn (1/\lambda-\sup \sigma(A_T^{1}))$, where we use the convention that $\sgn\, 0  =0 $.

Recall the Birman-Schwinger operator $A_T^0$ corresponding to $H_T^0$ from \eqref{BS-ti}.
Similarly, the Birman-Schwinger operator $ A_{T}^{\Omega_0}$ corresponding to $H_T^{\Omega_0}$ in relative and center of mass variables is defined on $\psi(r,z)\in L^2(\BR^d\times \BR^d)$ with $\psi(r,z)=\psi(-r,z)$ and satisfies
\begin{equation}
\langle \psi, A_{T}^{\Omega_0} \psi \rangle
=\int_{\BR^{6d}} \dd r \dd r' \dd p \dd q \dd  z \dd  z' \overline{\psi(r,z)}  V(r)^{1/2} \frac{e^{i(p \cdot (z-z')+q \cdot (r-r'))}}{(2\pi)^{2d}}
B_{T}\left(p,q\right)  \vert V(r')\vert^{1/2}\psi(r',z').
\end{equation}

Let $a_T^j=\sup \sigma (A_T^j)$.
Let us first observe that there is a $T_0>0$ such that $\af=\ati$ for $T<T_0$.
Let $\lambda_0>0$ such that $\tcf(\lambda)=\tcti(\lambda)$ for $\lambda\leq \lambda_0$, see Remark~\ref{ti_reduction_fix}.
Choose $T_0=\tcf(\lambda_0)=\tcti(\lambda_0)$ and let $T<T_0$.
Due to strict monotonicity of $H_T^0$ in $T$, $T=\tcti(\lambda)$ for some $\lambda<\lambda_0$.
By choice of $\lambda_0$ also $\tcf(\lambda)=T$.
The Birman-Schwinger principle implies $\af=\lambda=\ati$.

For $T\to0$, the asymptotics of $\af$ thus agrees with the asymptotics of $\ati$, i.e.~$\af=e_\mu \mu^{d/2-1}\ln(\mu/T) +O(1)$ \cite[Theorem 3.3]{hainzl_bardeencooperschrieffer_2016} and \cite[Theorem 2.5]{henheik_universality_2023}.
One can reformulate the claim of Theorem~\ref{thm3} in terms of the Birman-Schwinger operators.
Then
\begin{equation}
\lim_{\lambda \to 0} \frac{\tch(\lambda)-\tcf(\lambda)}{\tcf(\lambda)} =0 \Leftrightarrow \lim_{T\to 0} \left(\af-\ah\right)=0.
\end{equation}
This is a straightforward generalization of \cite[Lemma 4.1]{hainzl_boundary_nodate} and we refer to \cite[Lemma 4.1]{hainzl_boundary_nodate} for its proof.

\begin{proof}[Proof of Theorem~\ref{thm3}]
First we will argue that $\af\leq \ah$.
If $\inf \sigma(K_T^{\Omega_0}-\lambda V)<2T$, then $\inf \sigma(K_T^{\Omega_0}-\lambda' V)<\inf \sigma(K_T^{\Omega_0}-\lambda V)$ for all $\lambda'>\lambda$.
Furthermore, $\inf \sigma(K_T^{\Omega_0}-(\af)^{-1}V)=0=\inf \sigma(K_T^{\Omega_1}-(\ah)^{-1}V)\leq \inf \sigma(K_T^{\Omega_0}-(\ah)^{-1}V)$, where we used Lemma~\ref{H1_esspec} in the last step.
In particular, $\af\leq \ah$.

It remains to show that $\lim_{T\to 0} \left(\af-\ah\right)\geq 0$.
Let $\iota: L^2(\tilde \Omega_1)\to  L^2(\BR^{2d})$ be the isometry
\begin{equation}
\iota \psi(r_1,\tilde r,z_1,\tilde z)
= \frac{1}{\sqrt{2}} (\psi(r_1,\tilde r,z_1,\tilde z) \chi_{\tilde \Omega_1}(r,z) +\psi(-r_1,\tilde r,-z_1,\tilde z) \chi_{\tilde \Omega_1}(-r_1,\tilde r,-z_1,\tilde z)).
\end{equation}
Let $F_2$ denote the Fourier transform in the second variable
$F_2 \psi (r, q)=\frac{1}{(2\pi)^{d/2}} \int_{\BR^d} e^{-iq\cdot z} \psi(r,z) \dd z$
and $F_1$ the Fourier transform in the first variable
$F_1 \psi (p,q)=\frac{1}{(2\pi)^{d/2}} \int_{\BR^d} e^{-ip\cdot r} \psi(r,q) \dd r.$
Recall that by assumption $V\geq 0$ and for functions $\psi \in L^2(\BR^d\times \BR^d)$ we have $V^{1/2}\psi(r,q)=V^{1/2}(r)\psi(r,q)$.
We define self-adjoint operators $\tilde E_T$ and $G_T$ on $L^2(\BR^{2d})$ through
\begin{equation}
\langle \psi, \tilde E_{T} \psi \rangle= \af \lVert \psi \rVert_2^2 -\int_{\BR^{2d}}  B_{T}(p,q) |F_1 V^{1/2}\psi(p,q)|^2\dd p \dd q
\end{equation}
and
\begin{equation}
\langle \psi, G_{T} \psi \rangle=\int_{\BR^{2d}} \overline{F_1 V^{1/2}\psi((q_1,\tilde p),(p_1,\tilde q))} B_{T}(p,q) F_1 V^{1/2}\psi(p,q) \dd p \dd q.
\end{equation}
With this notation, we have $\af \BI-A_T^1=\iota^\dagger F_2^\dagger (\tilde E_T\pm G_T) F_2 \iota$, where $\BI$ denotes the identity operator on $L_{\rm s}^2(\tilde \Omega_1)$.
In particular,
\begin{equation}
\af-\ah=\inf_{\psi \in L^2_{\rm s}(\tilde \Omega_1), \lVert \psi \rVert_2=1} \langle F_2\iota \psi,(\tilde E_{T}\pm G_{T} )F_2\iota \psi \rangle \geq \inf_{ \psi \in L^2_{\rm s}(\BR^{2d}), \lVert \psi \rVert_2=1} \langle \psi,(\tilde E_{T}\pm G_{T}) \psi \rangle,
\end{equation}
where we used that $\lVert F_2 \iota \psi \lVert_2= \lVert \psi \lVert_2$.
Define the function
\begin{equation}
E_{T}(q)= \af-\lVert V^{1/2}B_{T}(\cdot ,q) V^{1/2} \rVert_{\textrm{s}},
\end{equation}
where $\lVert \cdot \rVert_{\textrm{s}}$ denotes the operator norm of the operator restricted to even functions.
Since $\af=\sup_q \lVert V^{1/2}B_{T}(\cdot ,q) V^{1/2} \rVert_{\textrm{s}}$, we have $E_T(q)\geq 0$ for all $T$.
Let $ E_{T}$ act on $L^2(\BR^{2d})$ as $ E_{T} \psi (r,q)=  E_{T}(q)\psi (r,q)$.
Then
\begin{equation}
\af-\ah\geq \inf_{ \psi \in L^2_{\rm s}(\BR^{2d}), \lVert  \psi \rVert_2=1} \langle  \psi,( E_{T}\pm G_{T} ) \psi \rangle.
\end{equation}
It thus suffices to prove that $\lim_{T\to 0} \inf \sigma ( E_{T}\pm G_{T})\geq 0$.
With the following three Lemmas, which are proved in the next sections, the claim follows completely analogously to the proof of \cite[Theorem 1.2 (ii)]{hainzl_boundary_nodate}.
For completeness, we provide a sketch of the argument in \cite[Theorem 1.2 (ii)]{hainzl_boundary_nodate} after the statement of the Lemmas.
\begin{lemma}\label{3d_w_lea3}
Let $\mu>0$, $d\in\{1,2,3\}$ and let $V\geq 0$ satisfy Assumption \ref{aspt_V_halfspace}\eqref{aspt.1}. Then $\sup_{T>0} \lVert G_{T} \rVert <\infty$.
\end{lemma}
\begin{lemma}\label{3d_w_lea4}
Let $\mu>0$, $d\in\{1,2,3\}$ and let $V\geq 0$ satisfy Assumption \ref{aspt_V_halfspace}\eqref{aspt.1}.
Let $\BI_{\leq \epsilon}$ act on $L^2(\BR^{2d})$ as $\BI_{\leq \epsilon} \psi(r,p)=\psi(r,p)\chi_{\vert p\vert\leq \epsilon}$.
Then $\lim_{\epsilon \to 0} \sup_{T>0} \lVert \BI_{\leq \epsilon} G_{T} \BI_{\leq \epsilon} \lVert =0$.
\end{lemma}
\begin{lemma}\label{3d_w_lea2}
Let $\mu>0$, $d\in\{1,2,3\}$ and let $V\geq 0$ satisfy Assumption \ref{aspt_V_halfspace}.
Let $0<\epsilon<\sqrt{\mu}$.
There are constants $c_1,c_2,T_1>0$ such that for $0<T<T_1$ and $|q|>\eps$ we have $ E_{T}(q)>c_1 \vert \ln(c_2/T)\vert$.
\end{lemma}
Since $E_{T}(q)\geq 0$, we can write
\begin{equation}\label{E+G+d}
E_{T}\pm G_{T}+\delta=\sqrt{E_{T}+\delta}\left(\BI \pm \frac{1}{\sqrt{E_{T}+\delta}}G_{T}\frac{1}{\sqrt{E_{T}+\delta}}\right)\sqrt{E_{T}+\delta}
\end{equation}
for any $\delta>0$.
It suffices to prove that for all $\delta>0$
\begin{equation}\label{eq:EBEto0}
\lim_{T\to 0} \left\lVert \frac{1}{\sqrt{E_{T}+\delta}}G_{T}\frac{1}{\sqrt{E_{T}+\delta}}\right\rVert=0 \,.
\end{equation}
To prove \eqref{eq:EBEto0}, with the notation introduced in  Lemma~\ref{3d_w_lea4} we have for all $0<\epsilon<\sqrt{\mu}$
\begin{multline}
\left\lVert \frac{1}{\sqrt{E_{T}+\delta}}G_{T}\frac{1}{\sqrt{E_{T}+\delta}}\right\rVert \leq
\left\lVert \BI_{\leq\epsilon} \frac{1}{\sqrt{E_{T}+\delta}}G_{T}\frac{1}{\sqrt{E_{T}+\delta}}\BI_{\leq\epsilon}\right\rVert\\
+\left\lVert \BI_{\leq\epsilon} \frac{1}{\sqrt{E_{T}+\delta}}G_{T}\frac{1}{\sqrt{E_{T}+\delta}}\BI_{>\epsilon}\right\rVert
+\left\lVert \BI_{>\epsilon} \frac{1}{\sqrt{E_{T}+\delta}}G_{T}\frac{1}{\sqrt{E_{T}+\delta}}\right\rVert\,.
\end{multline}
With $E_{T}\geq 0$ and Lemma~\ref{3d_w_lea2} we obtain
\begin{equation}
\lim_{T\to 0}\left\lVert \frac{1}{\sqrt{E_{T}+\delta}}G_{T}\frac{1}{\sqrt{E_{T}+\delta}}\right\rVert \leq \sup_{T>0} \frac{1}{\delta}\left\lVert \BI_{\leq\epsilon} G_{T} \BI_{\leq\epsilon}\right\rVert+\lim_{T\to 0} \frac{2}{(\delta  c_1 \vert \ln(c_2/T)\vert)^{1/2}} \lVert G_{T} \rVert.
\end{equation}
The second term vanishes by Lemma~\ref{3d_w_lea3} and the first term can be made arbitrarily small by Lemma~\ref{3d_w_lea4}.
Hence \eqref{eq:EBEto0} follows.
\end{proof}
\begin{rem}
The variational argument above relies on $A_T^1$ being self-adjoint.
This is why we assume $V\geq 0$ in Theorem \ref{thm3}.
\end{rem}

\subsection{Proof of Lemma~\ref{3d_w_lea3}}
\begin{proof}[Proof of Lemma~\ref{3d_w_lea3}]
We have $\lVert G_T \rVert \leq \lVert G_T^< \rVert+\lVert G_T^> \rVert$, where for $d\in\{2,3\}$
\begin{equation}\label{pf:3d_2_lea3.1}
\langle \psi, G_{T}^< \psi \rangle=\int_{\BR^{2d}} \overline{F_1 V^{1/2}\psi((q_1,\tilde p),(p_1,\tilde q))} B_{T}(p,q)\chi_{|\tilde p|<2\sqrt{\mu}} F_1 V^{1/2}\psi(p,q) \dd p \dd q,
\end{equation}
and for $ G_{T}^>$ change $\chi_{|\tilde p|<2\sqrt{\mu}}$ to $\chi_{|\tilde p|>2\sqrt{\mu}}$.
For $d=1$ set $G_T^{<}=G_T$ and $G_T^>=0$.
We will prove that $G_T^<$ and $G_T^>$ are bounded uniformly in $T$.

To bound $G_T^>$ in $d=2,3$ we use the Schwarz inequality in $p_1,q_1$ to obtain
 \begin{equation}\label{pf:3d_2_lea3.2}
\lVert G_T^>\rVert\leq  \sup_{ \psi \in L^2(\BR^{2d}), \lVert \psi \rVert=1} \int_{\BR^{2d}}  B_{T}\left(p,q\right)\chi_{\vert \tilde p\vert >2 \sqrt{\mu}}\vert F_1 V^{1/2}\psi(p,q) \vert^2 \dd q \dd p
\end{equation}
The right hand side defines a multiplication operator in $q$.
By \eqref{BT_bound} there is a constant $C>0$ independent of $T$ such that $\lVert G_T^>\rVert\leq C \lVert M\rVert$, where $M:=  V^{1/2} \frac{1}{1-\Delta} V^{1/2}$ on $L^2(\BR^d)$.
It follows from the Hardy-Littlewood-Sobolev and the H\"older inequalities that $M$ is a bounded operator \cite{lieb_analysis_2001,henheik_universality_2023,hainzl_bardeencooperschrieffer_2016}.

To bound $G_T^<$ note that for fixed $q$, $\lVert F_1 V^{1/2}\psi (\cdot, q )\rVert_\infty \leq C \lVert V \rVert_1^{1/2} \lVert \psi(\cdot, q ) \rVert_2$ by Lemma~\ref{lp_prop}\eqref{lp_prop.3}.
Therefore, we estimate
\begin{equation}\label{pf:3d_2_lea3.3}
\lVert G_T^< \rVert \leq C^2 \lVert V \rVert_1 \sup_{ \psi \in L^2(\BR^{2d}), \lVert \psi \rVert=1} \int_{\BR^{2d}}  \lVert \psi(\cdot, (p_1,\tilde q) ) \rVert_2 B_{T}(p,q)\chi_{\tilde p^2<2\mu} \lVert \psi(\cdot, q ) \rVert_2\dd p \dd q
\end{equation}
Since the right hand side defines a multiplication operator in $\tilde q$, we obtain
\begin{equation}
\lVert G_T^< \rVert \leq C^2 \lVert V \rVert_1 \sup_{\tilde q \in \BR^{d-1}} \sup_{ \psi \in L^2(\BR), \lVert \psi \rVert=1} \int_{\BR^{d+1}}  \overline{ \psi(p_1)} B_{T}(p,q)\chi_{\tilde p^2<2\mu}  \psi(q_1 ) \dd p \dd q_1,
\end{equation}
where for $d=1$ the supremum over $\tilde q$ is absent.
For $d=1$, the operator with integral kernel $B_{T}(p,q)$ is bounded uniformly in $T$ according to \cite[Lemma 4.2]{hainzl_boundary_nodate}, and thus  the claim follows.
For $d\in\{2,3\}$ we need to prove that the operators with integral kernel $\int_{\BR^{d-1}} B_{T}(p,q) \chi_{\vert \tilde p \vert <2 \sqrt{\mu}} \dd \tilde p$ are bounded uniformly in $\tilde q$ and $T$.
We apply the bound \cite[Lemma 4.6]{hainzl_boundary_nodate}
\begin{equation}\label{pf:3d_2_lea3.4}
B_{T}(p,q)\leq \frac{2}{\vert (p+q)^2-\mu \vert+ \vert (p-q)^2-\mu \vert}
\end{equation}
Then, we scale out $\mu$ and estimate the expression by pulling the supremum over $\psi$ into the $\tilde p$-integral
\begin{multline}\label{pf:3d_2_lea3.5}
\sup_{\tilde q\in \BR^{d-1}} \sup_{\psi \in L^2(\BR), \lVert \psi \rVert=1} \int_{\BR^{d+1}} \frac{2\chi_{\vert \tilde p \vert<2 \sqrt{\mu}}\overline{\psi(p_1)}\psi(q_1)}{\vert (p+q)^2-\mu \vert+ \vert (p-q)^2 -\mu \vert} \ \dd p\dd q_1 \\
=\mu^{d/2-1}\sup_{\tilde q\in \BR^{d-1}} \sup_{\psi \in L^2(\BR), \lVert \psi \rVert=1} \int_{\BR^{d+1}} \frac{2\chi_{\vert \tilde p \vert<2 }\overline{\psi(p_1)}\psi(q_1)}{\vert (p+q)^2-1\vert+ \vert (p-q)^2 -1 \vert} \ \dd  p \dd q_1 \\
\leq \mu^{d/2-1}\sup_{\tilde q\in \BR^{d-1}}  \int_{\BR^{d-1}} \chi_{\vert \tilde p \vert<2 } \left[ \sup_{\psi \in L^2(\BR), \lVert \psi \rVert=1}  \int_{\BR^2}\frac{2\overline{\psi(p_1)}\psi(q_1)}{\vert (p+q)^2-1\vert+ \vert (p-q)^2 -1 \vert} \ \dd p_1 \dd q_1\right] \dd \tilde p
\end{multline}

Let $\mu_1=1-(\tilde p+\tilde q)^2$ and $\mu_2=1-(\tilde p-\tilde q)^2$.
For fixed $\mu_1,\mu_2$ we need to bound the operator with integral kernel
\begin{equation}\label{Dmu1mu2}
D_{\mu_1,\mu_2}(p_1,q_1)=\frac{2}{\vert (p_1+q_1)^2-\mu_1\vert+ \vert (p_1-q_1)^2 -\mu_2\vert}.
\end{equation}

\begin{lemma}\label{Dmu1mu2_bound}
Let $\mu_1,\mu_2\leq 1$ with $\min\{\mu_1,\mu_2\}\neq 0$.
The operator $D_{\mu_1,\mu_2}$ on $L^2(\BR)$ with integral kernel given by \eqref{Dmu1mu2} satisfies
\begin{equation}
\lVert D_{\mu_1,\mu_2} \rVert \leq C(1+d(\mu_1,\mu_2)\vert \min\{\mu_1,\mu_2\} \vert^{-1/2})
\end{equation}
for some finite $C$ independent of $\mu_1,\mu_2$, where
\begin{equation}
d(\mu_1,\mu_2)=\left\{\begin{matrix} 1+\ln\left(1+\frac{\max\{\mu_1,\mu_2\} }{|\min\{\mu_1,\mu_2\}|}\right)  & {\rm if}\quad \min\{\mu_1,\mu_2\}<0\leq \max\{\mu_1,\mu_2\}, \\
1  & {\rm otherwise.} \end{matrix}\right.
\end{equation}
\end{lemma}
This is a generalization of \cite[Lemma 4.2]{hainzl_boundary_nodate}.
The proof of Lemma~\ref{Dmu1mu2_bound} is based on the Schur test and can be found in Section~\ref{pf:Dmu1mu2_bound}.
Since $\max\{\mu_1,\mu_2\} \leq 1$, it follows from Lemma~\ref{Dmu1mu2_bound} that for any $\alpha>1/2$ one has $\lVert D_{\mu_1,\mu_2} \rVert \leq C\left( 1+ \vert \min\{\mu_1,\mu_2\} \vert^{-\alpha}  \right)$ for a constant $C$ independent of $\mu_1,\mu_2$.
The following Lemma concludes the proof of $\sup_{T>0} \lVert G_T^<\rVert<\infty$.
\begin{lemma}\label{mu1mu2_bound}
Let $d\in\{2,3\}$ and $0\leq \alpha<1$. Let $\mu_1=1-(\tilde p+\tilde q)^2$ and $\mu_2=1-(\tilde p-\tilde q)^2$. Then
\begin{equation}\label{mu1mu2_bound_eq}
\sup_{\tilde q \in \BR^{d-1}} \int_{\BR^{d-1}} \frac{ \chi_{\vert \tilde p \vert<2 }}{\vert \min\{\mu_1,\mu_2\}\vert^{\alpha}}\dd \tilde p<\infty.
\end{equation}
\end{lemma}
Lemma~\ref{mu1mu2_bound} follows from elementary computations  carried out in Section~\ref{pf:mu1mu2_bound}.
\end{proof}

\subsection{Proof of Lemma~\ref{3d_w_lea4}}
\begin{proof}[Proof of Lemma~\ref{3d_w_lea4}]
With the notation introduced in the proof of Lemma~\ref{3d_w_lea3} we have $\lVert \BI_{\leq \eps} G_T  \BI_{\leq \eps} \rVert \leq \lVert \BI_{\leq \eps} G_T^<  \BI_{\leq \eps} \rVert +\lVert \BI_{\leq \eps} G_T^>  \BI_{\leq \eps} \rVert $.

For $d=2,3$ we have analogously to \eqref{pf:3d_2_lea3.2}
 \begin{equation}\label{pf:3d_2_lea4.1}
\lVert \BI_{\leq \eps} G_T^>  \BI_{\leq \eps} \rVert\leq  \sup_{ \psi \in L^2(\BR^{2d}), \lVert \psi \rVert=1} \int_{\BR^{2d}}   \chi_{|q|<\eps}\chi_{|(p_1,\tilde q)|<\eps}B_{T}\left(p,q\right)\chi_{\vert \tilde p\vert >2 \sqrt{\mu}}\vert F_1 V^{1/2}\psi(p,q) \vert^2\dd q \dd p  .
\end{equation}
Let $1<t<\infty$ such that $V\in L^t(\BR^d)$.
According to Lemma~\ref{lp_prop}\eqref{lp_prop.2}, for fixed $q$ we have
\begin{equation}
\lVert F_1 V^{1/2}\psi(\cdot, q) \rVert_{L^s(\BR^d)}\leq C \lVert V\rVert_t^{1/2} \lVert \psi(\cdot, q) \rVert_{L^2(\BR^d)},
\end{equation}
where $2\leq s=2t/(t-1)<\infty$.
By \eqref{BT_bound} and H\"older's inequality in $p$, there is a constant $C$ independent of $T$ such that
\begin{multline}
\lVert \BI_{\leq \eps} G_T^>  \BI_{\leq \eps} \rVert\leq  C\sup_{ \psi \in L^2(\BR^{2d}), \lVert \psi \rVert=1} \int_{\BR^{2d}}   \frac{\chi_{|p_1|<\eps}}{1+\tilde p^2}\vert F_1 V^{1/2}\psi(p,q) \vert^2 \dd p \dd q\\
\leq  C \lVert V\rVert_t \left(\int_{\BR^d} \frac{\chi_{|p_1|<\eps}}{(1+\tilde p^2)^t} \dd p\right)^{1/t}.
\end{multline}
In particular, the remaining integral is of order $O(\eps^{1/t})$ and vanishes as $\eps\to0$.

To estimate $\lVert \BI_{\leq \eps} G_T^<  \BI_{\leq \eps} \rVert$ we proceed as in the derivation of the bound on $\lVert G_T^<\rVert$ from \eqref{pf:3d_2_lea3.3} until the first line of \eqref{pf:3d_2_lea3.5} and obtain
\begin{equation}
\lVert \BI_{\leq \eps} G_T^<  \BI_{\leq \eps} \rVert \leq
C \lVert V \rVert_1  \sup_{|\tilde q|<\eps} \sup_{\psi \in L^2(\BR), \lVert \psi \rVert=1} \int_{\BR^{d+1}} \frac{2\chi_{|p_1|, |q_1|<\eps}\chi_{\vert \tilde p \vert<2 \sqrt{\mu}}\overline{\psi(p_1)}\psi(q_1)}{\vert (p+q)^2-\mu \vert+ \vert (p-q)^2 -\mu \vert} \ \dd p\dd q_1
\end{equation}
Hence, we need that the norm of the operator on $L^2(\BR)$ with integral kernel
\begin{equation}
\int_{\BR^{d-1}} \frac{2\chi_{|p_1|, |q_1|<\eps}\chi_{\vert \tilde p \vert<2 \sqrt{\mu}}}{\vert (p+q)^2-\mu \vert+ \vert (p-q)^2 -\mu \vert} \dd \tilde p
\end{equation}
vanishes uniformly in $\tilde q$ as $\eps\to0$.
In $d=1$, the Hilbert-Schmidt norm clearly vanishes as $\eps\to0$.
Similarly for $d=2,3$ the following Lemma implies that the Hilbert-Schmidt norm vanishes uniformly in $\tilde q$ as $\eps\to0$.
\begin{lemma}\label{lea:IGIto0}
Let $d\in\{2,3\}$. Then
\begin{equation}\label{pf_2d_w_lea4.2}
\lim_{\eps\to0} \sup_{|\tilde q|<\eps}\int_{\BR^2} \chi_{|p_1|, |q_1|<\eps} \left[\int_{\BR^{d-1}} \frac{2\chi_{\tilde p^2<2}}{\vert (p+q)^2-1 \vert+ \vert (p-q)^2-1 \vert} \dd \tilde p \right]^2 \dd p_1 \dd q_1=0
\end{equation}
\end{lemma}
The proof can be found in Section~\ref{sec:pf_lea:IGIto0}.
We give the proof for $d=2$ only; the one for $d=3$ works analogously and is left to the reader.
\end{proof}

\subsection{Proof of Lemma~\ref{3d_w_lea2}}
\begin{proof}[Proof of Lemma~\ref{3d_w_lea2}]
Since $\af$ diverges like $e_\mu \mu^{d/2-1}\ln(\mu/T) $ as $T\to0$, the claim follows if we prove that $\sup_{T>0} \sup_{|q|>\eps}\lVert V^{1/2} B_{T}(\cdot, q) V^{1/2}\rVert<\infty $.
For $d=1$ we have
\begin{multline}
\lVert V^{1/2} B_{T}(\cdot, q) V^{1/2}\rVert^2 \leq  \lVert V^{1/2} B_{T}(\cdot, q) V^{1/2} \rVert_{\rm{HS}}^2\\
= \int_{\BR^2} V(r) V(r') \left(\int_\BR B_{T}(p,q) \frac{e^{i p(r- r')} }{2\pi}\dd p\right)^2 \dd r \dd r' \leq \frac{1}{(2\pi)^2}\lVert V\rVert_1^2  \left(\int_\BR B_{T}(p,q)\dd p\right)^2
\end{multline}
It was shown in the proof of \cite[Lemma 4.4]{hainzl_boundary_nodate} that $\sup_{T>0,|q|>\epsilon} \int_\BR B_{T}(p,q)\dd p<\infty$ .

For $d\in\{2,3\}$, the claim follows from the following Lemma which is proved below.
\begin{lemma}\label{B_q_aspt}
Let $d\in\{2,3\}$ and $\mu>0$.
Let $V$ satisfy Assumption~\ref{aspt_V_halfspace} and $V\geq 0$.
Recall that $O_\mu=V^{1/2} \FT^\dagger \FT  V^{1/2}$ (defined above \eqref{mmu}).
Let $f(x)=\chi_{(0,1/2)}(x)\ln(1/x)$.
There is a constant $C(d,\mu,V)$ such that for all $T>0$, $q\in \BR^d$, and $\psi \in L^2(\BR^d)$ with $\lVert \psi \rVert_2=1$
\begin{equation}\label{supdeltaVBV}
\langle \psi, V^{1/2} B_{T}(\cdot, q) V^{1/2} \psi \rangle
\leq \mu^{d/2-1} \langle \psi, O_\mu \psi \rangle f(\max\{T/\mu, \vert q \vert/\sqrt{\mu}\}) + C(d,\mu,V).
\end{equation}
\end{lemma}
This concludes the proof.
\end{proof}

\begin{proof}[Proof of Lemma~\ref{B_q_aspt}]
Note that if we set $q=0$, and optimize over $\psi$, the left hand side would have the asymptotics $a_{T,\mu}^0\sim e_\mu\mu^{d/2-1}  \ln(1/T) $ as $T \to 0$.
Intuitively, keeping $q$ away from $0$ on a scale larger than $T$ will slow down the divergence.
In the case $q=0$, divergence comes from the singularity on the set $\vert p \vert =\sqrt{\mu}$.
For $\vert q\vert >0$, there will be two relevant sets, $(p+q)^2=\mu$ and $(p-q)^2=\mu$.
These sets are circles or spheres in 2d and 3d, respectively.
The function $B_{T}$ is very small on the region which lies inside exactly one of the disks or balls (see the shaded area in Figure~\ref{B_q_aspt_domains}).
The part lying inside or outside both disks (the white area in Figure~\ref{B_q_aspt_domains}) will be relevant for the asymptotics.
Define the family of operators $Q_{T}(q): L^1(\BR^d)\to L^\infty(\BR^d)$ for $q\in \BR^d$ through
\begin{equation}
\langle \psi, Q_{T}(q) \psi \rangle
=\chi_{\max\left\{\frac{T}{\mu},\frac{\vert q \vert}{\sqrt{\mu}}\right\} < \frac{1}{2}}\int_{\BR^d}\left\vert\widehat \psi \left(\sqrt{\mu}p/\vert p \vert\right) \right\vert^2 B_{T}(p,q)\chi_{((p+q)^2-\mu)((p-q)^2-\mu)>0} \chi_{p^2 < 3\mu} \dd p.
\end{equation}
We claim that $Q_{T}$ captures the divergence of $B_{T}$.
\begin{lemma}\label{B-Q_gen}
Let $d\in\{2,3\}$ and $\mu>0$.
Let $V$ satisfy Assumption~\ref{aspt_V_halfspace}.
Then
\begin{equation}\label{B-Q_eq}
\sup_{T>0} \sup_{q \in \BR^d} \lVert V^{1/2} B_{T}(\cdot, q) |V|^{1/2}- V^{1/2} Q_{T}(q)  |V|^{1/2}\rVert <\infty.
\end{equation}
\end{lemma}
The proof of Lemma~\ref{B-Q_gen} can be found in Section~\ref{sec:pf_B-Q}.
It now suffices to prove that there is a constant $C$ such that for all $T> 0$ and $q\in \BR^d$
\begin{equation}\label{Q-est-5}
\langle \psi, Q_{T}(q) \psi \rangle \leq  \mu^{d/2-1}\langle \psi, \mathcal{F^\dagger}\mathcal{F} \psi\rangle f(\max\{T/\mu, \vert q \vert/\sqrt{\mu}\}) + C \lVert \psi \rVert_1^2.
\end{equation}
Then for all $\psi \in L^2(\BR^d)$ with $\lVert \psi \rVert_2=1$
\begin{equation}
 \langle \psi, V^{1/2} Q_{T}(q) V^{1/2} \psi \rangle \leq  \mu^{d/2-1}\langle \psi, O_\mu \psi \rangle f(\max\{T/\mu, \vert q \vert/\sqrt{\mu}\}) + C \lVert V \rVert_1
\end{equation}
and the claim follows with Lemma~\ref{B-Q_gen}.

We are left with proving \eqref{Q-est-5}.
By the definition of $Q_T$, it suffices to restrict to $\vert q \vert <\sqrt{\mu}/2, T<\mu/2$.
Let $R$ be the rotation in $\BR^d$ around the origin such that $q=R(\vert q \vert, \tilde 0)$.
For $d=2$ the condition $((p+(|q|,0))^2-\mu)((p-(|q|,0))^2-\mu)>0$ holds exactly in the white region sketched in Figure~\ref{B_q_aspt_domains}.
The inner white region is characterized by $(\vert p_1\vert +\vert q\vert)^2+\tilde p^2<\mu$, and the outer region by $(\vert p_1\vert -\vert q\vert)^2+\tilde p^2>\mu$.
Thus,
\begin{equation}\label{Q-est-2}
\langle \psi, Q_{T}(q) \psi \rangle
= \int_{\BR^d}\left\vert\widehat \psi \left(\sqrt{\mu}R p/\vert p \vert\right) \right\vert^2 \left[  \chi_{(\vert p_1\vert +\vert q\vert)^2+\tilde p^2<\mu}+\chi_{(\vert p_1\vert -\vert q\vert)^2+\tilde p^2>\mu}\right] B_{T}(p,(\vert q\vert,\tilde 0))\chi_{p^2 < 3\mu} \dd p,
\end{equation}
where we substituted $p$ by $R p$.

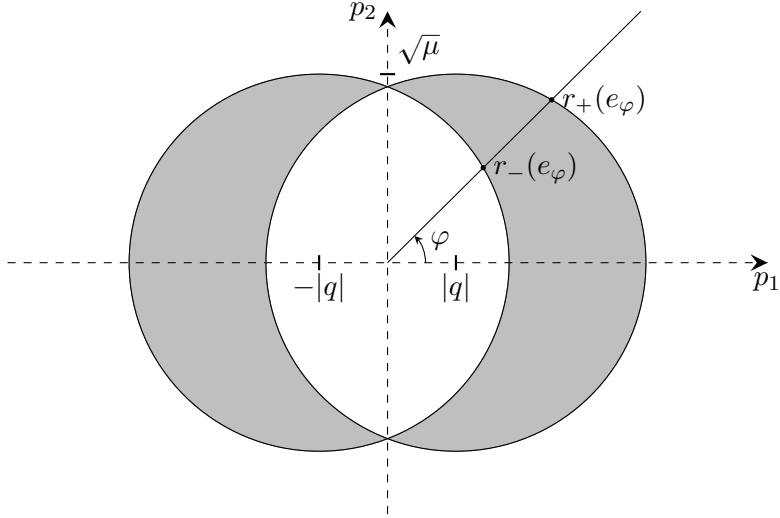
\begin{figure}
\centering
\begin{tikzpicture}[scale=0.5]
\tikzmath{\smu=5;\l=10;\q=1.8;\t=35;\d=0.2;}

\coordinate (end) at (2*\l/3,2*\l/3);
\coordinate (or) at (0,0);
\coordinate (xax) at (1,0);

 \draw[fill, fill opacity=0.5, gray,even odd rule] (-\q,0) circle (\smu) (\q,0) circle (\smu);
\draw[black, name path= c1] (-\q,0) circle (\smu);
\draw[black, name path =c2]  (\q,0) circle (\smu);

\draw (0,2*\l/3) node[left]{$p_2$};
\draw (\l,0) node[below]{$p_1$};
\draw (\q,0) node[below]{$\vert q \vert$};
\draw (-\q,0) node[below]{$-\vert q \vert$};
\draw (0,\smu) node[above right]{$\sqrt{\mu}$};

\draw[black, thick] ($ (\q,0)+(0,\d)$) -- ($ (\q,0)-(0,\d)$);
\draw[black, thick] ($ (-\q,0)+(0,\d)$) -- ($ (-\q,0)-(0,\d)$);
\draw[black, thick] ($ (0,\smu)+(\d,0)$) -- ($ (0,\smu)-(\d,0)$);

 \draw[black, dashed, decoration={markings, mark=at position 1 with {\arrow[scale=2,>=stealth]{>}}},
        postaction={decorate}]
 (-\l,0)--(\l,0);
 \draw[black, dashed, decoration={markings, mark=at position 1 with {\arrow[scale=2,>=stealth]{>}}},
        postaction={decorate}]
 (0,-2*\l/3)--(0,2*\l/3);
 \draw[black, name path= tilt] (or)--(end);

 \pic [draw, ->, "$\varphi$", angle eccentricity=1.5] {angle = xax--or--end};

\path [name intersections={of=c1 and tilt}];
\coordinate (OP1) at (intersection-1);

\path [name intersections={of=c2 and tilt}];
\coordinate (OP2) at (intersection-1);

\draw (OP1) node[right]{$r_{-}(e_\varphi)$};
\draw (OP2) node[right]{$r_{+}(e_\varphi)$};

\foreach \point in {OP1,OP2}
  \fill (\point) circle (2pt);

\end{tikzpicture}

\caption{Two circles of radius $\sqrt{\mu}$, centered at $(-\vert q \vert,0)$ and $(\vert q \vert,0)$. In $d=2$ the function $B_T(p,(\vert q \vert,0))$ diverges on the two circles as $T\to0$ and approaches zero in the shaded area. Given an angle $\varphi$, the numbers $r_\pm(e_\varphi)$ are the distances between zero and the intersections of the circles with the ray tilted by an angle $\varphi$ with respect to the $p_1$-axis.}
\label{B_q_aspt_domains}

\end{figure}
Let us use the notation $r_{\pm}(e)=\pm \vert e_1 \vert \vert q \vert +\sqrt{\mu-e_2^2 \vert q\vert^2}$ and $e_\varphi=(\cos \varphi,\sin \varphi)$, where the choice of $r_\pm$ is motivated in Figure~\ref{B_q_aspt_domains}.
For $d=2$ rewriting the integral \eqref{Q-est-2} in angular coordinates gives
\begin{equation}\label{Q-est-1}
\int_{0}^{2\pi} \left\vert\widehat \psi\left(\sqrt{\mu}R e_\varphi \vert\right) \right\vert^2 \left[  \int_0^{r_-(e_\varphi)} B_{T}(r e_\varphi ,(\vert q\vert,0)) r\dd r +\int_{r_+(e_\varphi)}^{\sqrt{3\mu}} B_{T}(r e_\varphi ,(\vert q\vert,0)) r\dd r \right] \dd \varphi.
\end{equation}
For $d=3$ with the notation $e_{\varphi,\theta} =(\cos \varphi, \sin \varphi \cos \theta, \sin \varphi \sin \theta)$ and using that $ B_{T}(r e_{\varphi,\theta} ,(\vert q\vert,0,0))= B_{T}(r e_\varphi ,(\vert q\vert,0))$, \eqref{Q-est-2} equals
\begin{equation}\label{Q-est-3}
\int_{0}^{\pi}\left(\int_0^{2\pi} \left\vert\widehat \psi\left(\sqrt{\mu}r e_{\varphi,\theta} \vert\right) \right\vert^2 \dd \theta\right) \left[  \int_0^{r_-(e_\varphi)} B_{T}(r e_\varphi ,(\vert q\vert,0)) r^2\dd r +\int_{r_+(e_\varphi)}^{\sqrt{3\mu}} B_{T}(r e_\varphi ,(\vert q\vert,0)) r^2\dd r \right] \sin \varphi   \dd \varphi.
\end{equation}

We distinguish two cases depending on whether $r$ is within distance $T/\sqrt{\mu}$ to $r_\pm$ or not.
Note that $r_{-}(e) \geq -\vert q\vert+\sqrt{\mu}\geq \frac{\sqrt{\mu}}{2} \geq \frac{T}{\sqrt{\mu}}$ and $r_+(e)+\frac{T}{\sqrt{\mu}}\leq \vert q\vert+\sqrt{\mu}+T\leq 2\sqrt{\mu} $.
If $r$ is close to $r_\pm$ we use that $B_{T}(p,q)\leq 1/2T$.
Otherwise we use \eqref{pf:3d_2_lea3.4}.
The expressions in the square brackets in \eqref{Q-est-1} and \eqref{Q-est-3} are thus bounded by
\begin{equation}\label{Q-est-4}
 \int_0^{r_-(e_\varphi)-\frac{T}{\sqrt{\mu}}} \frac{r^{d-1}}{\mu-r^2-q^2}\dd r + \int_{r_-(e_\varphi)-\frac{T}{\sqrt{\mu}}}^{r_-(e_\varphi)} \frac{r^{d-1}}{2T}\dd r
+ \int_{r_+(e_\varphi)}^{r_+(e_\varphi)+\frac{T}{\sqrt{\mu}}} \frac{r^{d-1}}{2T}\dd r+\int_{r_+(e_\varphi)+\frac{T}{\sqrt{\mu}}}^{\sqrt{3\mu}} \frac{r^{d-1}}{r^2+q^2-\mu}\dd r
\end{equation}
The second and third term are clearly bounded for $T<\mu/2$.
Since $\lVert \widehat \psi \rVert_\infty \leq (2\pi)^{-d/2} \lVert \psi \rVert_1$, they contribute $C \lVert \psi \rVert_1$ to the upper bound on $\langle \psi, Q_{T}(q) \psi \rangle $.

To bound the contributions of the first and the last term in \eqref{Q-est-4} we treat $d=2$ and $d=3$ separately.

{\bf Case $d=2$:}
The sum of the two integrals equals
\begin{equation}\label{Q-est}
 \ln \sqrt{\frac{(\mu-q^2)(2\mu +q^2)}{(\mu-q^2-(r_-(e_\varphi)-\frac{T}{\sqrt{\mu}})^2)((r_+(e_\varphi)+\frac{T}{\sqrt{\mu}})^2+q^2-\mu)}}
\end{equation}
To bound this expression, we first make a few observations.
Note that
\begin{multline}\label{Q-est-6}
\mu-q^2-\left(r_-(e_\varphi)-\frac{T}{\sqrt{\mu}}\right)^2=2 \vert e_1 \vert \vert q \vert (\sqrt{\mu-e_2^2 \vert q \vert^2}-\vert e_1\vert \vert q\vert) + \frac{T}{\sqrt{\mu}}\left (2r_-(e_\varphi) -\frac{T}{\sqrt{\mu}}\right) \\
\geq (\sqrt{3}-1) \sqrt{\mu}\vert e_1 \vert \vert q \vert  + \frac{T}{2},
\end{multline}
where we used that $r_-(e_\varphi)\geq \sqrt{\mu}-\vert q\vert$ and $\vert q\vert, T/\sqrt{\mu}\leq \sqrt{\mu}/2$.
Similarly,
\begin{multline}\label{Q-est-7}
\left(r_+(e_\varphi)+\frac{T}{\sqrt{\mu}}\right)^2+q^2-\mu=2 \vert e_1 \vert \vert q \vert (\sqrt{\mu-e_2^2 \vert q \vert^2}+\vert e_1\vert \vert q\vert) + \frac{T}{\sqrt{\mu}}\left (2r_+(e_\varphi) +\frac{T}{\sqrt{\mu}}\right) \\
\geq \sqrt{3} \sqrt{\mu} \vert e_1 \vert \vert q \vert + \sqrt{3} T
\end{multline}
Furthermore, note that $2\mu +q^2 \leq \frac{5\mu}{4}$.
The expression under the square root in \eqref{Q-est} is therefore bounded above by
\begin{equation}
\frac{ 5\mu^2}{4( (\sqrt{3}-1) \sqrt{\mu}\vert e_1 \vert \vert q \vert+ \frac{T}{2})(\sqrt{3} \sqrt{\mu} \vert e_1 \vert \vert q \vert + \sqrt{3} T)}
\end{equation}
We now bound this from above in two ways.
First we drop the $T$ terms in the denominator, and second we drop the other terms in the denominator, which gives
$\frac{ 5\mu }{4\sqrt{3} (\sqrt{3}-1) \vert e_1 \vert^2 \vert q \vert^2}$ and $\frac{5\mu^2 }{2\sqrt{3}T^2}$, respectively.
Thus, \eqref{Q-est} is bounded above by $ f(\max\{T/\mu, \vert q \vert/\sqrt{\mu}\}) + \ln(1/\vert e_1 \vert)+C$.
The contribution to the upper bound on $\langle \psi, Q_{T}(q) \psi \rangle $ is
\begin{equation}
\int_{0}^{2\pi} \left\vert\widehat \psi \left(\sqrt{\mu} e_\varphi \vert\right) \right\vert^2 f(\max\{T/\mu, \vert q \vert/\sqrt{\mu}\}) \dd \varphi
+(2\pi)^{-2}\lVert \psi \rVert_1^2 \int_{0}^{2\pi}\left( \ln \left(1/\vert \cos \varphi \vert \right)+C \right)\dd \varphi,
\end{equation}
where for the second term we used that $\vert\widehat \psi \left(\sqrt{\mu} e_\varphi \vert\right) \vert^2 \leq (2\pi)^{-2}\lVert \psi \rVert_1^2$.
Note that the first summand equals $\langle \psi, \mathcal{F^\dagger}\mathcal{F} \psi\rangle  f(\max\{T/\mu, \vert q \vert/\sqrt{\mu}\})$ and that the integral in the second summand is finite.
In total, we have obtained \eqref{Q-est-5} for $d=2$.

{\bf Case $d=3$:}
Note that $\frac{\dd}{\dd r} (-r + a \artanh (r/a))= r^2/(a^2-r^2)$ and $\frac{\dd}{\dd r} (r - a \arcoth (r/a))= r^2/(r^2-a^2)$.
The sum of the first and the last integral in \eqref{Q-est-4} hence equals
\begin{multline}
\sqrt{3\mu}-r_+(e_\varphi)-r_-(e_\varphi)-\frac{\sqrt{\mu-q^2}}{2}\ln\left(\frac{(\sqrt{\mu-q^2}+\sqrt{3\mu})}{(\sqrt{3\mu}-\sqrt{\mu-q^2})}\right)\\
+\frac{\sqrt{\mu-q^2}}{2}\ln\left(\frac{(\sqrt{\mu-q^2}+r_-(e_\varphi)-\frac{T}{\sqrt{\mu}})}{(\sqrt{\mu-q^2}-r_-(e_\varphi)+\frac{T}{\sqrt{\mu}})}\frac{(\sqrt{\mu-q^2}+r_+(e_\varphi)+\frac{T}{\sqrt{\mu}})}{(r_+(e_\varphi)+\frac{T}{\sqrt{\mu}}-\sqrt{\mu-q^2})}\right)
\end{multline}
The terms in the first line are bounded.
The argument of the logarithm in the second line equals
\begin{multline}
\frac{(\sqrt{\mu-q^2}+r_-(e_\varphi)-\frac{T}{\sqrt{\mu}})^2}{(\mu-q^2-(r_-(e_\varphi)-\frac{T}{\sqrt{\mu}})^2)}\frac{(\sqrt{\mu-q^2}+r_+(e_\varphi)+\frac{T}{\sqrt{\mu}})^2}{((r_+(e_\varphi)+\frac{T}{\sqrt{\mu}})^2-\mu+q^2)}\\
\leq \frac{C\mu^2}{( (\sqrt{3}-1) \sqrt{\mu}\vert e_1 \vert \vert q \vert+ \frac{T}{2})(\sqrt{3} \sqrt{\mu} \vert e_1 \vert \vert q \vert + \sqrt{3} T))}
\end{multline}
where we used \eqref{Q-est-6} and \eqref{Q-est-7}.
Analogously to the case $d=2$ the contribution to the upper bound on $\langle \psi, Q_{T}(q) \psi \rangle $ is
\begin{multline}
\sqrt{\mu}\int_{0}^{\pi}\left( \int_0^{2\pi} \left\vert\widehat \psi \left(\sqrt{\mu} e_{\varphi,\theta} \vert\right) \right\vert^2\dd \theta \right) f(\max\{T/\mu, \vert q \vert/\sqrt{\mu}\})\sin \varphi \dd \varphi\\
+(2\pi)^{-2}\sqrt{\mu}\lVert \psi \rVert_1^2 \int_{0}^{\pi}\left( \ln \left(1/\vert \cos \varphi \vert \right)+C \right) \sin \varphi \dd \varphi
\end{multline}
and \eqref{Q-est-5} follows.
\end{proof}

\section{Proofs of Auxiliary Lemmas}\label{sec:pf_aux}
\subsection{Proof of Lemma~\ref{Dmu1mu2_bound}}\label{pf:Dmu1mu2_bound}
\begin{proof}[Proof of Lemma~\ref{Dmu1mu2_bound}]
If we write $D_{\mu_1,\mu_2}$ as a sum $D_{\mu_1,\mu_2}=\sum_{j=1}^n D_{\mu_1,\mu_2}^j$ a.e. for some integral kernels $D_{\mu_1,\mu_2}^j$, then $\lVert D_{\mu_1,\mu_2} \rVert \leq \sum_{j=1}^n \lVert D_{\mu_1,\mu_2}^j \rVert$.
We will choose the $D_{\mu_1,\mu_2}^j$ as localized versions of $D_{\mu_1,\mu_2}$ in different regions (by multiplying $D_{\mu_1,\mu_2}$ by characteristic functions).

Let $D_{\mu_1,\mu_2}^1=D_{\mu_1,\mu_2}\chi_{\max\{\vert p_1 \vert, \vert q_1\vert \}>2}$ and $D_{\mu_1,\mu_2}^2=D_{\mu_1,\mu_2}\chi_{\max\{\vert p_1 \vert, \vert q_1\vert \}<2}$.
We first prove that the Hilbert-Schmidt norm of $ D_{\mu_1,\mu_2}^1 $ is bounded uniformly in $\mu_1, \mu_2$.
Note that if $\max\{\vert p_1 \vert, \vert q_1\vert \}>2$, we have $\max\{(p_1\pm q_1)^2\}=(\vert p_1 \vert+\vert q_1 \vert)^2>4$ and $\mu_1,\mu_2\leq1$.
Hence,
\begin{equation}
D_{\mu_1,\mu_2}^1(p_1,q_1)\leq \frac{2\chi_{\max\{\vert p_1 \vert, \vert q_1\vert \}>2}}{(\vert p_1 \vert+\vert q_1 \vert)^2-1}\leq  \frac{2\chi_{\max\{\vert p_1 \vert, \vert q_1\vert \}>2}}{p_1^2+q_1^2-1}.
\end{equation}
For the Hilbert-Schmidt norm we obtain
\begin{equation}
\lVert D_{\mu_1,\mu_2}^1 \rVert_{ \rm HS}^2\leq 4 \int_{\BR^2} \frac{\chi_{\max\{\vert p_1 \vert, \vert q_1\vert \}>2}}{(p_1^2+q_1^2-1)^2} \dd p_1 \dd q_1\leq 8 \pi \int_2^\infty \frac{r}{(r^2-1)^2} \dd r =\frac{4\pi}{3},
\end{equation}
and therefore $\lVert D_{\mu_1,\mu_2}^1 \rVert$ is indeed bounded uniformly in $\mu_1,\mu_2$.

For $D_{\mu_1,\mu_2}^2$ we first observe that $\lVert D_{\mu_2,\mu_1}^2 \rVert =\lVert D_{\mu_1,\mu_2}^2 \rVert $ since $D_{\mu_1,\mu_2}^2(p_1,q_1)=D_{\mu_2,\mu_1}^2(p_1,-q_1)$.
Hence, without loss of generality we may assume $\mu_1\leq\mu_2$ from now on.
To bound the norm of $D_{\mu_1,\mu_2}^2$ we distinguish the cases $\mu_1<0$ and $\mu_1>0$ and continue localizing.

{\bf Case $\mu_1<0$:}
We localize in the regions $\vert p_1-q_1 \vert^2<\mu_2$ and $\vert p_1-q_1 \vert^2>\mu_2$, where the first one only occurs if $\mu_2>0$.
Let $D_{\mu_1,\mu_2}^3=D_{\mu_1,\mu_2}^2 \chi_{\vert p_1-q_1 \vert^2<\mu_2}$ and $D_{\mu_1,\mu_2}^4=D_{\mu_1,\mu_2}^2 \chi_{\vert p_1-q_1 \vert^2>\mu_2}$.

For $D_{\mu_1,\mu_2}^3$ we do a Schur test with test function $h(p_1)=\vert p_1\vert^{1/2}$.
Using the symmetry of the integrand under $(p_1,q_1)\to -(p_1,q_1)$, we have
\begin{multline}
\lVert D_{\mu_1,\mu_2}^3 \rVert \leq\sup_{-2<p_1<2} \vert p_1\vert^{1/2} \int_{-2}^2  \frac{1}{2}\frac{\chi_{\vert p_1-q_1 \vert^2<\mu_2}}{p_1q_1+(\mu_2-\mu_1)/4}\frac{1}{\vert q_1\vert^{1/2}}\dd q_1\\
= \chi_{0<\mu_2}\sup_{0\leq p_1<2} \vert p_1\vert^{1/2} \int_{p_1-\sqrt{\mu_2}}^{p_1+\sqrt{\mu_2}} \frac{1}{2} \frac{1}{p_1q_1+(\mu_2-\mu_1)/4}\frac{1}{\vert q_1\vert^{1/2}}\dd q_1.
\end{multline}
For $\mu_2>0$, carrying out the integration we obtain
\begin{multline}
\lVert D_{\mu_1,\mu_2}^3 \rVert \leq\sup_{0\leq p_1<2} \frac{2}{\sqrt{\mu_2-\mu_1}}\left[\arctan \left(\sqrt{\frac{4p_1 (p_1+\sqrt{\mu_2})}{\mu_2-\mu_1}}\right)\right.\\
\left.-\chi_{p_1>\sqrt{\mu_2}}\arctan \left(\sqrt{\frac{4p_1 (p_1-\sqrt{\mu_2})}{\mu_2-\mu_1}}\right)+\chi_{p_1<\sqrt{\mu_2}}\artanh \left(\sqrt{\frac{4p_1 (\sqrt{\mu_2}-p_1)}{\mu_2-\mu_1}}\right)\right]\\
\leq \frac{2}{\sqrt{\mu_2-\mu_1}}\left[\frac{\pi}{2}+\artanh \left(\sqrt{\frac{\mu_2}{\mu_2-\mu_1}}\right)\right],
\end{multline}
where we used the monotonicity of $\artanh$.
Note that for $x\geq 0$,
\begin{equation}
\artanh\left(\sqrt{\frac{1}{1+x}}\right)=\ln \left(\sqrt{\frac{1}{x}+1}+\sqrt{\frac{1}{x}}\right)\leq \ln \left(2 \sqrt{\frac{1}{x}+1}\right)= \ln(2)+\frac{1}{2}\ln\left(1+\frac{1}{x}\right).
\end{equation}
In total, we obtain
\begin{equation}
\lVert D_{\mu_1,\mu_2}^3\rVert  \leq \frac{C}{\sqrt{-\mu_1}}\left(1+\ln \left(1+\frac{\mu_2}{-\mu_1}\right)\right)
\end{equation}
for some constant $C$.

The Hilbert-Schmidt norm of $D_{\mu_1,\mu_2}^4$ is given by
\begin{equation}
\lVert D_{\mu_1,\mu_2}^4 \rVert_{\rm HS}= \left(\int_{(-2,2)^2} \frac{\chi_{\vert p_1-q_1 \vert^2>\mu_2}}{(p_1^2+q_1^2-\frac{\mu_1+\mu_2}{2})^2} \dd p_1 \dd q_1\right)^{1/2}
\end{equation}
For $\mu_2<0$, we clearly have $\lVert D_{\mu_1,\mu_2}^4 \rVert_{\rm HS}\leq \lVert D_{\mu_1,0}^4 \rVert_{\rm HS}$.
For $\mu_2\geq 0$ observe that the constraint $\vert p_1-q_1 \vert^2>\mu_2$ implies $p_1^2+q_1^2  >\frac{\mu_2}{2}$.
Hence,
\begin{equation}
\lVert D_{\mu_1,\mu_2}^4 \rVert_{\rm HS}\leq \left(2\pi \int_{\sqrt{\frac{\mu_2}{2}}}^\infty \frac{r}{(r^2-\frac{\mu_1+\mu_2}{2})^2} \dd r\right)^{1/2}= \left(\frac{2\pi}{-\mu_1}\right)^{1/2}.
\end{equation}

{\bf Case $\mu_1>0$:}
We are left with estimating $D_{\mu_1,\mu_2}^2$ in the case that $\mu_1>0$.
First we sketch the location of the singularities of $D_{\mu_1,\mu_2}^2(p_1,q_1)$.
On each of the diagonal lines in Figure~\ref{claim4_domains}, one of the two terms $|(p_1+q_1)^2-\mu_1|,|(p_1-q_1)^2-\mu_2|$ in the denominator of $D_{\mu_1,\mu_2}^2(p_1,q_1)$ vanishes.
The function $D_{\mu_1,\mu_2}^2(p_1,q_1)$ thus has four singularities located at the crossings of the diagonal lines in Figure~\ref{claim4_domains}.
The coordinates of the singularities are $(p_1,q_1)\in\{(s_1,-s_2),(s_2,-s_1),(-s_1,s_2),(-s_2,s_1)\}$, where $s_1=\frac{\sqrt{\mu_1}+\sqrt{\mu_2}}{2}$, $s_2=\frac{\sqrt{\mu_2}-\sqrt{\mu_1}}{2}$.
Note that $s_1^2+s_2^2=\frac{\mu_1+\mu_2}{2}$ and $s_1 s_2=\frac{\mu_2-\mu_1}{4}$.

To bound $\lVert D_{\mu_1,\mu_2}^2\rVert$, the idea is to perform a Schur test with test function $h(p_1)=\min\{\vert \vert p_1 \vert -s_1 \vert^{1/2}, \vert \vert p_1 \vert -s_2 \vert^{1/2}\}$.
Since the behavior of $D_{\mu_1,\mu_2}^2(p_1,q_1)$ strongly depends on whether $|p_1+q_1|\gtrless \sqrt{\mu_1},|p_1-q_1|\gtrless \sqrt{\mu_2}$ and which singularity of $D_{\mu_1,\mu_2}^2$ is close to $p_1,q_1$, we distinguish the ten different regions sketched in Figure~\ref{claim4_domains}.
For $5\leq j \leq 14$, we define the operator $D_{\mu_1,\mu_2}^j$ to be localized in region $j$, $D_{\mu_1,\mu_2}^j=D_{\mu_1,\mu_2}^2 \chi_j$.
According to the Schur test,
\begin{equation}
\lVert D_{\mu_1,\mu_2}^j\rVert \leq \sup_{|p_1|<2}h(p_1)^{-1}\int_{-2}^2 D^j_{\mu_1,\mu_2}(p_1,q_1) h(q_1) \dd q_1.
\end{equation}
The bounds on $\lVert D_{\mu_1,\mu_2}^j\rVert$ we obtain from the Schur test are listed in Table~\ref{tab:1}.
In the following we prove all the bounds.

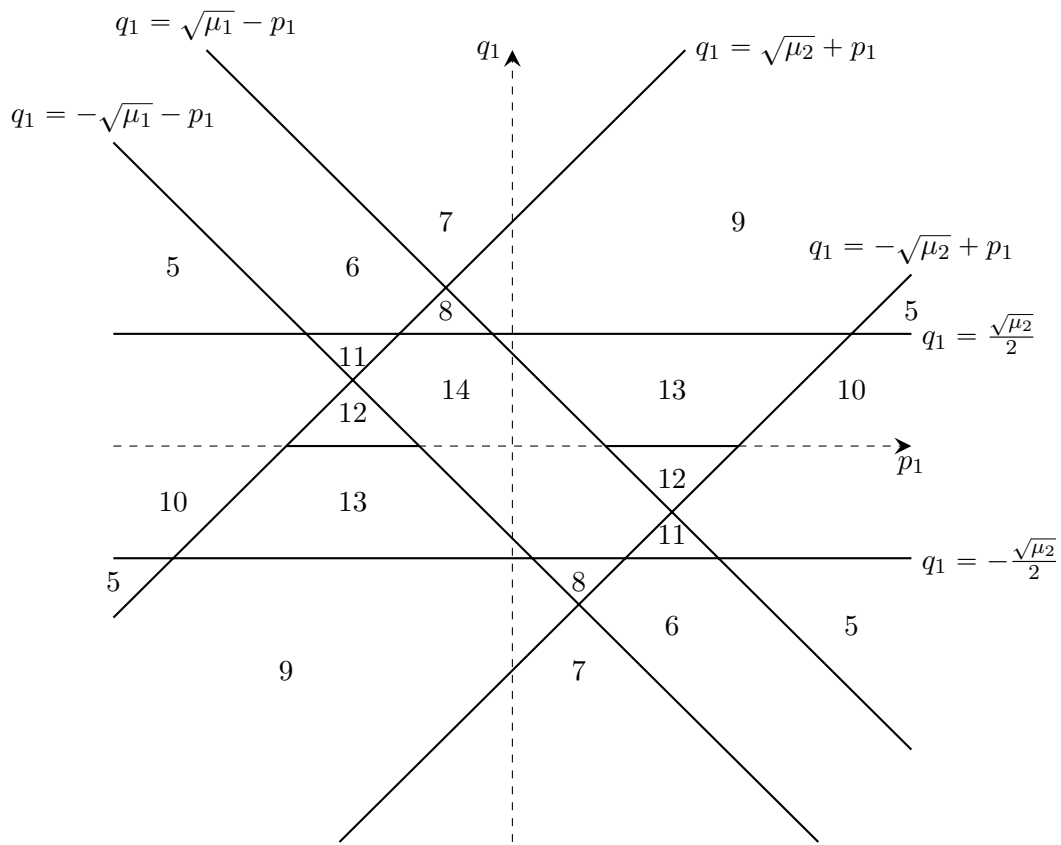
\begin{figure}
\centering
\begin{tikzpicture}[scale=0.35]
\tikzmath{\smu1=3.5;\smu2=8.5;\l=15;\t=0.3;}

\draw (0,\l) node[left]{$q_1$};
\draw (\l,0) node[below]{$p_1$};

\draw (-1.5*\smu2,0.8*\smu2) node{5};
\draw (1.5*\smu2,-0.8*\smu2) node{5};
\draw (-\smu2/2-\smu1/2,0.8*\smu2) node{6};
\draw (\smu2/2+\smu1/2,-0.8*\smu2) node{6};
\draw (-\smu2/2+\smu1/2,\smu2) node{7};
\draw (\smu2/2-\smu1/2,-\smu2) node{7};
\draw (-\smu2/2+\smu1/2,0.5*\smu2+0.25*\smu1) node{8};
\draw (\smu2/2-\smu1/2,-0.5*\smu2-0.25*\smu1) node{8};
\draw (\smu2,\smu2) node{9};
\draw (-\smu2,-\smu2) node{9};
\draw (\l,0.5*\smu2+0.25*\smu1) node{5};
\draw (-\l,-0.5*\smu2-0.25*\smu1) node{5};
\draw (-1.5*\smu2,-0.25*\smu2) node{10};
\draw (1.5*\smu2,0.25*\smu2) node{10};
\draw (-\smu2/2-\smu1/2,0.5*\smu2-0.25*\smu1) node{11};
\draw (\smu2/2+\smu1/2,-0.5*\smu2+0.25*\smu1) node{11};
\draw (-\smu2/2-\smu1/2,0.25*\smu2-0.25*\smu1) node{12};
\draw (\smu2/2+\smu1/2,-0.25*\smu2+0.25*\smu1) node{12};
\draw (-\smu2/2-\smu1/2,-0.25*\smu2) node{13};
\draw (\smu2/2+\smu1/2,0.25*\smu2) node{13};
\draw (-0.25*\smu2,0.25*\smu2) node{14};

 \draw[black, dashed, decoration={markings, mark=at position 1 with {\arrow[scale=2,>=stealth]{>}}},
        postaction={decorate}]
 (-\l,0)--(\l,0);
 \draw[black, dashed, decoration={markings, mark=at position 1 with {\arrow[scale=2,>=stealth]{>}}},
        postaction={decorate}]
 (0,-\l)--(0,\l);

 \draw[black, thick]
(\smu1-\l,\l)--(\l,\smu1-\l)
(-\l,-\smu1+\l)--(-\smu1+\l,-\l)
(\smu2-\l,-\l)--(\l,\l-\smu2)
(-\l,\smu2-\l)--(-\smu2+\l,\l)
 (-\l,\smu2/2)--(\l,\smu2/2)
 (-\l,-\smu2/2)--(\l,-\smu2/2)
(\smu2,0)--(\smu1,0)
(-\smu2,0)--(-\smu1,0);

\draw (\l,-\smu2/2) node[right]{$q_1=-\frac{\sqrt{\mu_2}}{2}$};
\draw (\l,\smu2/2) node[right]{$q_1=\frac{\sqrt{\mu_2}}{2}$};
\draw (\smu1-\l,\l) node[above]{$q_1=\sqrt{\mu_1}-p_1$};
\draw (-\l,\l-\smu1) node[above]{$q_1=-\sqrt{\mu_1}-p_1$};
\draw (-\smu2+\l,\l) node[right]{$q_1=\sqrt{\mu_2}+p_1$};
\draw (\l,\l-\smu2) node[above]{$q_1=-\sqrt{\mu_2}+p_1$};
\end{tikzpicture}
\caption{In the proof of Lemma~\ref{Dmu1mu2_bound}, in the case $0<\mu_1\leq\mu_2$ we split the domain of $p_1,q_1$ into ten different regions. The solid lines indicate the boundaries between these regions.}
\label{claim4_domains}
\end{figure}

\begin{table}[ht]
\renewcommand{\arraystretch}{1.4}
\begin{center}
\begin{tabular}{|c|c |c|}
\hline
 Operator & Upper bound & Proof \\
    \hline
    $D^5$ & $\frac{16}{\mu_1^{1/2}}$ &  \eqref{D5a}-\eqref{D5z} \\
    $D^6$ & $\frac{6}{\mu_1^{1/2}}$ &  \eqref{D6a}-\eqref{D6z}  \\
    $D^7$ & $\frac{(6+2\sqrt{2})^{1/2}}{\mu_1^{1/2}}$ & \eqref{D7a}-\eqref{D7z} \\
    $D^8$ & $\frac{2^{1/2} 4}{\mu_1^{1/2}}$  & \eqref{D8a} -\eqref{D8z}\\
    $D^9$ & $\frac{8}{\mu_1^{1/2}}$  & \eqref{D9a}-\eqref{D9z} \\
    $D^{10}$ &  $\frac{4}{\mu_1^{1/2}}$  & \eqref{D10a}-\eqref{D10z}\\
    $D^{11}$ & $ \frac{4}{\mu_1^{1/2}}$&  \eqref{D11a}- \eqref{D11z}\\
    $D^{12}$ & $\frac{2}{\mu_1^{1/2}}$ &  \eqref{D12a}- \eqref{D12z} \\
    $D^{13}$ & $\frac{4(\artanh(1/\sqrt{2})+\pi)}{\mu_1^{1/2}}$  &  \eqref{D13a}- \eqref{D13z} \\
    $D^{14}$ & $\frac{4(\sqrt{3}+1)}{\mu_1^{1/2}}$  &  \eqref{D14a}-\eqref{D14z} \\
    \hline
  \end{tabular}
\caption{Overview of the estimates used in the proof of Lemma~\ref{Dmu1mu2_bound}.}
\label{tab:1}
\end{center}
\end{table}

\textbf{Region 5:}
By symmetry of the integrand under $(p_1,q_1)\to -(p_1,q_1)$ we have
\begin{align*}
\lVert D_{\mu_1,\mu_2}^{5} \rVert & \leq \sup_{-2<p_1<2} h(p_1)\int_{-2}^2 \frac{\chi_{5}}{p_1^2+q_1^2-s_1^2-s_2^2} \frac{1}{h(q_1)}  \dd q_1\\
&= \sup_{\sqrt{\mu_1}+\frac{\sqrt{\mu_2}}{2}<p_1<2} \vert p_1-s_1 \vert^{1/2}\Bigg[\int_{\sqrt{\mu_1}-p_1}^{-\frac{\sqrt{\mu_2}}{2}} \frac{1}{p_1^2+q_1^2-s_1^2-s_2^2} \frac{1}{\vert q_1+s_1 \vert^{1/2}}  \dd q_1 \\
& \qquad+\int_{\sqrt{\mu_2}/2}^{p_1-\sqrt{\mu_2}} \frac{1}{p_1^2+q_1^2-s_1^2-s_2^2} \frac{1}{\vert q_1-s_1 \vert^{1/2}}  \dd q_1 \Bigg]\\
&\leq 2 \sup_{\sqrt{\mu_1}+\frac{\sqrt{\mu_2}}{2}<p_1<2} \vert p_1-s_1 \vert^{1/2} \int_{\sqrt{\mu_2}/2}^{p_1-\sqrt{\mu_1}} \frac{1}{p_1^2+q_1^2-s_1^2-s_2^2} \frac{1}{\vert q_1-s_1 \vert^{1/2}}  \dd q_1\\
&\leq 2\sup_{\sqrt{\mu_1}+\frac{\sqrt{\mu_2}}{2}<p_1<2}  \frac{\vert p_1-s_1 \vert^{1/2}}{p_1^2+\frac{\mu_2}{4}-s_1^2-s_2^2}\int_{\sqrt{\mu_2}/2}^{p_1-\sqrt{\mu_1}}  \frac{1}{\vert q_1-s_1 \vert^{1/2}}  \dd q_1 \nr \label{D5a}
\end{align*}
Note that $p_1^2+\frac{\mu_2}{4}-s_1^2-s_2^2=p_1^2-\frac{\mu_1}{2}-\frac{\mu_2}{4} \geq \frac{\sqrt{\mu_2}}{2} (p_1-\sqrt{\frac{\mu_1}{2}+\frac{\mu_2}{4}})$.
Carrying out the integration, \eqref{D5a} is bounded above by
\begin{equation}\label{D5b}
\frac{8}{\sqrt{\mu_2}} \sup_{\sqrt{\mu_1}+\frac{\sqrt{\mu_2}}{2}<p_1<2}  \frac{\vert p_1-s_1 \vert^{1/2}}{p_1-\sqrt{\frac{\mu_1}{2}+\frac{\mu_2}{4}}} \Bigg(\left( \frac{\sqrt{\mu_1}}{2}\right)^{1/2}+\chi_{p_1>s_1+\sqrt{\mu_1}}\vert p_1-s_1-\sqrt{\mu_1} \vert^{1/2}\Bigg)
\end{equation}
Note that $s_1>\sqrt{\frac{\mu_1}{2}+\frac{\mu_2}{4}}$.
Using that for $x\geq a\geq b$, $(x-a)/(x-b)\leq 1$ we bound \eqref{D5b} above by
\begin{equation}\label{D5z}
\frac{8}{\sqrt{\mu_2}} \Bigg(\frac{ \left( \frac{\sqrt{\mu_1}}{2}\right)^{1/2}}{|\sqrt{\mu_1}+\frac{\sqrt{\mu_2}}{2}-\sqrt{\frac{\mu_1}{2}+\frac{\mu_2}{4}}|^{1/2}} +1\Bigg)
\leq \frac{8}{\sqrt{\mu_2}}\left(\frac{ \sqrt{\mu_1}+\frac{\sqrt{\mu_2}}{2}+\sqrt{\frac{\mu_1}{2}+\frac{\mu_2}{4}}}{\sqrt{\mu_1}+2\sqrt{\mu_2}}\right)^{1/2} +\frac{8}{\sqrt{\mu_1}}
\leq \frac{16}{\sqrt{\mu_1}}.
\end{equation}

\textbf{Region 6:}
By symmetry under $(p_1,q_1)\to -(p_1,q_1)$, we obtain
\begin{align*}
\lVert D_{\mu_1,\mu_2}^{6} \rVert &\leq \sup_{-2<p_1<2} h(p_1)\int_{-2}^2 \frac{1}{2}\frac{\chi_{6}}{-p_1q_1-s_1s_2} \frac{1}{h(q_1)}  \dd q_1\\
&\leq \sup_{-2<p_1<-s_2} h(p_1) \int_{\max\{-\sqrt{\mu_1}-p_1,\frac{\sqrt{\mu_2}}{2},\sqrt{\mu_2}+p_1\}}^{\min\{-p_1+\sqrt{\mu_1},2\}} \frac{1}{-p_1q_1-s_1s_2} \frac{1}{\vert  q_1-s_1 \vert^{1/2}}  \dd q_1 \nr \label{D6a}
\end{align*}
We split the integral into the sum of the integral over $q_1>s_1$ and $q_1<s_1$.
For $p_1<-s_2$ and $q_1>s_1$ we have $-p_1q_1-s_1s_2>-(p_1+s_2)s_1$.
Hence,
\begin{multline}
\sup_{-2<p_1<-s_2} h(p_1) \int_{s_1}^{\min\{-p_1+\sqrt{\mu_1},2\}} \frac{1}{-p_1q_1-s_1s_2} \frac{1}{\vert  q_1-s_1 \vert^{1/2}}  \dd q_1\\
\leq \sup_{-2<p_1<-s_2} \frac{1}{\vert p_1+s_2 \vert^{1/2}s_1}\int_{s_1}^{-p_1+\sqrt{\mu_1}} \frac{1}{\vert  q_1-s_1 \vert^{1/2}}  \dd q_1= \frac{2}{s_1} \leq \frac{2}{\sqrt{\mu_1}}
\end{multline}
The case $q_1<s_1$ only occurs for $p_1>-s_1-\sqrt{\mu_1}$.
For $ -\frac{\sqrt{\mu_2}}{2}<p_1<-s_2$ and $\sqrt{\mu_2}+p_1<q_1<s_1$ note that $-p_1q_1-s_1s_2\geq -p_1 (\sqrt{\mu_2}+p_1)-s_1s_2=\vert p_1+s_2\vert(p_1+s_1)\geq \vert p_1+s_2\vert \frac{\sqrt{\mu_1}}{2}$.
Hence,
\begin{multline}
\sup_{-\frac{\sqrt{\mu_2}}{2}<p_1<-s_2} h(p_1) \int_{\sqrt{\mu_2}+p_1}^{s_1} \frac{1}{-p_1q_1-s_1s_2} \frac{1}{\vert  q_1-s_1 \vert^{1/2}}  \dd q_1\\
\leq \sup_{-\frac{\sqrt{\mu_2}}{2}<p_1<-s_2}\frac{2}{\sqrt{\mu_1}\vert p_1+s_2\vert^{1/2}} \int_{\sqrt{\mu_2}+p_1}^{s_1} \frac{1}{\vert  q_1-s_1 \vert^{1/2}}  \dd q_1= \frac{4}{\sqrt{\mu_1}}
\end{multline}
For $-s_1-\frac{\sqrt{\mu_1}}{2}<p_1< -\frac{\sqrt{\mu_2}}{2}$ and $\frac{\sqrt{\mu_2}}{2}<q_1<s_1$, we have $-p_1q_1-s_1s_2\geq \frac{\mu_2}{4}-s_1 s_2 =\frac{\mu_1}{4}$.
Therefore,
\begin{multline}
\sup_{-s_1-\frac{\sqrt{\mu_1}}{2}<p_1<-\frac{\sqrt{\mu_2}}{2}} h(p_1) \int_{\frac{\sqrt{\mu_2}}{2}}^{s_1} \frac{1}{-p_1q_1-s_1s_2} \frac{1}{\vert  q_1-s_1 \vert^{1/2}}  \dd q_1\\
\leq \sup_{-s_1-\frac{\sqrt{\mu_1}}{2}<p_1<-\frac{\sqrt{\mu_2}}{2}} \frac{4\vert p_1+s_1 \vert^{1/2}}{\mu_1} \int_{\frac{\sqrt{\mu_2}}{2}}^{s_1} \frac{1}{\vert  q_1-s_1 \vert^{1/2}}  \dd q_1
\leq \frac{8 \left(\frac{\sqrt{\mu_1}}{2}\right)^{1/2} }{\mu_1} \left(\frac{\sqrt{\mu_1}}{2}\right)^{1/2}=\frac{4}{\mu_1^{1/2}}
\end{multline}
For $-s_1-\sqrt{\mu_1}<p_1<-s_1-\frac{\sqrt{\mu_1}}{2}$ and $-p_1-\sqrt{\mu_1}<q_1<s_1$, we have $-p_1q_1-s_1s_2\geq p_1(p_1+\sqrt{\mu_1})-s_1 s_2 =-(p_1+s_1)(s_2-p_1)$.
Hence,
\begin{multline}
\sup_{-s_1-\sqrt{\mu_1}<p_1<-s_1-\frac{\sqrt{\mu_1}}{2}} h(p_1) \int_{-p_1-\sqrt{\mu_1}}^{s_1} \frac{1}{-p_1q_1-s_1s_2} \frac{1}{\vert  q_1-s_1 \vert^{1/2}}  \dd q_1\\
\leq \sup_{-s_1-\sqrt{\mu_1}<p_1<-s_1-\frac{\sqrt{\mu_1}}{2}} \frac{2\vert p_1+\sqrt{\mu_1}+s_1 \vert^{1/2}}{\vert p_1+s_1 \vert^{1/2}(s_2-p_1)}
= \frac{2}{s_2+s_1+\frac{\sqrt{\mu_1}}{2}}\leq \frac{4}{\sqrt{\mu_1}}
\end{multline}
In total, summing the contributions from $q_1>s_1$ and $q_1<s_1$ gives
\begin{equation}\label{D6z}
\lVert D_{\mu_1,\mu_2}^{6} \rVert \leq \frac{6}{\sqrt{\mu_1}}
\end{equation}

\textbf{Region 7:}
By symmetry of the two components of region 7 we have
\begin{align*}
\lVert D_{\mu_1,\mu_2}^7 \rVert &\leq \sup_{-2<p_1<2} h(p_1)\int_{-2}^2 \frac{\chi_7}{p_1^2+q_1^2-s_1^2-s_2^2} \frac{1}{h(q_1)}  \dd q_1\\
&\leq 2 \sup_{-2<p_1<2} \vert \vert p_1\vert-s_2 \vert^{1/2} \int_{\max\{\sqrt{\mu_1}-p_1, \sqrt{\mu_2}+p_1\}}^2 \frac{1}{p_1^2+q_1^2-s_1^2-s_2^2} \frac{1}{\vert q_1-s_1 \vert^{1/2}}  \dd q_1 \nr \label{D7a}
\end{align*}
For $\vert p_1 \vert>s_2$, $q_1>s_1$ we observe $p_1^2+q_1^2-s_1^2-s_2^2 \geq (q_1+s_1)(q_1-s_1) \geq 2s_1 (q_1-s_1)$.
Therefore,
\begin{multline}
\sup_{s_2<\vert p_1 \vert<2} \vert \vert p_1\vert-s_2 \vert^{1/2} \int_{\max\{\sqrt{\mu_1}-p_1, \sqrt{\mu_2}+p_1\}}^2 \frac{1}{p_1^2+q_1^2-s_1^2-s_2^2} \frac{1}{\vert q_1-s_1 \vert^{1/2}}  \dd q_1\\
\leq \sup_{s_2<\vert p_1 \vert<2} \frac{\vert \vert p_1\vert-s_2 \vert^{1/2}}{2s_1} \int_{\max\{\sqrt{\mu_1}-p_1, \sqrt{\mu_2}+p_1\}}^\infty \frac{1}{(q_1-s_1)^{3/2}} \dd q_1\\
=\sup_{s_2<\vert p_1 \vert<2}  \frac{\vert \vert p_1\vert-s_2 \vert^{1/2}}{s_1 (\max\{\sqrt{\mu_1}-p_1, \sqrt{\mu_2}+p_1\}-s_1)^{1/2}}=\frac{1}{s_1}\leq \frac{1}{\sqrt{\mu_1}}.
\end{multline}
For  $\vert p_1 \vert<s_2$, $q_1>s_1$ we have  $(p_1^2+q_1^2-s_1^2-s_2^2)(q_1-s_1)^{1/2} \geq (q_1+\sqrt{s_1^2+s_2^2-p_1^2})(q_1-\sqrt{s_1^2+s_2^2-p_1^2})^{3/2} \geq 2s_1 (q_1-\sqrt{s_1^2+s_2^2-p_1^2})^{3/2} $.
Hence,
\begin{align*}
\sup_{\vert p_1 \vert<s_2} &\vert \vert p_1\vert-s_2 \vert^{1/2} \int_{\max\{\sqrt{\mu_1}-p_1, \sqrt{\mu_2}+p_1\}}^2 \frac{1}{p_1^2+q_1^2-s_1^2-s_2^2} \frac{1}{\vert q_1-s_1 \vert^{1/2}}  \dd q_1\\
&\leq \sup_{\vert p_1 \vert<s_2} \frac{\vert p_1+s_2 \vert^{1/2}}{2s_1} \int_{\sqrt{\mu_2}+p_1}^\infty \frac{1}{(q_1-\sqrt{s_1^2+s_2^2-p_1^2})^{3/2} }\dd q_1\\
&= \sup_{\vert p_1 \vert<s_2} \frac{\vert p_1+s_2 \vert^{1/2}}{s_1} \frac{1}{(\sqrt{\mu_2}+p_1-\sqrt{s_1^2+s_2^2-p_1^2})^{1/2} }\\
&=\sup_{\vert p_1 \vert<s_2} \frac{1}{s_1} \frac{\vert p_1+s_2 \vert^{1/2}(\sqrt{\mu_2}+p_1+\sqrt{s_1^2+s_2^2-p_1^2})^{1/2} }{(p_1+s_1)^{1/2}(p_1+s_2)^{1/2}}\\
&=\sup_{\vert p_1 \vert<s_2} \frac{1}{s_1} \frac{(\sqrt{\mu_2}+p_1+\sqrt{s_1^2+s_2^2-p_1^2})^{1/2} }{(p_1+s_1)^{1/2}}\leq \frac{(\frac{3}{2}+\sqrt{2})^{1/2} s_1^{1/2}}{s_1 \mu_1^{1/4}}\leq  \frac{(\frac{3}{2}+\sqrt{2})^{1/2}}{\mu_1^{1/2}} \nr \label{D7z}
\end{align*}
In total, we obtain $\lVert D^7 \rVert \leq \frac{(6+2\sqrt{2})^{1/2}}{\sqrt{\mu_1}}$.

\textbf{Region 8:}
Taking the supremum separately over the two symmetric components of region 8, we have
\begin{align*}
\lVert D_{\mu_1,\mu_2}^{8} \rVert &\leq \sup_{-2<p_1<2} h(p_1)\int_{-2}^2 \frac{\chi_{8}}{s_1^2+s_2^2-p_1^2-q_1^2} \frac{1}{h(q_1)}  \dd q_1\\
&\leq 2 \sup_{- \frac{\sqrt{\mu_2}}{2}<p_1<\sqrt{\mu_1}-\frac{\sqrt{\mu_2}}{2}} h(p_1) \int_{\sqrt{\mu_2}/2}^{\min\{\sqrt{\mu_2}+p_1,\sqrt{\mu_1}-p_1\}} \frac{1}{s_1^2+s_2^2-p_1^2-q_1^2} \frac{1}{\vert s_1-q_1 \vert^{1/2}}  \dd q_1\\
&\leq 2\sup_{- \frac{\sqrt{\mu_2}}{2}<p_1<\sqrt{\mu_1}-\frac{\sqrt{\mu_2}}{2}}\frac{ h(p_1)}{\sqrt{\mu_2}} \int_{\sqrt{\mu_2}/2}^{\min\{\sqrt{\mu_2}+p_1,\sqrt{\mu_1}-p_1\}} \frac{1}{\sqrt{s_1^2+s_2^2-p_1^2}-q_1} \frac{1}{\vert s_1-q_1 \vert^{1/2}}  \dd q_1, \nr \label{D8a}
\end{align*}
since $\sqrt{s_1^2+s_2^2-p_1^2}+q_1> \sqrt{\frac{\mu_1}{2}+\frac{\mu_2}{2}-\frac{\mu_2}{4}}+\frac{\sqrt{\mu_2}}{2} \geq \sqrt{\mu_2}$.
For $\vert p_1 \vert >s_2$ we have $s_1>\sqrt{s_1^2+s_2^2-p_1^2}$, whereas for $\vert p_1 \vert<s_2$, $s_1<\sqrt{s_1^2+s_2^2-p_1^2}$.
For $p_1<-s_2$ we obtain
\begin{align*}
\sup_{- \frac{\sqrt{\mu_2}}{2}<p_1<-s_2}&\frac{2 h(p_1)}{\sqrt{\mu_2}} \int_{\sqrt{\mu_2}/2}^{\min\{\sqrt{\mu_2}+p_1,\sqrt{\mu_1}-p_1\}} \frac{1}{\sqrt{s_1^2+s_2^2-p_1^2}-q_1} \frac{1}{\vert s_1-q_1 \vert^{1/2}}  \dd q_1\\
&\leq \sup_{- \frac{\sqrt{\mu_2}}{2}<p_1<-s_2}\frac{2 \vert p_1+s_2 \vert^{1/2}}{\sqrt{\mu_2}} \int_{-\infty}^{\sqrt{\mu_2}+p_1} \frac{1}{(\sqrt{s_1^2+s_2^2-p_1^2}-q_1)^{3/2} }\dd q_1\\
&=\sup_{- \frac{\sqrt{\mu_2}}{2}<p_1<-s_2}\frac{4 \vert p_1+s_2 \vert^{1/2}}{\sqrt{\mu_2}(\sqrt{s_1^2+s_2^2-p_1^2}-\sqrt{\mu_2}-p_1)^{1/2}} \\
&\leq \sup_{- \frac{\sqrt{\mu_2}}{2}<p_1<-s_2}\frac{4 (\sqrt{s_1^2+s_2^2-p_1^2}+\sqrt{\mu_2}+p_1)^{1/2}}{2^{1/2}\sqrt{\mu_2}(p_1+s_1)^{1/2}}
\leq \frac{2^{1/2} 4 s_1^{1/2}}{\sqrt{\mu_2}\mu_1^{1/4}}\leq \frac{2^{1/2} 4}{\mu_1^{1/2}} \nr \label{D8b}
\end{align*}
Similarly, for $p_1>s_2$ (which only occurs if $2\sqrt{\mu_2}<3\sqrt{\mu_1}$),
\begin{multline}
\sup_{s_2<p_1<\sqrt{\mu_1}-\frac{\sqrt{\mu_2}}{2}}\frac{2 h(p_1)}{\sqrt{\mu_2}} \int_{\sqrt{\mu_2}/2}^{\min\{\sqrt{\mu_2}+p_1,\sqrt{\mu_1}-p_1\}} \frac{1}{\sqrt{s_1^2+s_2^2-p_1^2}-q_1} \frac{1}{\vert s_1-q_1 \vert^{1/2}}  \dd q_1\\
\leq \sup_{s_2<p_1<\frac{\sqrt{\mu_2}}{2}}\frac{2 \vert p_1-s_2 \vert^{1/2}}{\sqrt{\mu_2}} \int_{-\infty}^{\sqrt{\mu_2}-p_1} \frac{1}{(\sqrt{s_1^2+s_2^2-p_1^2}-q_1)^{3/2} }\dd q_1 \leq \frac{2^{1/2} 4}{\mu_1^{1/2}},
\end{multline}
by \eqref{D8b}.
For $\vert p_1 \vert<s_2$,
\begin{multline}
\sup_{- s_2<p_1<s_2}\frac{2 h(p_1)}{\sqrt{\mu_2}} \int_{\sqrt{\mu_2}/2}^{\sqrt{\mu_1}-p_1} \frac{1}{\sqrt{s_1^2+s_2^2-p_1^2}-q_1} \frac{1}{\vert s_1-q_1 \vert^{1/2}}\dd q_1\\
\leq \sup_{- s_2<p_1<s_2}\frac{2 \vert \vert p_1\vert-s_2 \vert^{1/2}}{\sqrt{\mu_2}} \int_{-\infty}^{\sqrt{\mu_1}-p_1}   \frac{1}{\vert s_1-q_1 \vert^{3/2}}\dd q_1\\
= \sup_{- s_2<p_1<s_2}\frac{4 \vert \vert p_1\vert-s_2 \vert^{1/2}}{\sqrt{\mu_2}\vert s_2+p_1 \vert^{1/2}}=\frac{4}{\sqrt{\mu_2}}
\end{multline}
In total, we have
\begin{equation}\label{D8z}
\lVert D_{\mu_1,\mu_2}^{8} \rVert \leq \frac{2^{1/2} 4}{\mu_1^{1/2}}.
\end{equation}

\textbf{Region 9:}
By taking the supremum separately over the two components of region 9 and using the symmetry in $(p_1,q_1)\to -(p_1,q_1)$, we obtain
\begin{multline}\label{D9a}
\lVert D_{\mu_1,\mu_2}^{9} \rVert \leq \sup_{-2<p_1<2} h(p_1)\int_{-2}^2 \frac{1}{2}\frac{\chi_{9}}{p_1q_1+s_1s_2} \frac{1}{h(q_1)}  \dd q_1\\
\leq \sup_{-s_2<p_1<2} h(p_1) \int_{\max\{\sqrt{\mu_1}-p_1,\sqrt{\mu_2}/2,p_1-\sqrt{\mu_2}\}}^{\min\{p_1+\sqrt{\mu_2},2\}} \frac{1}{p_1q_1+s_1s_2} \frac{1}{\vert  q_1-s_1 \vert^{1/2}}  \dd q_1
\end{multline}
For $p_1>-s_2$ and $\max\{\sqrt{\mu_1}-p_1, \frac{\sqrt{\mu_2}}{2}\}<q_1<\sqrt{\mu_2}+p_1$ note that
\begin{multline}
p_1 q_1+s_1 s_2 \geq \left\{\begin{matrix} p_1(\sqrt{\mu_2}+p_1)+s_1 s_2 = (p_1+s_2)(p_1+s_1) & {\rm if }\ p_1\leq 0\\ p_1(\sqrt{\mu_1}-p_1)+s_1 s_2 = (p_1+s_2)(s_1-p_1) & {\rm if }\ \sqrt{\mu_1}-\frac{\sqrt{\mu_2}}{2}\geq p_1\geq 0\\
 p_1\frac{\sqrt{\mu_2}}{2}+s_1s_2 & {\rm if }\ p_1\geq \max\{\sqrt{\mu_1}-\frac{\sqrt{\mu_2}}{2},0\}\end{matrix} \right\} \\
\geq \frac{\sqrt{\mu_1}}{2}(p_1+s_2)
\end{multline}
Hence,
\begin{equation}\label{D9z}
\lVert D_{\mu_1,\mu_2}^{9} \rVert \leq \sup_{-s_2<p_1<2} \frac{2}{\sqrt{\mu_1}(p_1+s_2)^{1/2}} \int_{\sqrt{\mu_1}-p_1}^{p_1+\sqrt{\mu_2}} \frac{1}{\vert  q_1-s_1 \vert^{1/2}}  \dd q_1=\frac{8}{\sqrt{\mu_1}}
\end{equation}

\textbf{Region 10:}
By symmetry in $p_1$, we have
\begin{multline}\label{D10a}
\lVert D_{\mu_1,\mu_2}^{10} \rVert \leq \sup_{-2<p_1<2} h(p_1)\int_{-2}^2 \frac{\chi_{10}}{p_1^2+q_1^2-s_1^2-s_2^2} \frac{1}{h(q_1)}  \dd q_1\\
= \sup_{s_1<p_1<2} \vert p_1-s_1 \vert^{1/2} \int_{\max\{\sqrt{\mu_1}-p_1,-\frac{\sqrt{\mu_2}}{2}\}}^{\min\{p_1-\sqrt{\mu_2},\frac{\sqrt{\mu_2}}{2}\}} \frac{1}{p_1^2+q_1^2-s_1^2-s_2^2} \frac{1}{\vert \vert q_1\vert-s_2 \vert^{1/2}}  \dd q_1
\end{multline}
If we mirror the part of region 10 with $p_1>0,q_1<0$ along $q_1=0$, its image contains the part of region 10 with $p_1>0,q_1>0$.
Since the integrand is symmetric in $q_1$, we can thus bound
\begin{equation}
\lVert D_{\mu_1,\mu_2}^{10} \rVert \leq \sup_{s_1<p_1<2} 2\vert p_1-s_1 \vert^{1/2} \int_{\max\{\sqrt{\mu_2}-p_1,0\}}^{\min\{p_1-\sqrt{\mu_1},\frac{\sqrt{\mu_2}}{2}\}} \frac{1}{p_1^2+q_1^2-s_1^2-s_2^2} \frac{1}{\vert q_1-s_2 \vert^{1/2}}  \dd q_1
\end{equation}
Note that for $q_1\geq\sqrt{\mu_2}-p_1$, $p_1>s_1$ we have
\begin{multline}
p_1^2+q_1^2-s_1^2-s_2^2=(p_1-s_1)^2+(q_1-s_2)^2+2s_1 (p_1-s_1)+2s_2(q_1-s_2) \\
\geq 2s_1 (p_1-s_1)+2s_2(s_1-p_1) = 2\sqrt{\mu_1}(p_1-s_1).
\end{multline}
Therefore,
\begin{equation}\label{D10z}
\lVert D_{\mu_1,\mu_2}^{10} \rVert \leq \sup_{s_1<p_1<2} \frac{1}{\sqrt{\mu_1}\vert p_1-s_1 \vert^{1/2}} \int_{\sqrt{\mu_2}-p_1}^{p_1-\sqrt{\mu_1}} \frac{1}{\vert q_1-s_2 \vert^{1/2}}  \dd q_1=\frac{4}{\sqrt{\mu_1}}.
\end{equation}

\textbf{Region 11:}
By symmetry in $p_1$, we obtain
\begin{multline}\label{D11a}
\lVert D_{\mu_1,\mu_2}^{11} \rVert \leq \sup_{-2<p_1<2} h(p_1)\int_{-2}^2 \frac{1}{2}\frac{\chi_{11}}{-p_1q_1-s_1s_2} \frac{1}{h(q_1)}  \dd q_1\\
= \sup_{-\mu_1-\frac{\sqrt{\mu_2}}{2}<p_1<-\frac{\sqrt{\mu_2}}{2}}\frac{1}{2} \vert p_1+s_1 \vert^{1/2} \int_{\max\{-\sqrt{\mu_1}-p_1,\sqrt{\mu_2}+p_1\}}^{\frac{\sqrt{\mu_2}}{2}} \frac{1}{-p_1q_1-s_1s_2} \frac{1}{\vert  q_1-s_2 \vert^{1/2}}  \dd q_1
\end{multline}
For $p_1<-s_1$ we have $-p_1q_1-s_1s_2>s_1(q_1-s_2)$.
Hence,
\begin{multline}
\sup_{-\mu_1-\frac{\sqrt{\mu_2}}{2}<p_1<-s_1} \frac{1}{2}\vert p_1+s_1 \vert^{1/2} \int_{-\sqrt{\mu_1}-p_1}^{\frac{\sqrt{\mu_2}}{2}} \frac{1}{-p_1q_1-s_1s_2}\frac{1}{\vert  q_1-s_2 \vert^{1/2}}  \dd q_1\\
\leq \sup_{-\mu_1-\frac{\sqrt{\mu_2}}{2}<p_1<-s_1} \frac{\vert p_1+s_1 \vert^{1/2}}{2s_1} \int_{-\sqrt{\mu_1}-p_1}^{\infty} \frac{1}{\vert  q_1-s_2 \vert^{3/2}}  \dd q_1
=\frac{1}{s_1} \leq \frac{1}{\sqrt{\mu_1}}
\end{multline}
For $p_1>-s_1$, we carry out the integration
\begin{multline}
\sup_{-s_1<p_1<-\frac{\sqrt{\mu_2}}{2}} \frac{1}{2}\vert p_1+s_1 \vert^{1/2} \int_{\sqrt{\mu_2}+p_1}^{\frac{\sqrt{\mu_2}}{2}} \frac{1}{-p_1q_1-s_1s_2}\frac{1}{\vert  q_1-s_2 \vert^{1/2}}  \dd q_1\\
\leq \sup_{-s_1<p_1<-\frac{\sqrt{\mu_2}}{2}} \frac{1}{\vert p_1 \vert^{1/2} s_2^{1/2}} \artanh\left(\frac{ s_2^{1/2}}{\vert p_1 \vert^{1/2}}\right)=\frac{2^{1/2}}{\mu_2^{1/4}s_2^{1/2}} \artanh\left(\frac{2^{1/2} s_2^{1/2}}{\mu_2^{1/4}}\right)
\end{multline}
With $\artanh(x)\leq \frac{x}{1-x}$, we obtain
\begin{multline}\label{D11z}
\frac{2^{1/2}}{\mu_2^{1/4}s_2^{1/2}} \artanh\left(\frac{2^{1/2} s_2^{1/2}}{\mu_2^{1/4}}\right)\leq \frac{2^{1/2}}{\mu_2^{1/4}s_2^{1/2}} \frac{s_2^{1/2}}{\frac{\mu_2^{1/4}}{2^{1/2}}-s_2^{1/2}}
=\frac{2^{1/2}}{\mu_2^{1/4}} \frac{\frac{\mu_2^{1/4}}{2^{1/2}}+s_2^{1/2}}{\frac{\mu_1^{1/2}}{2}}\leq \frac{4}{\mu_1^{1/2}}
\end{multline}
Therefore, $\lVert D_{\mu_1,\mu_2}^{11} \rVert\leq \frac{4}{\mu_1^{1/2}}$.

\textbf{Region 12:}
By symmetry in $p_1$, we obtain
\begin{multline}\label{D12a}
\lVert D_{\mu_1,\mu_2}^{12} \rVert \leq \sup_{-2<p_1<2} h(p_1)\int_{-2}^2 \frac{1}{2}\frac{\chi_{12}}{p_1q_1+s_1s_2} \frac{1}{h(q_1)}  \dd q_1\\
= \sup_{-\sqrt{\mu_2}<p_1<-\sqrt{\mu_1}} \frac{1}{2} h(p_1) \int_{0}^{\min\{p_1+\sqrt{\mu_2},-\sqrt{\mu_1}-p_1\}} \frac{1}{p_1q_1+s_1s_2} \frac{1}{\vert  s_2-q_1 \vert^{1/2}}  \dd q_1
\end{multline}
For $p_1\geq -s_1$ note that $p_1 q_1+s_1 s_2 \geq  s_1(s_2-q_1) \geq  \frac{\sqrt{\mu_1}}{2}( s_2-q_2)$.
For $p_1\leq -s_1$ and $q_1<p_1+\sqrt{\mu_2}$ observe that
\begin{multline}
p_1 q_1+s_1 s_2 = (-p_1-s_1)(s_2-q_1)+s_1(s_2-q_1)+s_2(p_1+s_1) \\
\geq \frac{\sqrt{\mu_1}}{2}(s_2-q_1)+\frac{\sqrt{\mu_2}}{2}(s_2-q_1)+s_2(q_1-\sqrt{\mu_2}+s_1)=\frac{\sqrt{\mu_1}}{2}(s_2-q_1)+\frac{\sqrt{\mu_2}}{2}(s_2-q_1)-s_2(s_2-q_1)\\
\geq \frac{\sqrt{\mu_1}}{2}(s_2-q_1)
\end{multline}
Therefore,
\begin{equation}\label{D12z}
\lVert D_{\mu_1,\mu_2}^{12} \rVert \leq \sup_{-\sqrt{\mu_2}<p_1<-\sqrt{\mu_1}} \frac{\vert p_1+s_1 \vert^{1/2}}{\sqrt{\mu_1}} \int_{-\infty}^{\min\{p_1+\sqrt{\mu_2},-\sqrt{\mu_1}-p_1\}}  \frac{1}{\vert  s_2-q_1 \vert^{3/2}}  \dd q_1
=\frac{2}{\sqrt{\mu_1}}
\end{equation}

\textbf{Region 13:}
By symmetry under $(p_1,q_1)\to -(p_1,q_1)$, we obtain
\begin{multline}\label{D13a}
\lVert D_{\mu_1,\mu_2}^{13} \rVert \leq \sup_{-2<p_1<2} h(p_1)\int_{-2}^2 \frac{1}{2}\frac{\chi_{13}}{p_1q_1+s_1s_2} \frac{1}{h(q_1)}  \dd q_1\\
= \sup_{-\frac{\sqrt{\mu_2}}{2}+\sqrt{\mu_1}<p_1<\frac{3\sqrt{\mu_2}}{2}} h(p_1) \int_{\max\{\sqrt{\mu_1}-p_1,0, -\sqrt{\mu_2}+p_1\}}^{\frac{\sqrt{\mu_2}}{2}} \frac{1}{p_1q_1+s_1s_2} \frac{1}{\vert  s_2-q_1 \vert^{1/2}}  \dd q_1
\end{multline}
For $p_1>\sqrt{\mu_1}, q_1>0$, we have $p_1q_1+s_1 s_2 \geq \sqrt{\mu_1}(q_1+s_2)$.
Therefore,
\begin{multline}
\sup_{\sqrt{\mu_1}<p_1<\sqrt{\mu_2}+s_2} h(p_1) \int_{\max\{\sqrt{\mu_1}-p_1,0, -\sqrt{\mu_2}+p_1\}}^{\frac{\sqrt{\mu_2}}{2}} \frac{1}{p_1q_1+s_1s_2} \frac{1}{\vert  s_2-q_1 \vert^{1/2}}  \dd q_1\\
\leq \sup_{\sqrt{\mu_1}<p_1<\sqrt{\mu_2}+s_2} \frac{ \vert p_1-s_1 \vert^{1/2}}{\sqrt{\mu_1}} \int_{0}^{\infty} \frac{1}{q_1+s_2} \frac{1}{\vert  s_2-q_1 \vert^{1/2}}  \dd q_1\\
=\frac{2^{1/2}}{\sqrt{\mu_1}} \int_{0}^{\infty} \frac{1}{q_1+1} \frac{1}{\vert  1-q_1 \vert^{1/2}}  \dd q_1
\leq \frac{2^{1/2}}{\sqrt{\mu_1}}\left[  \int_{0}^{2}\frac{1}{\vert  1-q_1 \vert^{1/2}}  \dd q_1+\int_{2}^{\infty} \frac{1}{\vert  q_1-1 \vert^{3/2}}  \dd q_1 \right]
=\frac{2^{1/2}6}{\sqrt{\mu_1}}
\end{multline}
and
\begin{align*}
&\sup_{\sqrt{\mu_2}+s_2<p_1<\frac{3\sqrt{\mu_2}}{2}} h(p_1) \int_{\max\{\sqrt{\mu_1}-p_1,0, -\sqrt{\mu_2}+p_1\}}^{\frac{\sqrt{\mu_2}}{2}} \frac{1}{p_1q_1+s_1s_2} \frac{1}{\vert  s_2-q_1 \vert^{1/2}}  \dd q_1\\
&\leq \sup_{\sqrt{\mu_2}+s_2<p_1<\frac{3\sqrt{\mu_2}}{2}} \frac{ \vert p_1-s_1 \vert^{1/2}}{\sqrt{\mu_1}} \int_{-\sqrt{\mu_2}+p_1}^{\infty} \frac{1}{q_1+s_2}  \frac{1}{\vert  q_1-s_2 \vert^{1/2}}  \dd q_1\\
&=\sup_{2s_2<x<\mu_2-\frac{\sqrt{\mu_1}}{2}} \frac{ \vert x \vert^{1/2}}{\sqrt{\mu_1}} \int_{x}^{\infty} \frac{1}{y}  \frac{1}{\vert  y-2s_2 \vert^{1/2}}  \dd y
=\sup_{2s_2<x<\mu_2-\frac{\sqrt{\mu_1}}{2}} \frac{1}{\sqrt{\mu_1}} \int_{1}^{\infty} \frac{1}{y}  \frac{1}{\vert  y-\frac{2s_2}{x} \vert^{1/2}}  \dd y\\
&=\frac{1}{\sqrt{\mu_1}} \int_{1}^{\infty} \frac{1}{y}  \frac{1}{\vert  y-1 \vert^{1/2}}  \dd y\leq \frac{1}{\sqrt{\mu_1}}\left[\int_{1}^{2}\frac{1}{\vert  y-1 \vert^{1/2}}\dd y+ \int_{2}^{\infty} \frac{1}{\vert  y-1 \vert^{3/2}}\dd y\right]=\frac{4}{\sqrt{\mu_1}}, \nr
\end{align*}
where we substituted $x=p_1-s_1$ and $y=q_1+s_2$.
Next, we consider the case $p_1<\frac{\sqrt{\mu_1}}{2}$.
For $\frac{\sqrt{\mu_2}}{2}\geq q_1 \geq \sqrt{\mu_1}-p_1$ and $-s_2<p_1<\frac{\sqrt{\mu_1}}{2}$ we have
\begin{multline}
p_1 q_1+s_1 s_2
\geq\left\{\begin{matrix} \frac{\sqrt{\mu_1}}{2}(p_1+s_2)& {\rm if}\ p_1>0\\(s_1-q_1)(p_1+s_2)-p_1(s_1-q_1)+q_1(p_1+s_2)   & {\rm if}\ p_1<0
\end{matrix}\right.\\
\geq \left\{\begin{matrix}\frac{\sqrt{\mu_1}}{2}(p_1+s_2) & {\rm if}\ p_1>0\\ (s_1-q_1)(p_1+s_2) & {\rm if}\ p_1<0
\end{matrix}\right.
\geq \frac{\sqrt{\mu_1}}{2}(p_1+s_2)
\end{multline}
Therefore,
\begin{multline}\label{D13b}
\sup_{-\frac{\sqrt{\mu_2}}{2}+\sqrt{\mu_1}<p_1<\frac{\sqrt{\mu_1}}{2}} h(p_1) \int_{\max\{\sqrt{\mu_1}-p_1,0, -\sqrt{\mu_2}+p_1\}}^{\frac{\sqrt{\mu_2}}{2}} \frac{1}{p_1q_1+s_1s_2} \frac{1}{\vert  s_2-q_1 \vert^{1/2}}  \dd q_1\\
\leq \sup_{-\frac{\sqrt{\mu_2}}{2}+\sqrt{\mu_1}<p_1<\frac{\sqrt{\mu_1}}{2}} \frac{2 h(p_1)}{\sqrt{\mu_1}(p_1+s_2)} \int_{\sqrt{\mu_1}-p_1}^{\frac{\sqrt{\mu_2}}{2}} \frac{1}{\vert  s_2-q_1 \vert^{1/2}}  \dd q_1\\
\leq\sup_{-\frac{\sqrt{\mu_2}}{2}+\sqrt{\mu_1}<p_1<\frac{\sqrt{\mu_1}}{2}} \frac{4}{\sqrt{\mu_1}(p_1+s_2)^{1/2}} \left\{\begin{matrix}\left(\frac{\sqrt{\mu_1}}{2}\right)^{1/2} & {\rm if}\ \sqrt{\mu_1}-p_1>s_2\\
\left(\frac{\sqrt{\mu_1}}{2}\right)^{1/2}+ (s_2-\sqrt{\mu_1}+p_1)^{1/2}& {\rm if}\ \sqrt{\mu_1}-p_1<s_2\end{matrix} \right.
\end{multline}
Note that
$
\sup_{-\frac{\sqrt{\mu_2}}{2}+\sqrt{\mu_1}<p_1<\frac{\sqrt{\mu_1}}{2}} (p_1+s_2)^{-1/2} \left(\frac{\sqrt{\mu_1}}{2}\right)^{1/2} =1
$ 
and
that for $p_1>\sqrt{\mu_1}-s_2$ we have $\left \vert  \frac{s_2-\sqrt{\mu_1}+p_1}{p_1+s_2} \right \vert\leq 1$.
One can hence bound \eqref{D13b} above by $\frac{8}{\sqrt{\mu_1}}$.

For $q_1 \geq 0$ and $p_1>\frac{\sqrt{\mu_1}}{2}$ we have $p_1 q_1+s_1 s_2 \geq \frac{\sqrt{\mu_1}}{2}(q_1+s_2) $.
Therefore,
\begin{multline}
\sup_{\frac{\sqrt{\mu_1}}{2}<p_1<\sqrt{\mu_1}} h(p_1) \int_{\max\{\sqrt{\mu_1}-p_1,0, -\sqrt{\mu_2}+p_1\}}^{\frac{\sqrt{\mu_2}}{2}} \frac{1}{p_1q_1+s_1s_2} \frac{1}{\vert  s_2-q_1 \vert^{1/2}}  \dd q_1\\
\leq \sup_{\frac{\sqrt{\mu_1}}{2}<p_1<\sqrt{\mu_1}} \frac{2h(p_1)}{\sqrt{\mu_1}} \int_{\sqrt{\mu_1}-p_1}^{\infty} \frac{1}{(q_1+s_2)\vert  s_2-q_1 \vert^{1/2}}  \dd q_1\\
=\sup_{\frac{\sqrt{\mu_1}}{2}<p_1<\sqrt{\mu_1}} \frac{4h(p_1)}{\sqrt{\mu_1} \sqrt{2s_2}} \left\{\begin{matrix}  \artanh \left(\sqrt{\frac{s_2-\sqrt{\mu_1}+p_1}{2 s_2}}\right)+\pi&{\rm if}\ s_2>\sqrt{\mu_1}-p_1\\ \arctan \left(\sqrt{\frac{2 s_2}{\sqrt{\mu_1}-p_1-s_2}}\right) & {\rm if}\ s_2<\sqrt{\mu_1}-p_1\end{matrix}\right.
\end{multline}
We estimate the two cases separately:
\begin{multline}
\sup_{\sqrt{\mu_1}-s_2<p_1<\sqrt{\mu_1}} \frac{4h(p_1)}{\sqrt{\mu_1} \sqrt{2s_2}} \left[\artanh \left(\sqrt{\frac{s_2-\sqrt{\mu_1}+p_1}{2 s_2}}\right)+\pi\right]\\
\leq \frac{4 \vert s_1-\sqrt{\mu_1}+s_2 \vert^{1/2}}{\sqrt{\mu_1} \sqrt{2s_2}} \left[\artanh \left(\frac{1}{\sqrt{2}}\right)+\pi\right]=4\frac{\artanh \left(\frac{1}{\sqrt{2}}\right)+\pi}{\sqrt{\mu_1}}
\end{multline}
and
\begin{align*}
\sup_{\frac{\sqrt{\mu_1}}{2}<p_1<\sqrt{\mu_1}-s_2}& \frac{4h(p_1)}{\sqrt{\mu_1} \sqrt{2s_2}} \arctan \left(\sqrt{\frac{2 s_2}{\sqrt{\mu_1}-p_1-s_2}}\right) \\
&\leq \frac{4}{\sqrt{\mu_1}}\sup_{\frac{\sqrt{\mu_1}}{2}<p_1<\sqrt{\mu_1}-s_2}\Bigg[ \frac{\vert s_1-p_1 \vert^{1/2}-\vert \sqrt{\mu_1}-p_1-s_2 \vert^{1/2}}{ \sqrt{2s_2}}\frac{\pi}{2} \\
&\qquad + \frac{\vert \sqrt{\mu_1}-p_1-s_2 \vert^{1/2}}{\sqrt{2s_2}} \arctan \left(\sqrt{\frac{2 s_2}{\sqrt{\mu_1}-p_1-s_2}}\right)\Bigg]\\
&\leq \frac{4}{\sqrt{\mu_1}}\sup_{\frac{\sqrt{\mu_1}}{2}<p_1<\sqrt{\mu_1}-s_2} \frac{\sqrt{2s_2}}{\vert s_1-p_1 \vert^{1/2}+\vert \sqrt{\mu_1}-p_1-s_2 \vert^{1/2}}\frac{\pi}{2} +1 \leq \frac{4 (\frac{\pi}{2}+1)}{\sqrt{\mu_1}} \nr
\end{align*}

In total, we obtain
\begin{multline}\label{D13z}
\lVert D_{\mu_1,\mu_2}^{13} \rVert \leq \max\left\{\frac{2^{1/2} 6}{\mu_1^{1/2}},\frac{4}{\mu_1^{1/2}},\frac{8}{\mu_1^{1/2}},\frac{4(\artanh(1/\sqrt{2})+\pi)}{\mu_1^{1/2}},\frac{4 (\pi/2+1)}{\mu_1^{1/2}}\right\}\\
= \frac{4(\artanh(1/\sqrt{2})+\pi)}{\mu_1^{1/2}}
\end{multline}

\textbf{Region 14:}
By symmetry in $p_1$, we have
\begin{align*}
\lVert D_{\mu_1,\mu_2}^{14} \rVert &\leq \sup_{-2<p_1<2} h(p_1)\int_{-2}^2 \frac{\chi_{14}}{s_1^2+s_2^2-p_1^2-q_1^2} \frac{1}{h(q_1)}  \dd q_1\\
&= \sup_{0<p_1<s_1}h(p_1) \int_{\max\{-\sqrt{\mu_1}-p_1,-\sqrt{\mu_2}/2,-\sqrt{\mu_2}+p_1\}}^{\min\{\sqrt{\mu_1}-p_1,\frac{\sqrt{\mu_2}}{2}\}} \frac{1}{s_1^2+s_2^2-p_1^2-q_1^2} \frac{1}{\vert \vert q_1\vert-s_2 \vert^{1/2}}  \dd q_1\\
&\leq \sup_{0<p_1<s_1}2 h(p_1) \int_{\max\{0,p_1-\sqrt{\mu_1}\}}^{\min\{\sqrt{\mu_1}+p_1,\frac{\sqrt{\mu_2}}{2},\sqrt{\mu_2}-p_1\}} \frac{1}{s_1^2+s_2^2-p_1^2-q_1^2} \frac{1}{\vert \vert q_1\vert-s_2 \vert^{1/2}}  \dd q_1, \nr\label{D14a}
\end{align*}
where in the last inequality we increased the domain to be symmetric in $q_1$ and used the symmetry of the integrand.

For $p_1 \leq s_2$ and $\sqrt{\mu_1}+p_1>q_1$ we have $s_1^2+s_2^2-p_1^2-q_1^2 \geq s_1^2+s_2^2-p_1^2-(\sqrt{\mu_1}+p_1)^2=2 (s_2-p_1)(p_1+s_1)$.
Hence,
\begin{multline}
\sup_{0<p_1<\frac{\sqrt{\mu_2}}{2}-\sqrt{\mu_1}}2 h(p_1) \int_{0}^{\sqrt{\mu_1}+p_1} \frac{1}{s_1^2+s_2^2-p_1^2-q_1^2} \frac{1}{\vert \vert q_1\vert-s_2 \vert^{1/2}}  \dd q_1\\
\leq \sup_{0<p_1<\frac{\sqrt{\mu_2}}{2}-\sqrt{\mu_1}} \frac{1}{(s_2-p_1)^{1/2}(p_1+s_1)}  \int_{0}^{\sqrt{\mu_1}+p_1} \frac{1}{\vert \vert q_1\vert-s_2 \vert^{1/2}}  \dd q_1\\
=\sup_{0<p_1<\frac{\sqrt{\mu_2}}{2}-\sqrt{\mu_1}}\frac{2(s_2^{1/2}+(p_1+\sqrt{\mu_1}-s_2)^{1/2})}{(s_2-p_1)^{1/2}(p_1+s_1)}  \leq \frac{2(s_2^{1/2}+\left(\frac{\sqrt{\mu_1}}{2}\right)^{1/2})}{\left(\frac{\sqrt{\mu_1}}{2}\right)^{1/2}s_1}\leq \frac{4\left(\frac{\sqrt{\mu_2}}{2}\right)^{1/2}}{\left(\frac{\sqrt{\mu_1}}{2}\right)^{1/2}\frac{\sqrt{\mu_2}}{2}}\leq \frac{8}{\sqrt{\mu_1}}
\end{multline}
Similarly, for $p_1 \geq s_2$ and $\sqrt{\mu_2}-p_1>q_1$ we have $s_1^2+s_2^2-p_1^2-q_1^2 \geq s_1^2+s_2^2-p_1^2-(\sqrt{\mu_2}-p_1)^2=2 (s_1-p_1)(p_1-s_2)$.
Therefore,
\begin{multline}
\sup_{\frac{\sqrt{\mu_2}}{2}<p_1<s_1}2 h(p_1) \int_{p_1-\sqrt{\mu_1}}^{\sqrt{\mu_2}-p_1} \frac{1}{s_1^2+s_2^2-p_1^2-q_1^2} \frac{1}{\vert  q_1-s_2 \vert^{1/2}}  \dd q_1 \\
\leq \sup_{\frac{\sqrt{\mu_2}}{2}<p_1<s_1} \frac{1}{(s_1-p_1)^{1/2}(p_1-s_2)}  \int_{p_1-\sqrt{\mu_1}}^{\sqrt{\mu_2}-p_1} \frac{1}{\vert  q_1-s_2 \vert^{1/2}}  \dd q_1
=\sup_{\frac{\sqrt{\mu_2}}{2}<p_1<s_1} \frac{4}{p_1-s_2} =\frac{8}{\sqrt{\mu_1}}.
\end{multline}
For $\frac{\sqrt{\mu_2}}{2}-\sqrt{\mu_1}\leq p_1\leq \frac{\sqrt{\mu_2}}{2}$ and $q_1<\frac{\sqrt{\mu_2}}{2}$, we have $s_1^2+s_2^2-p_1^2-q_1^2\geq \frac{\mu_1}{2}$.
Thus,
\begin{multline}
\sup_{\frac{\sqrt{\mu_2}}{2}-\sqrt{\mu_1}< p_1<\frac{\sqrt{\mu_2}}{2}}2 h(p_1) \int_{\max\{0,p_1-\sqrt{\mu_1}\}}^{\frac{\sqrt{\mu_2}}{2}} \frac{1}{s_1^2+s_2^2-p_1^2-q_1^2} \frac{1}{\vert  q_1-s_2 \vert^{1/2}}  \dd q_1 \\
\leq \frac{4}{\mu_1}  \left(\frac{\sqrt{\mu_1}}{2}\right)^{1/2} \int_{\frac{\sqrt{\mu_2}}{2}-2\sqrt{\mu_1}}^{\frac{\sqrt{\mu_2}}{2}}\frac{1}{\vert  q_1-s_2 \vert^{1/2}}  \dd q_1
=\frac{8 \left(\left(\frac{3\sqrt{\mu_1}}{2}\right)^{1/2}+\left(\frac{\sqrt{\mu_1}}{2}\right)^{1/2}\right)}{2^{1/2}\mu_1^{3/4}}
\leq  \frac{4}{\mu_1^{1/2}}\left(\sqrt{3}+1\right)
\end{multline}
In total, we have
\begin{equation}\label{D14z}
\lVert D_{\mu_1,\mu_2}^{14} \rVert \leq \frac{4}{\mu_1^{1/2}}\left(\sqrt{3}+1\right)
\end{equation}
\end{proof}

\subsection{Proof of Lemma~\ref{mu1mu2_bound}}\label{pf:mu1mu2_bound}
\begin{proof}[Proof of Lemma~\ref{mu1mu2_bound}]
The integral in \eqref{mu1mu2_bound_eq} is invariant under rotations of $\tilde q$.
Therefore, it suffices to take the supremum over $\tilde q=q_2\geq 0$ for $d=2$ and $\tilde q =(q_2,0)$ with $q_2\geq 0$ for $d=3$.
Furthermore, it suffices to restrict to $p_2\geq 0$ since the integrand is invariant under $\tilde p \to -\tilde p$.
Note that under these conditions $\mu_1 \leq \mu_2$.
We split the domain of integration in \eqref{mu1mu2_bound_eq} into two regions according to $\mu_1=\min\{\mu_1,\mu_2\} \lessgtr 0$.

{\bf Dimension three:}
We first consider the case $\mu_1<0$, i.e.~$\vert p_2+q_2\vert^2>1-p_3^2$.
In this case,
\begin{multline}\label{pf:mu1mu2_bound.1}
\sup_{\tilde q=(q_2,0), q_2\geq 0}\int_{\BR^{2}} \chi_{\vert \tilde p \vert<2 } \chi_{p_2\geq 0}\frac{\chi_{\min\{\mu_1,\mu_2\}<0}}{(-\min\{\mu_1,\mu_2\})^{\alpha}}\dd \tilde p\\
=\sup_{q_2\geq 0}\Bigg[\int_{-1}^1\dd p_3  \int_{\max\{\sqrt{1-p_3^2}-q_2,0\}}^{\sqrt{4-p_3^2}} \frac{1}{((p_2+q_2)^2+p_3^2-1)^{\alpha}}\dd p_2 \\
+\int_{1<|p_3|<2}\dd p_3 \int_{0}^{\sqrt{4-p_3^2}} \frac{1}{((p_2+q_2)^2+p_3^2-1)^\alpha}\dd p_2\Bigg]
\end{multline}
Let $q_2$ and $|p_3|\leq 1$ be fixed.
By substituting $x=p_2+q_2-\sqrt{1-p_3^2}$ if $q_2\leq\sqrt{1-p_3^2}$ one obtains
\begin{multline}\label{pf:mu1mu2_bound.2}
\int_{\max\{\sqrt{1-p_3^2}-q_2,0\}}^{2} \frac{1}{((p_2+q_2)^2+p_3^2-1)^{\alpha}}\dd p_2\\
\leq\int_{0}^{2} \frac{ \chi_{q_2\leq\sqrt{1-p_3^2}}}{(x+\sqrt{1-p_3^2})^2+p_3^2-1)^{\alpha}}\dd x
+\int_{0}^{2} \frac{\chi_{q_2>\sqrt{1-p_3^2}}}{((p_2+\sqrt{1-p_3^2})^2+p_3^2-1)^{\alpha}}\dd p_2\\
\leq
\int_{0}^{2} \frac{1}{(2p_2\sqrt{1-p_3^2})^{\alpha}}\dd p_2 \leq \frac{C}{(1-p_3^2)^{\alpha/2}}
\end{multline}
for some finite constant $C$.
Since $\int_{-1}^1(1-p_3^2)^{-\alpha/2}\dd p_3<\infty$, the first term in \eqref{pf:mu1mu2_bound.1} is bounded.
The second term is bounded by
\begin{equation}
\int_{1<|p_3|<2}\dd p_3 \int_{0}^{2} \frac{1}{(p_3^2-1)^\alpha}\dd p_2<\infty.
\end{equation}

For the case $\mu_1>0$ we have $\vert p_2+q_2\vert^2<1-p_3^2$.
Hence,
\begin{multline}\label{pf:mu1mu2_bound.3}
\sup_{{\tilde q=(q_2,0), q_2\geq 0}} \int_{\BR^{2}} \chi_{\vert \tilde p \vert<2 }\chi_{p_2\geq 0} \frac{\chi_{0<\min\{\mu_1,\mu_2\}}}{\min\{\mu_1,\mu_2\}^{\alpha}}\dd \tilde p\\
=\sup_{q_2\geq 0} \int_{-1}^1 \dd p_3 \chi_{q_2\leq \sqrt{1-p_3^2}} \int_0^{\sqrt{1-p_3^2}-q_2} \frac{1}{(1-(p_2+q_2)^2-p_3^2)^\alpha} \dd p_2
\end{multline}
For fixed $|p_3|<1$ and $q_2\leq \sqrt{1-p_3^2}$ substituting $x=\sqrt{1-p_3^2}-q_2-p_2$ gives
\begin{multline}\label{pf:mu1mu2_bound.4}
\int_0^{\sqrt{1-p_3^2}-q_2} \frac{1}{(1-(p_2+q_2)^2-p_3^2)^\alpha} \dd p_2
=\int_0^{\sqrt{1-p_3^2}-q_2} \frac{1}{(1-(\sqrt{1-p_3^2}-x)^2-p_3^2)^\alpha} \dd x\\
=\int_0^{\sqrt{1-p_3^2}-q_2} \frac{1}{x^\alpha (2\sqrt{1-p_3^2}-x)^\alpha} \dd x.
\end{multline}
Thus the expression in \eqref{pf:mu1mu2_bound.3} is bounded by
\begin{multline}
\sup_{q_2\geq 0} \int_{-1}^1 \dd p_3 \chi_{q_2\leq \sqrt{1-p_3^2}}  \int_0^{\sqrt{1-p_3^2}-q_2} \frac{1}{x^\alpha (\sqrt{1-p_3^2}+q_2)^\alpha} \dd x\\
\leq \int_{-1}^1 \dd p_3  \int_0^{1} \frac{1}{x^\alpha (\sqrt{1-p_3^2})^\alpha}\dd x<\infty
\end{multline}

{\bf Dimension two:}
For the case $\mu_1<0$ we have $\vert p_2+q_2\vert>1$.
Hence,
\begin{equation}
\sup_{q_2\geq 0} \int_{0}^2\frac{\chi_{\min\{\mu_1,\mu_2\}<0}}{(-\min\{\mu_1,\mu_2\})^{\alpha}}\dd p_2=\sup_{q_2\geq 0}\int_{\max\{1-q_2,0\}}^{2} \frac{1}{((p_2+q_2)^2-1)^\alpha} \dd p_2.
\end{equation}
This is finite according to \eqref{pf:mu1mu2_bound.2}.

For the case $\mu_1>0$,
\begin{equation}\label{dj_bound_2}
\sup_{q_2\geq 0}\int_{0}^2\frac{\chi_{0<\min\{\mu_1,\mu_2\}}}{\min\{\mu_1,\mu_2\}^{\alpha}}\dd p_2=\sup_{0 \leq q_2\leq 1}\int_{0}^{1-q_2}\frac{1}{(1-(p_2+q_2)^2)^\alpha}\dd p_2
=\int_{0}^{1}\frac{1}{x^\alpha(2-x)^\alpha}\dd x <\infty,
\end{equation}
where we used  \eqref{pf:mu1mu2_bound.4} in the second equality.
\end{proof}

\subsection{Proof of Lemma~\ref{lea:IGIto0}}\label{sec:pf_lea:IGIto0}
\begin{proof}[Proof of Lemma~\ref{lea:IGIto0}]
The proof follows from elementary computations.
We carry out the case $d=2$ and leave the case $d=3$, where  one additional integration over $q_3$ needs to be performed, to the reader.

By symmetry, we may restrict to $p_1,q_1,p_2\geq 0$.
Furthermore, we will partition the remaining domain of $p_2,q_2$ into nine subdomains.
Let $\chi_j$ be the characteristic function of domain $j$.
Since $(a+b)^2\leq 2(a^2+b^2)$, there is a constant $C$ such that the expression in \eqref{pf_2d_w_lea4.2} is bounded above by $C \sum_{j=1}^9 \lim_{\eps\to0} I_j$, where
\begin{equation}
I_j= \sup_{0\leq p_2<\eps}\int_{\BR^2} \chi_{0<p_1,q_1<\eps} \left[\int_{-\sqrt{2}}^{\sqrt{2}} \frac{2\chi_j}{\vert (p+q)^2-1 \vert+ \vert (p-q)^2-1 \vert} \dd q_2 \right]^2 \dd p_1 \dd q_1.
\end{equation}
Hence, we can consider the domains case by case and prove that $\lim_{\eps\to0}  I_j=0$ for each of them.

\begin{figure}
\centering
\begin{tikzpicture}[scale=0.35]
\tikzmath{\smu1=3.5;\smu2=8.5;\l=15;\t=0.3;}

\draw (0,\l) node[left]{$q_2$};
\draw (\l,0) node[below]{$p_2$};

\draw (\smu2/2-\smu1/2,1.75*\smu2-\smu1/4) node{1};
\draw (\smu2/4-\smu1/4,-\smu2-\smu1) node{2};
\draw (\smu2/2-\smu1/2+0.5*\smu1,-\smu2-\smu1) node{3};
\draw (1.5*\smu2,-0.5*\smu2) node{4};
\draw (\smu2,\smu2) node{5};
\draw (\smu2/4-\smu1/4,-0.5*\smu2-0.5*\smu1) node{6};
\draw (\smu2/2+\smu1/2,-0.5*\smu2-0.5*\smu1) node{7};
\draw (\smu2/2-\smu1/2-0.5*\smu1,-0.5*\smu2+0.5*\smu1) node{8};
\draw (\smu2/2-\smu1/2+0.5*\smu1,-0.5*\smu2+0.5*\smu1) node{9};

 \draw[black, dashed, decoration={markings, mark=at position 1 with {\arrow[scale=2,>=stealth]{>}}},
        postaction={decorate}]
 (0,0)--(\l,0);
 \draw[black, dashed, decoration={markings, mark=at position 1 with {\arrow[scale=2,>=stealth]{>}}},
        postaction={decorate}]
 (0,-\l)--(0,\l);

 \draw[black, thick]
(\smu2/2-\smu1/2,3*\smu1/2-\smu2/2)--(\smu2/2-\smu1/2,-\l)
(0,\smu1)--(\l,\smu1-\l)
(0,-\smu1)--(-\smu1+\l,-\l)
(0,-\smu2)--(\l,\l-\smu2)
(0,\smu2)--(-\smu2+\l,\l);

\draw (\l,\smu1-\l) node[right]{$q_2=\sqrt{\mu_1}-p_2$};
\draw (\l-\smu1,-\l) node[right]{$q_2=-\sqrt{\mu_1}-p_2$};
\draw (-\smu2+\l,\l) node[right]{$q_2=\sqrt{\mu_2}+p_2$};
\draw (\l,\l-\smu2) node[right]{$q_2=-\sqrt{\mu_2}+p_2$};

\end{tikzpicture}
\caption{Domains occurring in the proof of Lemma~\ref{lea:IGIto0}.}
\label{IGIto0_domains}
\end{figure}
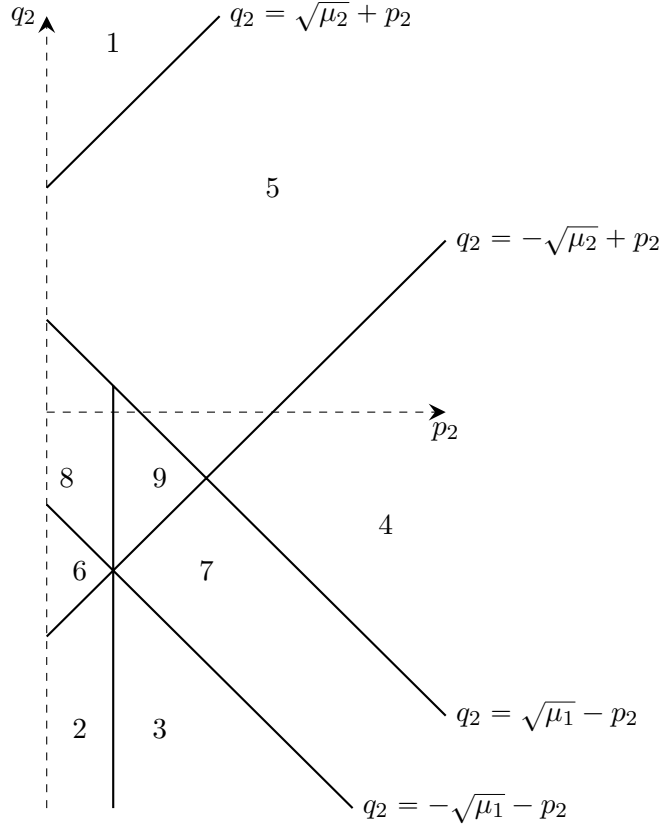

We use the notation $\mu_1=1-(p_1+q_1)^2$ and $\mu_2=1-(p_1-q_1)^2$. (Note that this differs from the notation in Lemma~\ref{Dmu1mu2_bound}).
Since $p_1,q_1\geq 0$ we have $\mu_2\geq \mu_1$.
We assume that $\eps<1/4$, and thus for $p_1,q_1<\eps$ we have $\mu_1,\mu_2>1-4\eps^2>3/4$.

For fixed $0<p_1,q_1<\eps$, we choose the subdomains for $p_2,q_2$ as sketched in Figure~\ref{IGIto0_domains}.
The subdomains are chosen according to the signs of $(p_2+q_2)^2-\mu_1$ and $(p_2-q_2)^2-\mu_2$, and to distinguish which of $-\sqrt{\mu_1}-p_2,-\sqrt{\mu_2}+p_2$ is larger.

We start with domains 1 to 4, where $(p+q)^2-1=(p_2+q_2)^2-\mu_1>0$ and $(p-q)^2-1=(p_2-q_2)^2-\mu_2>0$.
Note that in domain 4, $p_2\geq \frac{\sqrt{\mu_1}+\sqrt{\mu_2}}{2}\geq \sqrt{1-4\eps^2}$, which is larger than $\eps$.
Hence $\chi_4=0$ for $p_2<\sqrt{1-4\eps^2}$, giving $I_4=0$.
For domains 2 and 3, we have
\begin{align}
I_2&=\sup_{0\leq p_2<\eps}\int_{\BR^2} \chi_{0<p_1,q_1<\eps}\chi_{2p_2<\sqrt{\mu_2}-\sqrt{\mu_1}} \left[\int_{a_2}^{2} \frac{1}{p_1^2+q_1^2+p_2^2+q_2^2-1} \dd q_2 \right]^2 \dd p_1 \dd q_1,\\
I_3&=\sup_{0\leq p_2<\eps}\int_{\BR^2} \chi_{0<p_1,q_1<\eps}\chi_{2p_2>\sqrt{\mu_2}-\sqrt{\mu_1}} \left[\int_{a_3}^{2} \frac{1}{p_1^2+q_1^2+p_2^2+q_2^2-1} \dd q_2 \right]^2 \dd p_1 \dd q_1,
\end{align}
where $a_2=\sqrt{\mu_2}-p_2$ and $a_3=\sqrt{\mu_1}+p_2$.
Since $0\leq p_1,q_1,p_2<\eps$ we have $1-p_1^2-q_1^2-p_2^2>3/4$ and thus
\begin{multline}
\int_{a_j}^{2} \frac{1}{p_1^2+q_1^2+p_2^2+q_2^2-1} = \frac{\artanh \frac{\sqrt{1-p_1^2-q_1^2-p_2^2}}{a_j}-\artanh \frac{\sqrt{1-p_1^2-q_1^2-p_2^2}}{2}}{\sqrt{1-p_1^2-q_1^2-p_2^2}}\\
\leq C \artanh \frac{\sqrt{1-p_1^2-q_1^2-p_2^2}}{a_j}.
\end{multline}
Since $\artanh(x/y)=\ln((y+x)^2/(y^2-x^2))/2$ and $\sqrt{1-p_1^2-q_1^2-p_2^2}+a_j\leq 3$ we get
\begin{equation}\label{pf_2d_w_lea4.3}
I_2\leq \frac{C}{2} \sup_{0\leq p_2<\eps}\int_{\BR^2} \chi_{0<p_1,q_1<\eps}\chi_{2p_2<\sqrt{\mu_2}-\sqrt{\mu_1}} \left[ \ln \frac{9}{2(p_1 q_1-p_2(\sqrt{\mu_2}-p_2))} \right]^2 \dd p_1 \dd q_1
\end{equation}
and
\begin{equation}\label{pf_2d_w_lea4.4}
I_3\leq  \frac{C}{2} \sup_{0\leq p_2<\eps}\int_{\BR^2} \chi_{0<p_1,q_1<\eps}\chi_{2p_2>\sqrt{\mu_2}-\sqrt{\mu_1}}  \left[ \ln \frac{9}{2(p_1 q_1+p_2(\sqrt{\mu_1}+p_2))} \right]^2 \dd p_1 \dd q_1.
\end{equation}
For domain 2, we substitute $z=p_1+q_1$ and $r=p_1-q_1$ and obtain the bound
\begin{equation}
I_2\leq \frac{C}{4}  \sup_{0\leq p_2<\eps}\int_{\BR^2} \chi_{\vert r\vert <z<2\eps}\chi_{2p_2<\sqrt{\mu_2}-\sqrt{\mu_1}} \left[ \ln \frac{18}{z^2-r^2-4p_2(\sqrt{1-r^2}-p_2)} \right]^2 \dd r \dd z
\end{equation}
The condition $2p_2<\sqrt{\mu_2}-\sqrt{\mu_1}$ implies that $x:=z^2-r^2-4p_2(\sqrt{1-r^2}-p_2)\geq 0$.
Substituting $z$ by $x$ gives
\begin{multline}
I_2\leq \frac{C}{4} \sup_{0\leq p_2<\eps}\int_{-2\eps}^{2\eps}\dd r \int_{0}^{\eps^2} \dd x \left[ \ln \frac{18}{x} \right]^2 \frac{1}{2\sqrt{x+r^2+4p_2(\sqrt{1-r^2}-p_2)}}\\
\leq \frac{C}{2}\eps \int_{0}^{\eps^2} \left[ \ln \frac{18}{x} \right]^2 \frac{1}{\sqrt{x}}\dd x
\end{multline}
This vanishes as $\eps\to 0$.
For domain 3 we bound \eqref{pf_2d_w_lea4.4} by
\begin{equation}
I_3\leq C \int_{\BR^2} \chi_{0<p_1,q_1<\eps}  \left[ \ln \frac{9}{2 p_1 q_1} \right]^2 \dd p_1 \dd q_1,
\end{equation}
which vanishes in the limit $\eps \to 0$.
For domain 1 note that since $\sqrt{\mu_2}+p_2\geq a_2,a_3$, we have $I_1\leq I_2+I_3$.

Now consider domain 5, where $(p+q)^2-1=(p_2+q_2)^2-\mu_1>0$ and $(p-q)^2-1=(p_2-q_2)^2-\mu_2<0$.
We have
\begin{equation}
I_5= \sup_{0\leq p_2<\eps}\int_{\BR^2} \chi_{0<p_1,q_1<\eps} \left[\int_{\sqrt{\mu_1}-p_2}^{\sqrt{\mu_2}+p_2} \frac{1}{2(p_1q_1+p_2 q_2)} \dd q_2 \right]^2 \dd p_1 \dd q_1.
\end{equation}
Integration over $q_2$ gives
\begin{equation}
\int_{\sqrt{\mu_1}-p_2}^{\sqrt{\mu_2}+p_2} \frac{1}{2(p_1q_1+p_2 q_2)} \dd q_2 =\frac{1}{2p_2}\ln\left(1+ \frac{(\sqrt{\mu_2}-\sqrt{\mu_1})p_2+2p_2^2}{p_1 q_1+(\sqrt{\mu_1}-p_2)p_2}\right).
\end{equation}
Note that $\sqrt{\mu_2}-\sqrt{\mu_1}= 4p_1q_1/(\sqrt{\mu_2}+\sqrt{\mu_1})\leq 2 p_1 q_1 /\sqrt{1-4\eps^2}$ and $\sqrt{\mu_1}-p_2\geq \sqrt{1-4\eps^2}-\eps$.
We can therefore bound the previous expression from above by
\begin{equation}
\frac{1}{2p_2}\ln\left(1+ \frac{2 p_2}{\sqrt{1-4\eps^2}}+\frac{2p_2}{\sqrt{1-4\eps^2}-\eps}\right) \leq \frac{1 }{\sqrt{1-4\eps^2}}+\frac{1}{\sqrt{1-4\eps^2}-\eps}<C,
\end{equation}
where we used that $\ln(1+ x)/x \leq 1$ for $x\geq 0$.
Therefore $I_5\leq C^2 \eps^2$ vanishes as $\eps\to0$.

For region 6 we have
\begin{equation}
I_6= \sup_{0\leq p_2<\eps}\int_{\BR^2} \chi_{0<p_1,q_1<\eps}\chi_{2p_2\leq \sqrt{\mu_2}-\sqrt{\mu_1}}\left[\int_{-\sqrt{\mu_2}+p_2}^{-\sqrt{\mu_1}-p_2} \frac{1}{2(p_1 q_1+p_2 q_2)} \dd q_2 \right]^2 \dd p_1 \dd q_1.
\end{equation}
Integration over $q_2$ gives
\begin{equation}\label{pf_2d_w_lea4.5}
\int_{-\sqrt{\mu_2}+p_2}^{-\sqrt{\mu_1}-p_2} \frac{1}{2(p_1q_1+p_2 q_2)} \dd q_2 =\frac{1}{2p_2}\ln\left(1+ p_2\frac{(\sqrt{\mu_2}-\sqrt{\mu_1}-2p_2)}{p_1 q_1-(\sqrt{\mu_2}-p_2)p_2}\right).
\end{equation}
One can compute that
\begin{equation}\label{pf_2d_w_lea4.6}
\frac{\p}{\p p_2}\frac{\sqrt{\mu_2}-\sqrt{\mu_1}-2p_2}{p_1 q_1-(\sqrt{\mu_2}-p_2)p_2}=\frac{8}{(\sqrt{\mu_2}+\sqrt{\mu_1}-2p_2)^2}>0.
\end{equation}
Thus, for $\chi_{2p_2\leq \sqrt{\mu_2}-\sqrt{\mu_1}}$ we have
\begin{equation}
\frac{\sqrt{\mu_2}-\sqrt{\mu_1}-2p_2}{p_1 q_1-(\sqrt{\mu_2}-p_2)p_2}\leq \lim_{p_2\to (\sqrt{\mu_2}-\sqrt{\mu_1})/2}\frac{\sqrt{\mu_2}-\sqrt{\mu_1}-2p_2}{p_1 q_1-(\sqrt{\mu_2}-p_2)p_2}=\frac{2}{\sqrt{\mu_1}}.
\end{equation}
The expression in \eqref{pf_2d_w_lea4.5} is thus bounded above by
\begin{equation}
\frac{1}{2p_2}\ln\left(1+ p_2 \frac{2}{\sqrt{\mu_1}}\right)\leq \frac{1}{\sqrt{\mu_1}}\leq \frac{1}{\sqrt{1-4\eps^2}},
\end{equation}
which is bounded as $\eps\to0$.
In total, we have $I_6 \leq C \eps^2$, which vanishes in the limit $\eps\to0$.

For region 7,
\begin{equation}
I_7= \sup_{0\leq p_2<\eps}\int_{\BR^2} \chi_{0<p_1,q_1<\eps}\chi_{2p_2\geq \sqrt{\mu_2}-\sqrt{\mu_1}}\left[\int_{-\sqrt{\mu_1}-p_2}^{-\sqrt{\mu_2}+p_2} \frac{1}{-2(p_1 q_1+p_2 q_2)} \dd q_2 \right]^2 \dd p_1 \dd q_1.
\end{equation}
Integration over $q_2$ gives
\begin{equation}
\int_{-\sqrt{\mu_1}-p_2}^{-\sqrt{\mu_2}+p_2} \frac{1}{-2(p_1 q_1+p_2 q_2)} \dd q_2 =\frac{1}{2p_2}\ln\left(1+ p_2\frac{(\sqrt{\mu_2}-\sqrt{\mu_1}-2p_2)}{p_1 q_1-(\sqrt{\mu_2}-p_2)p_2}\right).
\end{equation}
According to \eqref{pf_2d_w_lea4.6}, for $(\sqrt{\mu_2}-\sqrt{\mu_1})/2\leq  p_2<\eps$ this is bounded by
\begin{equation}
\frac{1}{2p_2}\ln\left(1+ p_2\frac{(\sqrt{\mu_2}-\sqrt{\mu_1}-2\eps)}{p_1 q_1-(\sqrt{\mu_2}-\eps)\eps}\right)\leq \frac{1}{2} \frac{2\eps-(\sqrt{\mu_2}-\sqrt{\mu_1})}{(\sqrt{\mu_2}-\eps)\eps-p_1 q_1}.
\end{equation}
For $p_1,q_1<\eps$ this can be further estimated by
\begin{equation}
\frac{1}{2} \frac{2\eps}{(\sqrt{\mu_2}-\eps)\eps-\eps^2} \leq \frac{1}{\sqrt{1-4\eps^2}-2\eps},
\end{equation}
which is bounded for $\eps \to 0$.
Hence, $I_7 \leq C \eps^2$ vanishes for $\eps\to0$.

For domains 8 and 9, we have
\begin{align}
I_8&=\sup_{0\leq p_2<\eps}\int_{\BR^2} \chi_{0<p_1,q_1<\eps}\chi_{2p_2<\sqrt{\mu_2}-\sqrt{\mu_1}} \left[\int_{-\sqrt{\mu_1}-p_2}^{\sqrt{\mu_1}-p_2} \frac{1}{1-p_1^2-q_1^2-p_2^2-q_2^2} \dd q_2 \right]^2 \dd p_1 \dd q_1,\\
I_9&=\sup_{0\leq p_2<\eps}\int_{\BR^2} \chi_{0<p_1,q_1<\eps}\chi_{2p_2>\sqrt{\mu_2}-\sqrt{\mu_1}} \left[\int_{-\sqrt{\mu_2}+p_2}^{\sqrt{\mu_1}-p_2} \frac{1}{1-p_1^2-q_1^2-p_2^2-q_2^2} \dd q_2 \right]^2 \dd p_1 \dd q_1.
\end{align}
We bound
\begin{multline}\label{pf_2d_w_lea4.7}
\int_{-\sqrt{\mu_1}-p_2}^{\sqrt{\mu_1}-p_2} \frac{1}{1-p_1^2-q_1^2-p_2^2-q_2^2} \dd q_2\leq 2\int_{0}^{\sqrt{\mu_1}+p_2} \frac{1}{1-p_1^2-q_1^2-p_2^2-q_2^2} \dd q_2\\
=\frac{1}{\sqrt{1-p_1^2-q_1^2-p_2^2}}\ln \left(  \frac{\sqrt{1-p_1^2-q_1^2-p_2^2}+\sqrt{\mu_1}+p_2}{\sqrt{1-p_1^2-q_1^2-p_2^2}-\sqrt{\mu_1}-p_2} \right) \\
=\frac{1}{\sqrt{1-p_1^2-q_1^2-p_2^2}}\ln \left(  \frac{(\sqrt{1-p_1^2-q_1^2-p_2^2}+\sqrt{\mu_1}+p_2)^2}{2(p_1 q_1-p_2(\sqrt{\mu_1}+p_2))} \right) \\
\leq \frac{1}{\sqrt{1-3\eps^2}}\ln \left(  \frac{9}{2(p_1 q_1-p_2(\sqrt{\mu_1}+p_2))} \right)
\end{multline}
Substituting $z=p_1+q_1$ and $r=p_1-q_1$ we obtain
\begin{multline}
I_8\leq C \sup_{0\leq p_2<\eps}\int_{\BR^2} \chi_{0<p_1,q_1<\eps}\chi_{2p_2<\sqrt{\mu_2}-\sqrt{\mu_1}} \ln \left(  \frac{9}{2(p_1 q_1-p_2(\sqrt{\mu_1}+p_2))} \right)^2 \dd p_1 \dd q_1\\
\leq \frac{C}{2}\sup_{0\leq p_2<\eps}\int_{0}^{2\eps} \dd z \int_{-\eps}^\eps \dd r  \chi_{|r|<z}\chi_{2p_2<\sqrt{1-r^2}-\sqrt{1-z^2}} \ln \left(  \frac{18}{z^2-r^2-4p_2(\sqrt{1-z^2}+p_2)} \right)^2.
\end{multline}
Substituting $r$ by $x=z^2-r^2-4p_2(\sqrt{1-z^2}+p_2)$ and using H\"older's inequality we obtain
\begin{multline}
I_8\leq \frac{C}{4}\sup_{0\leq p_2<\eps}\int_{4p_2-4p_2^2}^{2\eps} \dd z \int_{0}^{z^2-4p_2(p_2+\sqrt{1-z^2})}   \ln \left(  \frac{18}{x} \right)^2 \frac{1}{\sqrt{z^2-4p_2(\sqrt{1-z^2}+p_2)-x}}\dd x\\
\leq  \frac{C}{4}\sup_{0\leq p_2<\eps}\int_{4p_2-4p_2^2}^{2\eps} \dd z \Bigg[\left (\int_{0}^{z^2-4p_2(p_2+\sqrt{1-z^2})}   \ln \left(  \frac{18}{x} \right)^6 \dd x \right)^{1/3}\times \\
\left (\int_{0}^{z^2-4p_2(p_2+\sqrt{1-z^2})} \frac{1}{(z^2-4p_2(\sqrt{1-z^2}+p_2)-x)^{3/4}}\dd x\right)^{2/3}\Bigg].
\end{multline}
In the last line we substitute $y=z^2-4p_2(\sqrt{1-z^2}+p_2)-x$, and then we use $z^2-4p_2(\sqrt{1-z^2}+p_2)-x\leq 4\eps^2$ to arrive at the bound
\begin{multline}
I_8\leq  \frac{C}{4}\sup_{0\leq p_2<\eps}\int_{4p_2-4p_2^2}^{2\eps} \dd z \Bigg[\left (\int_{0}^{4\eps^2}   \ln \left(  \frac{18}{x} \right)^6 \dd x \right)^{1/3}\left (\int_{0}^{4\eps^2} \frac{1}{y^{3/4}}\dd y\right)^{2/3}\Bigg] \\
\leq \frac{C}{2} \eps \left (\int_{0}^{4\eps^2}   \ln \left(  \frac{18}{x} \right)^6 \dd x \right)^{1/3}\left (\int_{0}^{4\eps^2} \frac{1}{y^{3/4}}\dd y\right)^{2/3} ,
\end{multline}
which vanishes as $\eps \to 0$.
For $I_9$ we bound (analogously to \eqref{pf_2d_w_lea4.7})
\begin{multline}
\int_{-\sqrt{\mu_2}+p_2}^{\sqrt{\mu_1}-p_2} \frac{1}{1-p_1^2-q_1^2-p_2^2-q_2^2} \dd q_2 \leq 2 \int_{0}^{\sqrt{\mu_2}-p_2} \frac{1}{1-p_1^2-q_1^2-p_2^2-q_2^2}\dd q_2 \\
=\frac{1}{\sqrt{1-p_1^2-q_1^2-p_2^2}}\ln \left(  \frac{(\sqrt{1-p_1^2-q_1^2-p_2^2}+\sqrt{\mu_2}-p_2)^2}{2(p_2(\sqrt{\mu_2}-p_2)-p_1 q_1)} \right) \\
\leq \frac{1}{\sqrt{1-3\eps^2}}\ln \left(  \frac{4}{2(p_2(\sqrt{\mu_2}-p_2)-p_1 q_1)} \right)
\end{multline}
Substituting $z=p_1+q_1$ and $r=p_1-q_1$ we obtain
\begin{multline}
I_9\leq C \sup_{0\leq p_2<\eps}\int_{\BR^2} \chi_{0<p_1,q_1<\eps}\chi_{2p_2>\sqrt{\mu_2}-\sqrt{\mu_1}} \ln \left(  \frac{4}{2(p_2(\sqrt{\mu_2}-p_2)-p_1 q_1)} \right)^2 \dd p_1 \dd q_1\\
\leq \frac{C}{2}\sup_{0\leq p_2<\eps}\int_{-\eps}^\eps \dd r \int_{0}^{2\eps} \dd z  \chi_{|r|<z}\chi_{2p_2>\sqrt{1-r^2}-\sqrt{1-z^2}} \ln \left( \frac{8}{4p_2(\sqrt{1-r^2}-p_2) -z^2+r^2} \right)^2.
\end{multline}
Substituting $z$ by $x=4p_2(\sqrt{1-r^2}-p_2) -z^2+r^2$ and using H\"older's inequality we obtain
\begin{multline}
I_9\leq \frac{C}{4}\sup_{0\leq p_2<\eps}\int_{-\eps}^\eps \dd r \int_{0}^{4p_2(\sqrt{1-r^2}-p_2)+r^2} \ln \left( \frac{8}{x} \right)^2 \frac{1}{\sqrt{4p_2(\sqrt{1-r^2}-p_2) +r^2-x}} \dd x \\
\leq  \frac{C}{4}\sup_{0\leq p_2<\eps}\int_{-\eps}^\eps \dd r \left(\int_{0}^{4p_2(\sqrt{1-r^2}-p_2)+r^2} \ln \left( \frac{8}{x} \right)^6 \dd x\right)^{1/3}\times \\
\left(\int_{0}^{4p_2(\sqrt{1-r^2}-p_2)+r^2} \frac{1}{(4p_2(\sqrt{1-r^2}-p_2) +r^2-x)^{3/4}} \dd x\right)^{2/3} \\
\leq \frac{C}{2} \eps\left(\int_{0}^{4\eps+\eps^2}  \ln \left( \frac{8}{x} \right)^6 \dd x\right)^{1/3}\left(\int_{0}^{4\eps+\eps^2} \frac{1}{y^{3/4}} \dd y\right)^{2/3},
\end{multline}
which vanishes for $\eps \to0$.
\end{proof}

\subsection{Proof of Lemma~\ref{B-Q_gen}}\label{sec:pf_B-Q}
\begin{proof}[Proof of Lemma~\ref{B-Q_gen}.]
To prove Lemma~\ref{B-Q_gen} we show that the following expressions are finite.
\begin{enumerate}[(i)]
\item $\sup_{T>\mu/2} \sup_{ q \in \BR^d} \lVert V^{1/2} B_{T}(\cdot , q)  |V |^{1/2}\rVert $  \label{B-Q_0_g}
\item $\sup_T \sup_{ q \in \BR^d} \lVert V^{1/2} B_{T}(\cdot , q) \chi_{|\cdot|^2>3\mu}  |V |^{1/2}\rVert $  \label{B-Q_1_g}
\item $\sup_T \sup_{q\in \BR^d} \lVert V^{1/2} B_{T}(\cdot, q) \chi_{((\cdot+q)^2-\mu)((\cdot-q)^2-\mu)<0}  |V |^{1/2}\rVert $  \label{B-Q_3_g}
\item $\sup_T \sup_{\vert q \vert >\frac{\sqrt{\mu}}{2}} \lVert V^{1/2} B_{T}(\cdot, q) \chi_{p^2<3\mu} \chi_{((\cdot+q)^2-\mu)((\cdot-q)^2-\mu)>0} |V |^{1/2}\rVert $  \label{B-Q_2_g}
\item $\sup_T \sup_{\vert q \vert <\frac{\sqrt{\mu}}{2}} \left\lVert V^{1/2} \left[B_{T}(\cdot, q) \chi_{|\cdot|^2<3\mu} \chi_{((\cdot+q)^2-\mu)((\cdot-q)^2-\mu)>0}-Q_{T}(q) \right]  |V |^{1/2}\right\rVert $ \label{B-Q_4_g}
\end{enumerate}
In combination, they prove \eqref{B-Q_eq}.

We start with \eqref{B-Q_0_g} and \eqref{B-Q_1_g}.
By Lemma~\ref{KT-Laplace} there is a constant $c_0$ depending only on $\mu$, such that $B_T(p,q)\leq c_0/(1+p^2)$ for all $T>\mu/2$ and $p,q\in \BR^d$.
Similarly, using \eqref{BT_bound} one sees that there is a constant $c_1$ depending only on $\mu$ such that $B_T(p,q)\leq c_1/(1+p^2)$ for all $T>0$ and $p,q\in \BR^d$ with $q^2>3\mu$.
The claim follows since $\lVert |V|^{1/2}\frac{1}{1-\Delta} |V |^{1/2}\rVert$ is bounded \cite{lieb_analysis_2001,henheik_universality_2023,hainzl_bardeencooperschrieffer_2016}.

For \eqref{B-Q_3_g}, it suffices to prove that
\begin{equation}
Y= \sup_{T} \sup_{q\in \BR^d} \int_{\BR^d} B_{T}(p, q)\chi_{((p+q)^2-\mu)((p-q)^2-\mu)<0} \dd p <\infty
\end{equation}
since \eqref{B-Q_3_g} is bounded by $\lVert V \rVert_1 Y$.
The integrand is invariant under rotation of $(p,q)\to (R p, Rq)$ around the origin.
Hence, the integral only depends on the absolute value of $q$ and we may take the supremum over $q$ of the form $q=(\vert q \vert,0)$ only.
For $p,(q_1,0)$ satisfying $((p+(q_1,0))^2-\mu)((p-( q_1,0))^2-\mu)<0$, we can estimate by \cite[Lemma 4.7]{hainzl_boundary_nodate}
\begin{equation}
B_{T}(p,(q_1,0))\leq \frac{2}{T} \exp\left(-\frac{1}{T}\min\{(\vert p_1\vert + \vert q_1 \vert)^2+\tilde p^2-\mu, \mu-(\vert p_1\vert - \vert q_1 \vert)^2-\tilde p^2 \}\right)
\end{equation}
Note that $(\vert p_1\vert + \vert q_1 \vert)^2+\tilde p^2-\mu<\mu-(\vert p_1\vert - \vert q_1 \vert)^2-\tilde p^2 \leftrightarrow p^2+q_1^2<\mu$.
We can therefore further estimate
\begin{multline}
B_{T}(p,(q_1,0)) \chi_{(\vert p_1\vert + \vert q_1 \vert)^2+\tilde p^2>\mu>(\vert p_1\vert - \vert q_1 \vert)^2+\tilde p^2} \\
\leq \frac{2}{T} \exp\left(-\frac{1}{T}((\vert p_1\vert + \vert q_1 \vert)^2+\tilde p^2-\mu)\right)  \chi_{(\vert p_1\vert + \vert q_1 \vert)^2+\tilde p^2>\mu} \chi_{p^2+q_1^2<\mu}\\
+\frac{2}{T} \exp\left(-\frac{1}{T}(\mu-(\vert p_1\vert - \vert q_1 \vert)^2-\tilde p^2)\right) \chi_{\mu>(\vert p_1\vert - \vert q_1 \vert)^2+\tilde p^2}
\end{multline}
We now integrate the bound over $p$ and use the symmetry in $p_1$ to restrict to $p_1>0$, replace $\vert p_1\vert$ by $p_1$ and then extend the domain to $p_1\in \BR$.
We obtain
\begin{multline}
Y \leq \sup_{T} \sup_{q_1\in \BR} \frac{4}{T} \left[\int_{\BR^d} \exp\left(-\frac{1}{T}(( p_1 + \vert q_1 \vert)^2+\tilde p^2-\mu)\right)  \chi_{( p_1 + \vert q_1 \vert)^2+\tilde p^2>\mu} \chi_{p^2+q_1^2<\mu}\dd p \right.\\
+\left. \int_{\BR^d} \exp\left(-\frac{1}{T}(\mu-( p_1 - \vert q_1 \vert)^2-\tilde p^2)\right) \chi_{\mu>( p_1 - \vert q_1 \vert)^2+\tilde p^2} \dd p\right].
\end{multline}
Now we substitute $p_1 \pm \vert q_1 \vert $ by $p_1$ and obtain
\begin{multline}
Y \leq \sup_{T} \sup_{\vert q_1\vert<\sqrt{\mu}}\frac{4}{T} \int_{\BR^d} \exp\left(-\frac{1}{T}(p_1^2+\tilde p^2-\mu)\right)  \chi_{p_1^2+\tilde p^2>\mu} \chi_{(p_1-\vert q_1 \vert)^2+\tilde p^2+q_1^2<\mu}\dd p\\
+\sup_{T}\frac{4}{T} \int_{\BR^d} \exp\left(-\frac{1}{T}(\mu- p_1^2-\tilde p^2)\right) \chi_{\mu>p_1^2+\tilde p^2} \dd p\\
\leq  \sup_{T} \frac{4 \vert \BS^{d-1}\vert (2\sqrt{\mu})^{d-1}e^{\mu/T}}{T} \int_{\sqrt{\mu}}^{\infty} e^{-r^2/T} \dd r+\sup_{T}\frac{4 \vert \BS^{d-1}\vert \sqrt{\mu}^{d-1}e^{-\mu/T}}{T} \int_{0}^{\sqrt{\mu}}  e^{r^2/T} \dd r,
\end{multline}
where we used that $(p_1-\vert q_1 \vert)^2+\tilde p^2+q_1^2<\mu\Rightarrow p^2<2\mu$.
Note that
\begin{equation}
 \frac{  \sqrt{\mu}e^{\mu/T}}{T} \int_{\sqrt{\mu}}^{\infty} e^{-r^2/T} \dd r= \frac{\pi^{1/2}}{2}  \sqrt{\frac{\mu}{T}}e^{\mu/T} \text{erfc}\left(\sqrt{\frac{\mu}{T}}\right)
\end{equation}
and
\begin{equation}
\frac{ \sqrt{\mu}e^{-\mu/T}}{T} \int_{0}^{\sqrt{\mu}}  e^{r^2/T} \dd r=  \frac{\pi^{1/2}}{2}  \sqrt{\frac{\mu}{T}}e^{-\mu/T} \text{erfi}\left(\sqrt{\frac{\mu}{T}}\right)
\end{equation}
As in the proof of \cite[Lemma 4.4]{hainzl_boundary_nodate}, we conclude that $Y<\infty$ since the functions $xe^{x^2} \text{erfc}(x)$ and $xe^{-x^2} \text{erfi}(x)$ are bounded for $x\geq0$.

For \eqref{B-Q_2_g}, it again suffices to prove that
\begin{equation}
X=\sup_T \sup_{\vert q \vert> \frac{\sqrt{\mu}}{2}}\int_{\BR^d} B_{T}(p,q) \chi_{p^2<3\mu} \chi_{((p+q)^2-\mu)((p-q)^2-\mu)>0} \dd p <\infty,
\end{equation}
since \eqref{B-Q_2_g} is bounded by $\lVert V \rVert_1 X $.
Again we can restrict to $q$ of the form $q=(\vert q \vert,0)$.
The idea is to split the integrand in $X$ into four terms localized in different regions.
The integrand is supported on the intersection and the complement of the two disks/balls with radius $\sqrt{\mu}$ centered at $(\pm q_1,0)$. (For $d=2$ this is the white region in Figure~\ref{B_q_aspt_domains}).
\begin{itemize}
\item The first term covers the domain with $\vert \tilde p \vert > \sqrt{\mu}$ outside the disks/balls:\\
 $X_1=\sup_T \sup_{\vert q_1 \vert> \frac{\sqrt{\mu}}{2}}\int_{\BR^d} B_{T}(p,(q_1,0)) \chi_{p^2<3\mu} \chi_{ \tilde p^2 >\mu} \dd p $
\item The second term covers the remaining domain with $\vert p_1 \vert > \vert q_1 \vert $ outside  the two disks/balls:\\
$X_2=\sup_T \sup_{\vert q_1 \vert>  \frac{\sqrt{\mu}}{2}} \int_{\tilde p^2 <\mu} \dd \tilde p  \int_{\vert p_1 \vert>\sqrt{\mu-\tilde p^2}+\vert q_1 \vert} \dd p_1 B_{T}(p,(q_1,0)) \chi_{p^2<3\mu} $
\item The third term covers the remaining domain with $\vert p_1 \vert < \vert q_1 \vert $ outside  the two disks/balls:\\
$X_3=\sup_T \sup_{\vert q_1 \vert>  \frac{\sqrt{\mu}}{2}} \int_{\mu-q_1^2<\tilde p^2 <\mu} \dd \tilde p \int_{\vert p_1 \vert<-\sqrt{\mu-\tilde p^2}+\vert q_1 \vert} \dd p_1 B_{T}(p,(q_1,0)) \chi_{p^2<3\mu} $
\item The fourth term covers the domain in the intersection of the two disks/balls:\\
$X_4=\sup_T \sup_{\vert q_1 \vert>  \frac{\sqrt{\mu}}{2}} \int_{\tilde p^2<\mu-q_1^2} \dd \tilde p \int_{\vert p_1 \vert<\sqrt{\mu-\tilde p^2}-\vert q_1 \vert} \dd p_1 B_{T}(p,(q_1,0)) \chi_{p^2<3\mu} $
\end{itemize}
We prove that each $X_j$ is finite.
It then follows that $X \leq X_1+X_2+X_3+X_4 $ is finite.
We use the bounds
\begin{equation}
B_{T}(p,(q_1,0))\leq \left\{ \begin{matrix} \frac{1}{p^2+q_1^2-\mu} & \text{if}\ (\vert p_1\vert - \vert q_1 \vert)^2+\tilde p^2 >\mu, \\  \frac{1}{\mu- p^2-q_1^2} & \text{if}\ (\vert p_1\vert +\vert q_1 \vert)^2+\tilde p^2<\mu,  \end{matrix}\right.
\end{equation}
which follow from \eqref{pf:3d_2_lea3.4}.
The first line applies to $X_1,X_2, X_3$, the second line to $X_4$.
For $X_1$, we have $p^2+q_1^2-\mu>q_1^2>\mu/4$ and thus $X_1<\infty$.
Similarly, for $X_2$, we have $p^2+q_1^2-\mu= (\sqrt{ q_1^2+p_1^2} +\sqrt{\mu-\tilde p^2})(\sqrt{ q_1^2+p_1^2} -\sqrt{\mu-\tilde p^2}) \geq \vert q_1\vert(\vert p_1\vert -\sqrt{\mu-\tilde p^2})\geq q_1^2\geq \mu/4 $ and thus $X_2<\infty $.
For $X_3$, we have $p^2+q_1^2-\mu\geq \vert q_1\vert(\vert q_1\vert -\sqrt{\mu-\tilde p^2})\geq \frac{\sqrt{\mu}}{2}(\vert q_1\vert -\sqrt{\mu-\tilde p^2}) $.
Hence, $X_3 \leq \sup_{\vert q_1 \vert>  \frac{\sqrt{\mu}}{2}}\frac{4}{\sqrt{\mu}} \int_{\mu-q_1^2<\tilde p^2 <\mu} \dd \tilde p <\infty$.
For $X_4$ we have $\mu-p^2-q_1^2\geq \mu-(\sqrt{\mu-\tilde p^2}-\vert q_1 \vert)^2- \tilde p^2-q_1^2 =2 \vert q_1 \vert (\sqrt{\mu-\tilde p^2}-|q_1|)$.
Thus,
\begin{equation}
X_4\leq \sup_{\vert q_1 \vert>  \frac{\sqrt{\mu}}{2}}  \int_{\tilde p^2 <\mu-q_1^2} \frac{1}{|q_1|} \dd \tilde p<\infty.
\end{equation}

To prove that \eqref{B-Q_4_g} is finite, let $S_{T,d}(q):L^1(\BR^d)\to L^\infty(\BR^d)$ be the operator with integral kernel
\begin{equation}
S_{T,d}(q)(x,y)=\frac{1}{(2\pi)^d} \int_{\BR^d} \left[e^{i (x-y)\cdot p } -e^{i \sqrt{\mu} (x-y)\cdot p/\vert p \vert  } \right]B_{T}(p,q) \chi_{((p+q)^2-\mu)((p-q)^2-\mu)>0} \chi_{p^2<3\mu}  \dd p
\end{equation}
Then \eqref{B-Q_4_g} equals $\sup_T \sup_{\vert q \vert <\frac{\sqrt{\mu}}{2}} \left\lVert V^{1/2} S_{T,d}(q) |V |^{1/2}\right\rVert $.
With \eqref{BT_bound} and $\vert e^{i x}- e^{i y}\vert \leq \min\left\{\vert x-y\vert,2\right\}$ we obtain
\begin{multline}
\vert  S_{T,d}(q)(x,y) \vert \leq \frac{1}{(2\pi)^d} \int_{\BR^d} \frac{\min\left\{\vert  (\vert p \vert-\sqrt{\mu}) (x-y) \cdot p/\vert p \vert\vert,2\right\}}{\vert p^2+q^2-\mu \vert} \chi_{((p+q)^2-\mu)((p-q)^2-\mu)>0} \chi_{p^2<3\mu}  \dd p\\
\leq \frac{1}{(2\pi)^d} \int_{\BR^d} \frac{\min\left\{\vert \vert p\vert -\sqrt{\mu}\vert \vert x-y \vert,2\right\}}{\vert p^2+q^2-\mu \vert} \chi_{((p+q)^2-\mu)((p-q)^2-\mu)>0} \chi_{p^2<3\mu}  \dd p
\end{multline}
Again, the integral only depends on $\vert q \vert$, so we may restrict to $q=(\vert q \vert,0)$.
We now switch to angular coordinates.
Recall the notation $r_\pm$ and $e_\varphi$ introduced before \eqref{Q-est-1} and that $(\vert r \cos \varphi \vert \mp \vert q_1 \vert)^2+r^2 \sin \varphi^2\gtrless \mu \leftrightarrow r\gtrless r_\pm(e_\varphi)$.
For $d=2$ we have
\begin{multline}
\vert   S_{T,2}((q_1,0))(x,y) \vert
\leq \frac{1}{(2\pi)^2} \int_0^{2\pi} \Bigg[\int_{r_+(e_\varphi)}^{\sqrt{3\mu}}  \frac{\min\left\{ \vert r-\sqrt{\mu} \vert \vert x-y \vert,2\right\}}{r^2+q_1^2-\mu}r\dd r\\
+ \int_{0}^{r_-(e_\varphi)} \frac{\min\left\{(\sqrt{\mu}-r) \vert x-y \vert,2\right\}}{\mu-r^2-q_1^2} r\dd r\Bigg] \dd \varphi =: g(x,y,q_1)
\end{multline}
and for $d=3$
\begin{multline}
\vert   S_{T,3}((q_1,0))(x,y) \vert
\leq \frac{1}{(2\pi)^2} \int_0^{\pi} \Bigg[\int_{r_+(e_\theta)}^{\sqrt{3\mu}}  \frac{\min\left\{ \vert r-\sqrt{\mu} \vert \vert x-y \vert,2\right\}}{r^2+q_1^2-\mu}\sin \theta r^2\dd r\\
+ \int_{0}^{r_-(e_\theta)} \frac{\min\left\{(\sqrt{\mu}-r) \vert x-y \vert,2\right\}}{\mu-r^2-q_1^2} \sin \theta r^2\dd r\Bigg] \dd \theta
\leq \frac{\sqrt{3\mu}}{2}g(x,y,q_1).
\end{multline}
We bound $g$ by
\begin{multline}
\vert  g(x,y,q_1) \vert \\
\leq  \frac{1}{(2\pi)^2} \int_0^{2\pi}  \left[\int_{r_+(e_\varphi)}^{\sqrt{3\mu}} \frac{\min\left\{ (r-r_+(e_\varphi) ) \vert x-y \vert,2\right\}+\min\left\{\vert\sqrt{\mu}- r_+(e_\varphi)\vert \vert x-y \vert,2\right\}}{r^2+q_1^2-\mu} r \dd r \right.\\
\left. + \int_{0}^{r_-(e_\varphi)} \frac{\min\left\{(r_-(e_\varphi)-r) \vert x-y \vert,2\right\} +\min\left\{(\sqrt{\mu}-r_-(e_\varphi))\vert x-y \vert,2\right\}}{\mu-r^2-q_1^2} r \dd r \right] \dd \varphi
\end{multline}
Note that $r_+(e_\varphi)$ attains the minimal value $\sqrt{\mu-q_1^2}$ at $\vert \varphi\vert =\frac{\pi}{2}$ and the maximal value $\sqrt{\mu}+\vert q_1 \vert$ at $\vert \varphi\vert =0$.
Similarly, $r_-(e_\varphi)$ attains the maximal value $\sqrt{\mu-q_1^2}$ at $\vert \varphi\vert =\frac{\pi}{2}$ and the minimal value $\sqrt{\mu}-\vert q_1 \vert$ at $\vert \varphi\vert =0$.
For the first summand in both integrals we take the supremum over the angular variable.
For the second summand in both integrals, we carry out the integration over $r$ and use that $\vert \sqrt{\mu}-r_-(e_\varphi)\vert, \vert\sqrt{\mu}- r_+(e_\varphi)\vert \leq \vert q_1 \vert$.
We obtain the bound
\begin{multline}
\vert g(x,y,q_1) \vert \leq  \frac{1}{2\pi}  \int_{0}^{\sqrt{3\mu}} \frac{\min\left\{ \vert r-\sqrt{\mu-q_1^2} \vert \vert x-y \vert,2\right\} r}{\vert r^2+q_1^2-\mu \vert }\dd r \\
+\frac{\min\left\{\vert q_1 \vert \vert x-y \vert,2\right\} }{2(2\pi)^2} \int_0^{2\pi}  \left[ \ln\left(\frac{2\mu+q_1^2}{r_+(e_\varphi)^2+q_1^2-\mu}\right) +\ln\left(\frac{\mu-q_1^2}{\mu-q_1^2 -r_-(e_\varphi)^2}\right)  \right] \dd \varphi
\end{multline}
Recall that we are only interested in $\vert q_1 \vert<\sqrt{\mu}/2$.
For the first term, we use that $r\leq \sqrt{3\mu}$ and $\vert r^2+q_1^2-\mu \vert  =  \vert r-\sqrt{\mu-q_1^2} \vert  \vert r+\sqrt{\mu-q_1^2} \vert \geq  \vert r-\sqrt{\mu-q_1^2} \vert \sqrt{\mu-q_1^2}$.
This gives the following bound, where we first carry out the $r$-integration and then use that $\sqrt{\mu-q_1^2}\geq \sqrt{3\mu}/2$:
\begin{multline}
 \frac{\sqrt{3\mu}}{\pi\sqrt{\mu-q_1^2}}  \int_{0}^{\sqrt{3\mu}}\min\left\{ \frac{\vert x-y \vert}{2}, \frac{1}{\vert r-\sqrt{\mu-q_1^2}\vert}\right\}\dd r\\
\leq\frac{\sqrt{3\mu}}{\pi \sqrt{\mu-q_1^2}} \left[\ln \left(\max\left\{1, \frac{\sqrt{\mu-q_1^2}\vert x-y\vert}{2}\right\}\right)+2+\ln\left(\max\left\{1, \frac{(\sqrt{3\mu}-\sqrt{\mu-q_1^2})\vert x-y \vert}{2}\right\}\right)\right]\\
\leq C \left[1+\ln\left(1+ \frac{\sqrt{3\mu}\vert x-y \vert}{2}\right)\right]
\end{multline}
For the second term, we use that
\begin{equation}
\frac{2\mu+q_1^2}{r_+(e_\varphi)^2+q_1^2-\mu}\frac{\mu-q_1^2}{\mu-q_1^2 -r_-(e_\varphi)^2}= \frac{2\mu+q_1^2}{4 \vert e_{\varphi,1} \vert^2 \vert q_1 \vert^2}
\end{equation}
and $\vert q_1 \vert <\sqrt{\mu}/2$ as well as $\vert e_{\varphi,1}\vert=\vert\cos \varphi\vert \geq \frac{1}{2}\min\{\vert \frac{\pi}{2}-\varphi \vert,\vert \frac{3\pi}{2}-\varphi \vert\}$ to arrive at the bound
\begin{multline}
\frac{\min\left\{\vert q_1 \vert \vert x-y \vert,2\right\}}{(2\pi)^2} \int_0^{2\pi}  \ln\left(\frac{\sqrt{3\mu}}{2\vert e_{\varphi,1} q_1\vert }\right) \dd \varphi
\leq \frac{4 \min\left\{\vert q_1 \vert \vert x-y \vert,2\right\}}{(2\pi)^2} \int_0^{\pi/2}  \ln\left(\frac{\sqrt{3\mu}}{\vert \varphi q_1\vert }\right) \dd \varphi\\
 =\frac{\min\left\{\vert q_1 \vert \vert x-y \vert,2\right\}}{2\pi} \left(1+\ln\left(\frac{2\sqrt{3\mu}}{\pi \vert q_1\vert }\right) \right)\\
=\frac{\min\left\{\vert q_1 \vert \vert x-y \vert,2\right\}}{2\pi} \left(1+\ln\left(\sqrt{3\mu}\vert x-y \vert\right)+\ln\left(\frac{2\pi}{\vert x- y \vert\vert q_1\vert }\right)\right),
\end{multline}
where we used $\int \ln(1/x) \dd x =x+ x \ln(1/x)$.
Since $ x \ln(1/x) \leq C$, this is bounded above by
\begin{equation}
\frac{1}{\pi} \left(1+\max\left\{\ln\left(\sqrt{3\mu}\vert x-y \vert\right),0\right\}\right)+C.
\end{equation}
In total, we obtain the bound
\begin{equation}
\sup_{\vert q_1 \vert<\frac{\sqrt{\mu}}{2}} \vert g(x,y,q_1) \vert \leq C\left[1+\ln\left(1+ \sqrt{\mu}\vert x-y \vert\right)\right].
\end{equation}
Let $M:L^2(\BR^d)\to L^2(\BR^d)$ be the operator with integral kernel $M(x,y)=|V|^{1/2} (x) (1+\ln \left( 1+\sqrt{\mu}\vert x-y \vert\right)) |V|^{1/2}(y)$.
We have
\begin{equation}
\sup_T \sup_{\vert q \vert <\frac{\sqrt{\mu}}{2}} \left\lVert V^{1/2} S_{T,d}(q) |V|^{1/2}\right\rVert \leq C(\mu,d) \lVert M \rVert
\end{equation}
for some constant $C(\mu,d)<\infty$.
The operator $M$ is Hilbert-Schmidt since the function $x\mapsto (1+\ln(1+\vert x\vert)^2)|V(x)|$ is in $L^1(\BR^d).$
\end{proof}

%

\bibliographystyle{abbrv}

\end{document}